\documentclass[10pt,a4paper]{article}

\usepackage{amstext,amssymb,amsfonts,latexsym}
\usepackage{theorem}
\usepackage{pifont}
\usepackage{amsmath}

 \newcommand{\bs}{\bigskip}
 \newcommand{\ms}{\medskip}
 \newcommand{\n}{\noindent}
 \newcommand{\s}{\smallskip}
 \newcommand{\hs}[1]{\hspace*{ #1 mm}}
 \newcommand{\vs}[1]{\vspace*{ #1 mm}}

 \newcommand{\real}{\mathbb{R}}

 \newcommand{\nat}{\mathbb{N}}
 
 \newcommand{\integer}{\mathbb{Z}}
 \newcommand{\rational}{\mathbb{Q}}
 
 \newcommand{\complex}{\mathbb{C}}
 \newcommand{\appcomplex}{\tilde{\mathbb{C}}}
 \newcommand{\ptcomplex}{\tilde{\mathbb{C}}}

 \newcommand{\bm}[1]{\mbox{\boldmath $ #1 $}}

 \newcommand{\FF}{{\cal F}}
 
 \newcommand{\HH}{{\cal H}}

 \newcommand{\GG}{{\cal G}}

 \newcommand{\p}{\mathrm{P}}
 \newcommand{\np}{\mathrm{NP}}

 \newcommand{\bqp}{\mathrm{BQP}}

 \newcommand{\squareqp}{\Box^{\mathrm{QP}}_{1}}
 \newcommand{\hatsquareqp}{\widehat{\Box^{\mathrm{QP}}_{1}}}

 \newcommand{\fbqp}{\mathrm{FBQP}}

 \newcommand{\comb}[2]{\left(\:\begin{subarray}{c} #1 \\%
      #2 \end{subarray}\right)}

 \def\bbox{\vrule height6pt width6pt depth1pt}

\theoremstyle{plain}
\theoremheaderfont{\bfseries}
 \newtheorem{theorem}{Theorem}[section]
 \newtheorem{lemma}[theorem]{Lemma}
 \newtheorem{proposition}[theorem]{Proposition}

 {\theorembodyfont{\rmfamily}
\newtheorem{definition}[theorem]{Definition}}
 {\theorembodyfont{\rmfamily} }
 {\theorembodyfont{\rmfamily} }

 \newenvironment{proof}{\par \noindent
            {\bf Proof. \hs{2}}}{\hfill$\Box$ \vspace*{3mm}}

 \newenvironment{proofof}[1]{\vspace*{5mm} \par \noindent
         {\bf Proof of #1.\hs{2}}}{\hfill$\Box$ \vspace*{3mm}}

 \newcommand{\floors}[1]{\lfloor #1 \rfloor}
 \newcommand{\pair}[1]{\langle #1 \rangle}

 \newcommand{\qubit}[1]{| #1 \rangle}
 \newcommand{\bra}[1]{\langle #1 |}
 \newcommand{\ket}[1]{| #1 \rangle}
 \newcommand{\measure}[2]{\langle #1 | #2 \rangle}
 \newcommand{\density}[2]{| #1 \rangle \! \langle #2 |}
 \newcommand{\qustrings}{\Phi}
 \newcommand{\tracenorm}[1]{\| #1 \|_{\mathrm{tr}}}

 \newcommand{\ignore}[1]{}

\begin{document}
\setcounter{page}{1}


\begin{center}
{\Large {\bf A Schematic Definition of Quantum Polynomial
Time Computability}}\footnote{This work was done while the author was at the University of Ottawa between 1999 and 2003, and it was financially supported by the Natural Sciences and Engineering Research Council of Canada.}\footnote{An extended abstract appeared under the title ``A recursive definition of quantum polynomial time computability (extended abstract)'' in the Proceedings of the 9th Workshop on Non-Classical Models of Automata and Applications (NCMA 2017), Prague, Czech Republic, August 17--18, 2017, \"{O}sterreichische Computer Gesellschaft 2017, the Austrian Computer Society, pp.243--258, 2017. The current paper intends to correct erroneous descriptions in the extended abstract and provide more detailed explanations to all omitted proofs due to the page limit.} \bs\\

{\sc Tomoyuki Yamakami}\footnote{Current Affiliation: Faculty of Engineering, University of Fukui,  3-9-1 Bunkyo, Fukui, 910-8507 Japan} \ms\\
\end{center}
\bs

\begin{abstract}
In the past four decades, the notion of quantum polynomial-time computability has been mathematically modeled by quantum Turing machines as well as quantum circuits. This paper seeks the third model, which is a quantum analogue of the schematic (inductive or constructive) definition of (primitive) recursive functions. For quantum functions mapping finite-dimensional Hilbert spaces to themselves, we present such a schematic definition, composed of a small set of initial quantum
functions and a few construction rules that dictate how to build a new quantum  function
from the existing ones.
We prove that our schematic definition precisely characterizes all
functions that can be computable with high success probabilities on well-formed quantum Turing
machines in polynomial time, or equivalently uniform families of polynomial-size quantum circuits. Our new, schematic definition is quite simple and intuitive and, more importantly, it avoids the cumbersome  introduction of the
well-formedness condition imposed on a quantum Turing machine model as
well as of the uniformity condition necessary for a quantum circuit
model. Our new approach can further open a door to the descriptional complexity of quantum functions, to the theory of higher-type quantum functionals, to the development of new first-order theories for quantum computing, and to the designing of programming languages for real-life quantum computers.

\s
\n{\sf Key words:} quantum computing, quantum function, quantum Turing machine, quantum circuit, schematic definition, descriptional complexity, polynomial-time computability, normal form theorem

\end{abstract}

\sloppy
\section{Background, Motivation, and the Main Results}\label{sec:introduction}

In early 1980s emerged a groundbreaking idea of exploiting quantum physics to build  mechanical computing devices, dubbed as \emph{quantum computers}, which have completely altered the way we used to envision ``computers.''
Subsequent discoveries of more efficient quantum computations for factoring positive integers \cite{Sho97} and searching unstructured databases \cite{Gro96,Gro97} than classical computations prompted us to look for more mathematical and practical problems that can be solvable effectively on the quantum computers.
Efficiency in quantum computing has since then rapidly become an important research subject of computer science as well as physics.

As a
mathematical model to realize quantum computation, Deutsch
\cite{Deu85} introduced a notion of \emph{quantum Turing machine} (or QTM, for short), which was later discussed by Yao \cite{Yao93} and further refined by Bernstein and Vazirani \cite{BV97}. This mechanical model
largely expands the classical model of (probabilistic) Turing machine by allowing a physical phenomenon, called
\emph{quantum interference}, to take place on its computation.
A different Hamiltonian formalism of Turing machines was also suggested by Benioff \cite{Ben80}. A QTM has an ability of computing a \emph{quantum function} mapping a finite-dimensional Hilbert
space to itself by evolving unitarily a superposition of (classical)
configurations of the machine, starting with a given input string and an initial inner state.
A more restrictive use of  the term of ``quantum function'' is found in, e.g., \cite{Yam03}, in which quantum functions take classical input strings and produce either classical output strings of QTMs or acceptance probabilities of QTMs. Throughout this paper, nevertheless, quantum functions refer only to functions acting on Hilbert spaces of arbitrary dimensions.

To ensure the unitary nature of quantum
computation, a QTM requires its mechanism to meet the so-called \emph{well-formedness condition} on a single-tape model of QTMs \cite{BV97} and a multi-tape model \cite{Yam99,Yam03} as well as \cite{ON00}. Refer to Section \ref{sec:def-QTMs} for their precise definitions.

Bernstein and Vazirani further formulated a new complexity class, denoted by $\bqp$, as the collection of all languages recognized by well-formed QTMs running in polynomial time with
error probability bounded from above by $1/3$.
Furthermore, QTMs equipped with output tapes can compute  string-valued functions in place of languages, and those functions form a function class, called  $\fbqp$.

From a different viewpoint, Yao \cite{Yao93} expanded Deutsch's notion of  \emph{quantum
network} \cite{Deu89} and formalized a notion of \emph{quantum circuit}, which is a quantum analogue of classical Boolean circuit.
Different from the classical Boolean circuit model, a quantum circuit
is composed of \emph{quantum gates}, each of which represents a unitary transformation
acting on a Hilbert space of a small, fixed dimension.
To act as a ``programmable'' unitary operator, a family of quantum circuits requires the so-called \emph{uniformity condition}, which ensures that a blueprint of each quantum circuit is easily rendered.
Yao further demonstrated that a uniform family of quantum circuits
is powerful enough to simulate a well-formed quantum Turing machine. As Nishimura and Ozawa \cite{NO02} pointed out, the
uniformity condition of a quantum circuit family is necessary to precisely
capture quantum polynomial-time computation. With this uniformity condition, $\bqp$ and $\fbqp$ are characterized exactly by uniform families of quantum circuits made up of polynomially many quantum gates.

This current paper boldly takes the third approach toward the characterization of
quantum polynomial-time computability. Unlike the aforementioned mechanical device models, our approach is to extend the schematic (inductive or constructive) definition of \emph{(primitive) recursive functions} on natural numbers. Such a schematic definition was thought in the 19th century by Peano \cite{Pea91}, opposed to the definition given by Turing's  machine model \cite{Tur36}. This classical scheme comprises a small set of initial functions and a small set of rules,
which dictate how to construct a new function from the existing ones.
For instance, every primitive recursive function is built from the constant, successor, and projection functions by applying composition and primitive recursion finitely many times. In particular, the primitive recursion introduces a new function whose values are defined \emph{by induction}. Recursive functions (in the form of \emph{$\mu$-recursive functions} \cite{Kle36,Kle43}) further  require an additional scheme, known as the minimization (or the least number) operator. These functions coincide with the Herbrand-G\"{o}del formalism of \emph{general recursive functions} (see \cite{Dav65}).
For a historical account of these notions, refer to, e.g., \cite{Soa96}.
Similar schematic approaches to capture classical polynomial-time
computability have already been sought in the literature \cite{Cob64,Con73,CK89,Meh76,Yam95}. Those approaches have led to quite different research subjects from what the Turing machine model provides.

Our purpose in this paper is to give a schematic definition of quantum functions to capture the notion of quantum polynomial-time computability and, more importantly, to make such a definition simpler and more intuitive for a practical merit of our own. Our schematic definition (Definition \ref{def:initial}) includes a set of initial quantum functions,  $I$ (identity),
$NOT$ (negation of a qubit), $PHASE_{\theta}$
(phase shift by $e^{i \theta}$),
$ROT_{\theta}$ (rotation around $xy$-axis by angle $\theta$), $SWAP$ (swap between two qubits), and $MEAS$ (partial projective measurement), as well as construction rules, composed of  composition ($Compo[\cdot,\cdot]$), branching ($Branch[\cdot,\cdot]$),
and multi-qubit quantum recursion ($kQRec[\cdot,\cdot|\cdot]$). Our choice of these initial quantum functions and construction rules stems mostly from a universal set of quantum gates in use in the past literature.
Our quantum recursion, on the contrary, is quite different in nature from the primitive recursion used to build primitive recursive functions. Instead of using the successor function to count down the number of inductive iterations in the primitive recursion, the quantum recursion uses a divide-and-conquer strategy of reducing the number of accessible qubits needed for performing a specified quantum function.
Within our new framework, we can implement typical unitary operators, such as the Walsh-Hadamard transform (WH), the controlled-NOT (CNOT), and the global phase shift (GPS).

An immediate merit of our schematic definition is that we can avoid the cumbersome introduction of the well-formedness condition imposed on the QTM model and the uniformity condition on the quantum circuit
model. Another advantage of our schemata is that each scheme has its own inverse; namely, for any quantum function $g$ defined by one of the schemata, its inverse $g^{-1}$ is also defined by the same kind of scheme. For instance, the inverses of the quantum functions $ROT_{\theta}$ and $kQRec_t[g,h,p|\{f_s\}_{s\in\{0,1\}^k}]$ introduced in Definition \ref{def:initial} are exactly $ROT_{-\theta}$ and $kQRec_t[g^{-1},p^{-1},h^{-1}|\{f_s^{-1}\}_{s\in\{0,1\}^k}]$, respectively (Proposition \ref{inverse}).

For a further explanation of our main contributions, it is time to introduce a succinct notation of $\squareqp$ (where $\Box$ is pronounced ``square'') to denote the set of all quantum functions built from the initial quantum functions and by a finite series of sequential applications of the construction rules.
Since the partial measurement ($MEAS$) is not a unitary operator, we denote the class obtained from $\squareqp$ without use of $MEAS$ by $\hatsquareqp$.
Briefly, let us discuss clear differences between our schematic definition and the aforementioned two formalisms of polynomial-time quantumly computable functions in terms of QTMs and quantum circuits. Two major differences are listed below.

1) While a single quantum circuit takes a fixed number of input qubits, our quantum function takes an ``arbitrary'' number of qubits as an input.
This situation is similar to QTMs because a QTM has an infinite tape and  uses an arbitrary number of tape cells during its computation as extra storage space.
On the contrary to the QTMs, a $\hatsquareqp$-function is constructed using the same number of qubits as its original input in such a way that a quantum circuit has the same number of input qubits and output qubits.

2) The two machine models exhort an algorithmic description to dictate the behavior of each machine; more specifically, a QTM uses a transition function, which algorithmically describes how each step of the machine acts on certain qubits, and a  family of quantum circuits uses its uniformity condition to render the design of quantum gates in each quantum circuit. Unlike these two models, no $\squareqp$-function has any mechanism to store information on the description of the function itself but the construction process itself specifies the behavior of the function.

As a consequence, the above mentioned differences help the $\squareqp$-functions take a distinctive  position among all possible computational models that characterize quantum polynomial-time computability, and therefore we expect them to play an important role in analyzing the features of quantum polynomial-time computation from a quite different perspective.

In Section \ref{sec:definition}, we will formally present our schematic definition of
$\squareqp$- functions (as well as $\hatsquareqp$-functions) and show in Section \ref{sec:main-theorem} that $\squareqp$ (also $\hatsquareqp$) can  characterize all  functions in $\fbqp$. More precisely,  we assert in the main theorem  (Theorem \ref{theorem:character}) that any function from $\{0,1\}^*$ to $\{0,1\}^*$ in $\fbqp$ can be characterized by a certain polynomial $p$ and a certain quantum function $g\in \squareqp$ in such a way that, by using an appropriate coding scheme, in the final quantum state of $g$ on instances $x$ and the runtime bound $p(|x|)$, we observe an output value $f(x)$ with high probability.
This theorem will be split into two lemmas, Lemmas \ref{lemma:first-part} and \ref{lemma:second-part}. The former lemma will be proven in Section \ref{sec:main-theorem}; however, the proof of the latter lemma is so lengthy that it will be postponed until Section \ref{sec:second-part}. In this proof, we will construct a $\squareqp$-function that can simulate the behavior of a given QTM.

Notice that, since $\bqp$ is a special case of $\fbqp$, $\bqp$ is also characterized by
our model. In our proof of the characterization theorem (Theorem \ref{theorem:character}), we will utilize a main result of Bernstein and Vazirani
\cite{BV97} and that of Yao \cite{Yao93} extensively.
In Section \ref{sec:normal-form}, we will apply our
characterization, in the help of a universal QTM \cite{BV97,NO02}, to obtain a quantum version of Kleene's \emph{normal form theorem} \cite{Kle36,Kle43}, in which there is a universal pair of primitive recursive predicate and function that can describe the behavior of every recursive function.

Unlike classical computation on natural numbers (equivalently, \emph{strings}  over finite alphabets by appropriate coding schemata), quantum computation is a series of    certain manipulations of a single vector in a finite-dimensional Hilbert space and we need only high precision to approximate each function in $\fbqp$ by such a vector. This fact allows us to choose a different set of schemata (initial quantum functions and construction rules) to capture the essence of quantum computation.
In Section \ref{sec:open-question}, we will discuss this issue using an example of a general form of the \emph{quantum Fourier transform} (QFT). This transform may not be ``exactly'' computed in our current framework of $\squareqp$ but we can easily expand $\squareqp$ to compute the generalized QFT exactly if we include an additional initial quantum function, such as $CROT$ (controlled rotation).

Concerning future research on the current subject, we will discuss in Section \ref{sec:applications} four new directions of the subject.
Our schematic definition provides not only a different way of describing languages and functions computable quantumly in polynomial time but also a simple way of measuring the ``descriptional'' complexity of a given language and a  function restricted to instances of specified length.
This new complexity measure will be useful to prove basic properties of $\hatsquareqp$-functions in Section \ref{sec:recursive-def}. Its future application  will be briefly discussed in Section \ref{sec:descriptional-complexity}.

Kleene \cite{Kle59,Kle63} defined recursive functionals of higher types by extending the aforementioned recursive functions on natural numbers.
A more general study of \emph{higher-type functionals} has been conducted in computational complexity theory for decades
\cite{Con73,CK89,Meh76,Tow90,Yam95}.
In a similar spirit, our schematic definition
enables us to study \emph{higher-type quantum functionals}.
In Section \ref{sec:type-2-functional}, using \emph{oracle functions} (\emph{function oracles} or \emph{oracles}), we will define type-2 quantum functionals, which may guide us to a rich field of research in the future.


A schematic definition of how to construct a target $\squareqp$-function can be viewed as a ``program'' that describes a series of instructions on which schemata to use. Hence, our schematic formulation opens a door to a new, practical application to the designing of succinct  programming languages to control the operations of real-life quantum computers.
In Section \ref{sec:programming-language}, we will briefly argue on an application of the schematic definition to the future development of ``quantum programming languages.''
As a further application of our schematic definition, we can look into a new aspect of first-order theories and their subtheories. Earlier, a quantum analogue of $\np$ (nondeterministic polynomial time class) and beyond were sought in \cite{Yam02} with the use of bounded quantifiers over quantum states in finite-dimensional Hilbert spaces.
In a similar vein, we expect that $\squareqp$ will serve as a foundation to the introduction of first-order theories and their subtheories over quantum states in Hilbert spaces.

\section{Fundamental Notions and Notation}\label{sec:notion-notation}

We begin with explaining basic notions and notation necessary to read through the subsequent sections. Let us assume the reader's familiarity with classical Turing machines (see, e.g., \cite{HU79}). For the foundation of quantum information and computation, in contrast, the reader refers to basic textbooks, e.g., \cite{KSV02,NC00}.

\subsection{Numbers, Languages, and Qustrings}

The notation $\integer$ indicates the set of all \emph{integers} and $\nat$ expresses  the set of all {\em natural numbers} (that is, non-negative
integers). For convenience, we set  $\nat^{+}=\nat-\{0\}$. Moreover, $\rational$ denotes the set of all {\em rational numbers} and $\real$ indicates the set of all {\em real numbers}. For two numbers $m,n\in\integer$ with $m\leq n$, the notation $[m,n]_{\integer}$ denotes an \emph{integer interval} $\{m,m+1,m+2,\ldots,n\}$, compared to a real interval $[\alpha,\beta]$ for $\alpha,\beta\in\real$ with $\alpha\leq \beta$. In particular, $[n]$ is shorthand for $[1,n]_{\integer}$ for any $n\in\nat^{+}$.
By $\complex$, we express the set of all {\em complex numbers}. Given $\alpha\in\complex$, $\alpha^*$ expresses the \emph{complex conjugate} of $\alpha$.
\emph{Polynomials} are assumed to have natural numbers as their coefficients and they thus produce nonnegative values from nonnegative inputs.
A real number $\alpha$ is called {\em
polynomial-time approximable}\footnote{Ko and Friedman \cite{KF82}
first introduced this notion under the name of ``polynomial-time
computable.'' To avoid reader's confusion in this paper,
we prefer to use the term
``polynomial-time approximation.''} if there exists a multi-tape
polynomial-time deterministic Turing machine $M$ (equipped with a write-only output tape) that, on each input of the form $1^n$ for a natural number $n$, produces a finite binary fraction, $M(1^n)$, on its designated output tape with $|M(1^n)-\alpha|\leq 2^{-n}$. Let $\appcomplex$ be the set of complex numbers whose real and
imaginary parts are both polynomial-time approximable. For a bit $a\in\{0,1\}$, $\overline{a}$ indicates $1-a$. Given a matrix $A$, $A^{T}$ denotes its \emph{transpose} and $A^{\dagger}$ denotes the transposed conjugate of $A$.

An \emph{alphabet} is a finite nonempty set of ``symbols'' or ``letters.''
Given such an alphabet $\Sigma$, a \emph{string} over $\Sigma$ is a finite series of symbols taken from $\Sigma$.
The \emph{concatenation} of two strings $u$ and $w$ is expressed as $u\cdot w$ or more simply $uw$. The \emph{length} of a string $x$, denoted by $|x|$,  is the number of all occurrences of symbols in $x$. In particular, the \emph{empty string} has length $0$ and is denoted $\lambda$. We write $\Sigma^n$ for the subset of $\Sigma^*$ consisting only of all strings of length $n$ and we set $\Sigma^*=\bigcup_{n\in\nat} \Sigma^n$ (the set of all strings over $\Sigma$).
A \emph{language} over $\Sigma$ is a subset of $\Sigma^*$.
Given a language $S$, its \emph{characteristic function} is also expressed by $S$; that is, $S(x)=1$ for all $x\in S$ and $S(x)=0$ for all $x\notin S$. A function on $\Sigma^*$ (i.e., from $\Sigma^*$ to $\Sigma^*$) is \emph{polynomially bounded} if there exists a polynomial $p$ satisfying  $|f(x)|\leq p(|x|)$ for all strings $x\in\Sigma^*$.

For each natural number $k\geq1$, $\HH_k$ expresses a \emph{Hilbert space} of
dimension $k$ and each element of $\HH_{k}$ is expressed as $\ket{\phi}$ using Dirac's ``ket'' notation. In this paper, we are interested only in the case where $k$ is a power of $2$ and we implicitly assume that $k$ is of the form $2^n$ for a certain $n\in\nat$. Any element of $\HH_2$ that has the unit norm is called a \emph{quantum bit} or a {\em qubit}.
By choosing a standard computational basis
$B_1=\{\qubit{0},\qubit{1}\}$, every qubit $\qubit{\phi}$ can be expressed as
$\alpha_0\qubit{0}+\alpha_1\qubit{1}$ for an appropriate choice of two values
$\alpha_0,\alpha_1\in\complex$ (called \emph{amplitudes}) satisfying $|\alpha_0|^2+|\alpha_1|^2=1$. We also express $\qubit{\phi}$ as a \emph{column vector} of the form $\comb{\alpha_0}{\alpha_1}$; in particular, $\qubit{0}=\comb{1}{0}$ and $\qubit{1}=\comb{0}{1}$.
In a more general case of $n\geq1$, we use $B_n=\{\qubit{s}\mid s\in\{0,1\}^n\}$ as a computational basis of
$\HH_{2^n}$ with $|B_n|=2^n$.
Given any number $n\in\nat^{+}$, a {\em qustring of length $n$} is a vector $\qubit{\phi}$ of
$\HH_{2^n}$ with unit norm; namely, it is of the form
$\sum_{s\in\{0,1\}^n}\alpha_s\qubit{s}$, where each amplitude $\alpha_s$ is in $\complex$
with $\sum_{s\in\{0,1\}^n} |\alpha_s|^2 =1$.
Notice that a qubit is a
qustring of length $1$. The exception is the \emph{null vector}, denoted simply by $\bm{0}$, which has norm $0$.
Although the null vector could be a qustring of ``arbitrary'' length $n$, we instead refer to it as the \emph{qustring of length} $0$ for convenience. We use the notation $\qustrings_n$  for each $n\in\nat$ to denote the collection of all qustrings of length $n$. Finally, we set  $\qustrings_{\infty}=\bigcup_{n\in\nat}\qustrings_{n}$ (the set of all qustrings).

When $s=s_1s_2\cdots s_n$ with $s_i\in\{0,1\}$ for any index $i\in[n]$, the qustring $\qubit{s}$ coincides with $\qubit{s_1}\otimes \qubit{s_2}\otimes\cdots \otimes \qubit{s_n}$, where $\otimes$ denotes the \emph{tensor product} and is expressed as, for example, $\qubit{00} = (1\;0\;0\;0)^T$, $\qubit{01}=(0\;1\;0\;0)^T$, and $\qubit{11}=(0\;0\;0\;1)^T$.
The \emph{transposed conjugate} of $\qubit{s}$ is denoted by $\bra{s}$ (with the ``bra'' notation). For instance, if $\qubit{\phi}=\alpha\qubit{01}+\beta\qubit{10}$, then $\bra{\phi} = \alpha^*\bra{01}+\beta^*\bra{10}$.
The \emph{inner product} of $\qubit{\phi}$ and $\qubit{\psi}$ is expressed as $\measure{\phi}{\psi}$ and the \emph{norm} of $\qubit{\phi}$ is thus $\sqrt{\measure{\phi}{\psi}}$.
When we \emph{observe} or \emph{measure}  $\qubit{\phi}$ in the computational basis $B_n$, we obtain each string $s\in\{0,1\}^n$ with probability $|\measure{s}{\phi}|^2$.

Let $\HH_{\infty}= \bigcup_{n\in\nat^+}\HH_{2^n}$. We extend the ``length'' notion to arbitrary quantum states in $\HH_{\infty}$. Given each non-null vector $\qubit{\phi}$ in $\HH_{\infty}$, the {\em length} of $\qubit{\phi}$, denoted by $\ell(\qubit{\phi})$, is the minimal number  $n\in\nat$ satisfying $\qubit{\phi}\in\HH_{2^n}$; in other words, $\ell(\qubit{\phi})$ is the logarithm of the dimension of the vector $\qubit{\phi}$. Conventionally, we further set $\ell(\bm{0})=\ell(\alpha)=0$ for the null vector $\bm{0}$ and any scalar    $\alpha\in\complex$.
By this convention, if $\ell(\qubit{\phi})=0$ for a quantum state  $\qubit{\phi}$, then $\qubit{\phi}$ must be the qustring of length $0$.
A qustring $\qubit{\phi}$ of length $n$ is called
{\em basic} if $\qubit{\phi}=\qubit{s}$ for a certain binary string $s$ and we often identify such a basic qustring $\qubit{s}$ with the corresponding classical binary string $s$ for convenience.
Since all basic qustrings in $\Phi_n$ form $B_n$, $\HH_{2^n}$ is spanned by all elements in $\Phi_n$.

The \emph{partial trace} over a system $B$ of a composite system $AB$, denoted by $tr_B$, is a quantum operator for which $tr_B(\density{\phi}{\phi})$ is a vector obtained by tracing out $B$ from the \emph{outer product} $\density{\phi}{\phi}$ of a quantum state $\qubit{\phi}$.
Regarding a quantum state $\qubit{\phi}$ of $n$ qubits, we use a handy notation $tr_{k}(\density{\phi}{\phi})$ to mean the quantum state obtained from $\qubit{\phi}$ by tracing out all qubits except for the first $k$ qubits.
For example, it follows for $\sigma_1,\sigma_2,\tau_1,\tau_2\in\{0,1\}$ that  $tr_1(\density{\sigma_1}{\sigma_2}\otimes \density{\tau_1}{\tau_2}) = \density{\sigma_1}{\sigma_2} \cdot tr(\density{\tau_1}{\tau_2})$, where $tr(B)$ denotes the trace of a matrix $B$.
The \emph{trace norm} $\tracenorm{A}$ of a square matrix $A$ is defined by $\tracenorm{A}=tr(\sqrt{AA^{\dagger}})$. The \emph{total variation distance} between two ensembles  $p=\{p_i\}_{i\in A}$ and $q=\{q_i\}_{i\in A}$ of real numbers over a finite index set $A$  is
$\frac{1}{2}\| p-q \|_{1} = \frac{1}{2} \sum_{i\in A}| |p_i| - |q_i| |$.

Throughout this paper, we take special conventions concerning three notations, $\qubit{ \cdot }$, $\otimes$, and $\|\cdot\|$, which respectively express quantum states, the tensor product, and the $\ell_2$-norm. These conventions slightly deviate from the standard ones used in, e.g., \cite{NC00}, but they make our mathematical descriptions in later sections simpler and more succinct.

\s
\n{\bf Notational Conventions:}
We freely abbreviate
$\qubit{\phi}\otimes\qubit{\psi}$ as $\qubit{\phi}\qubit{\psi}$ for any two vectors $\qubit{\phi}$ and $\qubit{\psi}$.
Given two binary strings $s$ and $t$, $\qubit{st}$ means
$\qubit{s}\otimes\qubit{t}$ or $\qubit{s}\qubit{t}$.
Let $k$ and $n$ be two integers with $0<k<n$. Any qustring $\qubit{\phi}$ of length $n$
is expressed in general as $\qubit{\phi} = \sum_{s:|s|=k}
\qubit{s}\qubit{\phi_s}$, where each $\qubit{\phi_s}$ is a qustring of length $n-k$.
This qustring $\qubit{\phi_s}$ can be viewed as a consequence of applying a partial projective measurement to the first $k$ qubits of $\qubit{\phi}$, and therefore it is possible to express $\qubit{\phi_s}$ succinctly as  $\measure{s}{\phi}$.
With this new, convenient notation, $\qubit{\phi}$ coincides with   $\sum_{s:|s|=k}\qubit{s}\otimes \measure{s}{\phi}$, which is simplified as  $\sum_{s:|s|=k}\qubit{s}\!\measure{s}{\phi}$.
Notice that, when $k=n$, $\measure{s}{\phi}$ is a scalar, say, $\alpha$ in $\complex$. Hence, $\ket{s}\otimes \measure{s}{\phi}$ is $\qubit{s}\otimes \alpha$ and it is treated as a column vector $\alpha\qubit{s}$; similarly, we identify $\alpha\otimes \qubit{s}$ with $\alpha\qubit{s}$.
In these cases, $\otimes$ is treated merely as the \emph{scalar  multiplication}. As a consequence, the equality  $\qubit{\phi} = \sum_{s:|s|=k}\qubit{s}\! \measure{s}{\phi}$ holds even when $k=n$.
Concerning the null vector $\bm{0}$, we also  take
the following special treatment: for any
vector $\qubit{\phi}\in\HH_{\infty}$, (i) $0\otimes \qubit{\phi} = \qubit{\phi}\otimes 0 = \bm{0}$, (ii)
$\qubit{\phi}\otimes \bm{0} = \bm{0}\otimes\qubit{\phi} = \bm{0}$, and (iii) when $\qubit{\psi}$ is the null vector, $\measure{\phi}{\psi}=\measure{\psi}{\phi}=0$. Associated with those conventions on the partial projective measurement $\measure{\phi}{\psi}$, we also extend the use of the norm notation $\|\cdot\|$ to scalars.
When $\ell(\qubit{\phi})=\ell(\qubit{\psi})$, $\|\measure{\phi}{\psi}\|$ expresses the absolute value $|\measure{\phi}{\psi}|$; more generally, for any number $\alpha\in\complex$,  $\|\alpha\|$ means $|\alpha|$. With these extra conventions, when $\qubit{\phi}$ has the form $\sum_{s:|s|=k} \qubit{s}\! \measure{s}{\phi}$, the equation $\|\qubit{\phi}\|^2 = \sum_{s:|s|=k}\|\measure{s}{\phi}\|^2$ always holds for any index $k\in[n]$.

\subsection{Quantum Turing Machines}\label{sec:def-QTMs}

We assume the reader's fundamental knowledge on the notion of
quantum Turing machine (or QTM)
defined in \cite{BV97}. As was done in \cite{Yam99}, we allow a QTM to equip
multiple tapes and to move its multiple tape heads non-concurrently either to the
right or to the left, or to make the tape heads stay still. Such a QTM was  also discussed elsewhere (e.g., \cite{ON00}) and is known to
polynomially equivalent to the model proposed in \cite{BV97}.

To compute functions from $\Sigma^*$ to $\Sigma^*$ over a given alphabet $\Sigma$, we generally introduce QTMs as machines equipped with output tapes on which output strings are written in a certain specified way by the time the machines halt. By identifying languages with their characteristic functions, such QTMs can be seen as language acceptors as well.

Formally, a \emph{$k$-tape quantum Turing machine}
(referred to as \emph{$k$-tape QTM}), for $k\in\nat^{+}$, is a sextuple
$(Q,\Sigma, \Gamma_1\times\cdots\times\Gamma_k, \delta,q_0,Q_f)$,
where $Q$ is a finite set of inner states including the initial
state $q_0$ and a set $Q_f$ of final states with $Q_f\subseteq Q$, each $\Gamma_i$ is an
alphabet used for tape $i$ with a distinguished blank symbol $\#$ satisfying    $\Sigma\subseteq\Gamma_1$, and
$\delta$ is a quantum transition function {}from
$Q\times \tilde{\Gamma}^{(k)} \times Q\times \tilde{\Gamma}^{(k)} \times\{L,N,R\}^k$ to
$\complex$, where  $\tilde{\Gamma}^{(k)}$ stands for
$\Gamma_1\times\cdots\times\Gamma_k$.
For convenience, we identify
$L$, $N$, and $R$ with $-1$, $0$, and $+1$, respectively, and we set  $D=\{0,\pm1\}$.
For more information, refer to \cite{Yam99}.

All tape cells of each tape are indexed sequentially by integers. The cell indexed $0$ on each tape is called the \emph{start cell}.
At the beginning of the computation, $M$ is in inner state $q_0$, all the tapes except for the input tape are blank, and all tape heads are scanning the start  cells. A given input string $x_1x_2\cdots x_n$ is initially written on the input  tape in such a way that, for each index $i\in[n]$,  $x_i$ is in cell $i$ (not cell $i-1$).
When $M$ enters a final state, an \emph{output} of $M$ is the content of the string written on an output tape (if $M$ has only a single tape, then an output tape is the same as the tape used to hold inputs) from the start cell, stretching to the right until the first blank symbol.
A \emph{configuration} of $M$ is expressed as a triplet $(p,(h_i)_{i\in[k]},(z_i)_{i\in[k]})$, which indicates that $M$ is currently in inner state $p$ having $k$ tape heads at cells indexed by $h_1,\ldots,h_k$ with tape contents $z_1,\ldots,z_k$, respectively. The notion of configuration will be slightly modified in Sections \ref{sec:main-contributions}--\ref{sec:second-part} to make the proof of our main theorem simpler.
An \emph{initial configuration} is of the form $(q_0,0,x)$ and a \emph{final configuration} is a configuration having a final state.
The configuration space is spanned by the basis vectors in $\{\qubit{q,h,z}\mid q\in Q,h\in \integer^{k}, z\in \Gamma_1^*\times \cdots \Gamma_k^*\}$.
For a nonempty string $z_i$ and an index $h\in[|z_i|]$, $z_i[h]$ denotes the $h$th symbol in $z_i$. For example, if $z_i=01101$, then $z_i[1]=0$, $z_i[2]=1$, and $z_i[5]=1$.  The \emph{time-evolution operator} $U_{\delta}$ of $M$ acting on the configuration space  is induced from $\delta$ as
\[
U_{\delta}\qubit{p,h,z} = \sum_{q,w,d} \delta(p,z_h,q,z'_{h},d) \qubit{q,h_d,z'},
\]
where $p\in Q$, $h=(h_i)_{i\in[k]}\in \integer^k$,  $z=(z_i)_{i\in[k]}\in \Gamma_1^*\times\cdots\times\Gamma_k^*$, $z_h=(z_i[h_i])_{i\in[k]}$, $h_d = (h_i+d_i)_{i\in[k]}$, and $z'=(z'_i)_{i\in[k]}$,
where each $z'_i$ is the same as $z$ except for the $h_i$-th symbol. Moreover, in the summation,  variables $q$, $z'_h=(z'_i[h_i])_{i\in[k]}$, and $d=(d_i)_{i\in[k]}$ respectively range over $Q$, $\tilde{\Gamma}^{(k)}$, and $D^k$. Any entry of $U_{\delta}$ is called an \emph{amplitude}. Quantum mechanics demand the time-evolution operator $U_{\delta}$ of the QTM to be unitary.

Each step of $M$ consists of two phases: first apply $\delta$ and then take a \emph{partial projective  measurement}, in which we check whether $M$ is in a final state (i.e., an inner state in $Q_f$).
Formally, a \emph{computation} of $M$ on input $x$ is a series of superpositions of configurations produced by sequential applications of $U_{\delta}$, starting from an initial configuration of $M$ on $x$.
If $M$ enters a final state along a computation path, this computation path terminates; otherwise, its computation must continue.

A $k$-tape QTM $M=(Q,\Sigma, \tilde{\Gamma}^{(k)}, \delta,q_0,Q_{f})$  is {\em well-formed} if $\delta$ satisfies three local
conditions: unit length, separability, and orthogonality. To explain these conditions, as presented
in \cite[Lemma 1]{Yam99}, we first introduce the following notations.
For  our convenience, we set
$E=\{0,\pm1,\pm2\}$ and $H=\{0,\pm1,\natural\}$.  Given elements $(p,\sigma,\tau)\in Q\times (\tilde{\Gamma}^{(k)})^2$,
$\epsilon=(\varepsilon_i)_{i\in[k]} \in E^k$, and
$d=(d_i)_{i\in[k]}\in D^k$, we define $D_{\epsilon} =\{d\in D^k\mid \forall
i\in [k]\,(|2d_i-\varepsilon_i|\leq1)\}$ and $E_{d}
=\{\varepsilon\in E^k\mid d\in D_{\epsilon}\}$. Moreover, let $h_{d,\epsilon} =(h_{d_i,\varepsilon_i})_{i\in[k]}$,
where $h_{d_i,\varepsilon_i}=2d_i-\varepsilon_i$ if $\varepsilon_i\neq 0$ and
$h_{d_i,\varepsilon_i}=\natural$ otherwise.
Finally, we define $\delta(p,\sigma) = \sum_{q,\tau,d} \delta(p,\sigma,q,\tau,d)\qubit{q,\tau,d}$ and
$\delta[p,\sigma,\tau|\epsilon] = \sum_{q\in Q} \sum_{d\in
D_{\epsilon}} \delta(p,\sigma,q,\tau,d)|E_{d}|^{-1/2}
\qubit{q} \qubit{h_{d,\epsilon}}$, where $\sigma,\tau\in\tilde{\Gamma}^{(k)}$ and $d\in D^k$.

\begin{enumerate}\vs{-2}
  \setlength{\topsep}{-2mm}%
  \setlength{\itemsep}{1mm}%
  \setlength{\parskip}{0cm}%

\item (unit length) $\|\delta(p,\sigma)\|=1$ for
all $(p,\sigma)\in Q\times \tilde{\Gamma}^{(k)}$.

\item (orthogonality) $\delta(p_1,\sigma_1)
\cdot\delta(p_2,\sigma_2)=0$ for any distinct pair $(p_1,\sigma_1),
(p_2,\sigma_2)\in Q\times \tilde{\Gamma}^{(k)}$.

\item (separability) $\delta[p_1,\sigma_1,\tau_1|\epsilon]\cdot
\delta[p_2,\sigma_2,\tau_2|\epsilon']=0$ for any distinct pair
$\epsilon,\epsilon'\in E^k$ and for any pair
$(p_1,\sigma_1,\tau_1), (p_2,\sigma_2,\tau_2)\in Q\times
(\tilde{\Gamma}^{(k)})^2$.
\end{enumerate}\vs{-2}

The well-formedness of a QTM captures the unitarity of its time-evaluation operator.

\begin{lemma}\label{lemma:well-formedness}
(Well-Formedness Lemma of {\rm \cite{Yam99}})\hs{2}
A $k$-tape QTM
$M$ with a transition function $\delta$ is
well-formed iff the time-evolution operator of $M$ preserves the $\ell_2$-norm.
\end{lemma}

Given a nonempty subset $K$ of $\complex$, we say that a QTM is of {\em K-amplitude} if all values of its quantum transition function belong to $K$. It is of significant importance to limit the choice of amplitude within an appropriate  set $K$ of reasonable numbers. With the use of such a set $K$, we introduce two important complexity classes $\bqp_K$ and $\fbqp_K$.

\begin{definition}
Let $K$ be any nonempty subset of $\complex$ and let $\Sigma$
be any alphabet.
\begin{enumerate}\vs{-2}
  \setlength{\topsep}{-2mm}%
  \setlength{\itemsep}{1mm}%
  \setlength{\parskip}{0cm}%

\item  A subset $S$ of $\Sigma^*$ is in $\bqp_{K}$ if there exists a
multi-tape, polynomial-time, well-formed QTM $M$ with $K$-amplitudes
such that, for every string $x$, $M$ outputs $S(x)$ with probability
at least $2/3$ {\cite{BV97}}.

\item A single-valued function $f$ {}from $\Sigma^*$ to $\Sigma^*$ is called {\em
bounded-error quantum polynomial-time computable} if there exists a
multi-tape, polynomial-time, well-formed QTM $M$ with $K$-amplitudes
such that, on every input $x$, $M$ outputs $f(x)$ with probability at
least $2/3$.  Let $\fbqp_{K}$ denote the set of all such functions \cite{Yam03}.
\end{enumerate}
\end{definition}

The use of \emph{arbitrary} complex amplitudes turns out to make $\bqp_K$ quite powerful. As Adleman, DeMarrais, and Huang \cite{ADH97} demonstrated, $\bqp_{\complex}$ contains all possible languages, and thus  $\bqp_{\complex}$ is no longer recursive. Therefore, we usually pay our attention only to  polynomial-time approximable amplitudes and, for this reason, when $K=\appcomplex$, we always drop subscript $K$ and briefly write $\bqp$ and $\fbqp$ instead of $\bqp_K$ and $\fbqp_K$, respectively.
It is also possible to further limit the amplitude set $K$ to $\{0,\pm1,\pm\frac{3}{5},\pm\frac{4}{5}\}$ because  $\bqp=\bqp_{\{0,\pm1,\pm\frac{3}{5},\pm\frac{4}{5}\}}$ holds
\cite{ADH97}.

\subsection{Quantum Circuits}\label{sec:quantum-circuit}

A {\em $k$-qubit quantum gate}, for $k\in\nat^{+}$, is a unitary operator acting on a Hilbert space of dimension $2^k$.
Since any quantum state is a vector  in a certain Hilbert space, each entry of such a quantum state is customarily called an \emph{amplitude}. Unitary operators, such as  the Walsh-Hadamard transform (WH) and the controlled-NOT transform (CNOT) defined as
\[
WH = \frac{1}{\sqrt{2}} \left( \begin{array}{cc}
1 & 1 \\ 1 & -1
\end{array} \right)
\;\; \text{and} \;\;
CNOT = \left( \begin{array}{cccc} 1 & 0 & 0 & 0 \\
0 & 1 & 0 & 0 \\  0 & 0 & 0 & 1 \\ 0 & 0 & 1 & 0
\end{array} \right),
\]
are typical quantum gates acting on $1$ qubit and $2$
qubits, respectively.
If a quantum gate $U$ acting on $k$ qubits is applied to
a $k$-qubit quantum state $\qubit{\phi}$, then we obtain a new quantum state $U\qubit{\phi}$.  Notice that every quantum gate preserves the norm of any quantum state given as an input.
A \emph{quantum circuit} is a product of a finite number of layers, where each layer is a Kronecker product of allowed quantum gates. We often concentrate on a particular set of quantum gates to construct quantum circuits. Let us consider the specific quantum gates:
the CNOT gate and three one-qubit gates of the form
\[
 Z_{1,\theta}=
\left( \begin{array}{cc} e^{i \theta} & 0 \\ 0 & 1
\end{array} \right), \;\;
 Z_{2,\theta}=
\left( \begin{array}{cc} 1 & 0 \\ 0 & e^{i \theta} \end{array}
\right), \;\; \text{and} \;\;
 R_{\theta}=
\left( \begin{array}{rr} \cos\theta & -\sin\theta \\ \sin\theta &
\cos\theta \end{array} \right),
\]
where $\theta$ is a real number with $0\leq \theta < 2\pi$. Notice that $WH$ equals $R_{\frac{\pi}{4}}$.
Those gates form a \emph{universal set} of quantum
gates \cite{BBC95} since $WH$ and $Z_{2,\frac{\pi}{4}}$ (called the $\pi/8$ gate) can approximate any single qubit unitary operator to arbitrary accuracy. For convenience, we call them \emph{elementary gates}.
The set of $CNOT$, $WH$, and $Z_{2,\frac{\pi}{4}}$ is also known to be universal \cite{BMP+99}.

Given an amplitude set $K$, a quantum circuit $C$ is said to have \emph{$K$-amplitudes} if all entries of each quantum gate used inside $C$ are drawn from $K$.
For any $k$-qubit quantum gate and any integer $n>k$, $G^{(n)}$
denotes  $G\otimes I^{\otimes n-k}$, the \emph{$n$-qubit expansion} of $G$.  An $n$-qubit quantum circuit is formally defined as a finite sequence
$(G_m,\pi_m),(G_{m-1},\pi_{m-1}),\cdots,(G_1,\pi_1)$ such that each
$G_i$ is an $n_i$-qubit quantum gate with $n_i\leq n$ and $\pi_i$ is a
permutation on $\{1,2,\ldots,n\}$. This quantum circuit represents the
unitary operator $U=U_mU_{m-1}\cdots U_1$, where
$U_i$ is of the form $V_{\pi_i}^{\dagger} G_i^{(n)} V_{\pi_i}$ and $V_{\pi_i}(\qubit{x_1\cdots
x_n})= \qubit{x_{\pi_i(1)}\cdots x_{\pi_i(n)}}$ for each $i\in[m]$.
The {\em size} of a quantum circuit is the total number of quantum gates in it. Yao
\cite{Yao93} and later Nishimura and Ozawa \cite{NO02} showed that,
for any $k$-tape QTM and a polynomial $p$, there exists a family of
quantum circuits of size $O(p(n)^{k+1})$ that exactly simulates $M$.

A family  $\{C_n\}_{n\in\nat}$ of a quantum circuit is said to be {\em
$\p$-uniform} if there exists a deterministic (classical) Turing
machine that, on input $1^n$, produces a code of $C_n$ in time polynomial in the size of $C_n$, provided that
we use a fixed, efficient coding scheme to describe each quantum circuit.

\begin{proposition}\label{uniform-circuit-BQP}{\rm \cite{Yao93} (see also \cite{NO02})} \hs{2}
For any language $L$ over an alphabet $\{0,1\}$, $L$ is in $\bqp$ iff there exist a polynomial $p$
and a $\p$-uniform family  $\{C_n\}_{n\in\nat}$ of quantum circuits having $\tilde{\complex}$-amplitudes
such that
$\|\bra{L(x)}C_{|x|} \ket{x10^{p(|x|)}}\|^2 \geq\frac{2}{3}$ holds for all $x\in \{0,1\}^*$, where $L$ is seen as the characteristic function of $L$.
\end{proposition}

Yao's inspiring proof \cite{Yao93} of Proposition \ref{uniform-circuit-BQP} gives a foundation to our proof of Lemma \ref{lemma:second-part}, which provides in Section \ref{sec:definition} a simulation of a well-formed QTM by an appropriately chosen $\squareqp$-function.

\section{A New, Simple Schematic Definition}\label{sec:recursive-def}

As noted in Section \ref{sec:introduction}, the ``schematic'' definition of recursive function means an inductive (or constructive) way of defining the set of computable functions and it involves
a small set of so-called \emph{initial functions} as well as
a small set of \emph{construction rules}, which are sequentially applied
finitely many times to build more complex functions from certain functions that have been already constructed. A similar schematic
characterization is known for polynomial-time computable functions
(as well as languages) \cite{Cob64,Con73,CK89,Meh76,Yam95}.  Along this line of work, we wish to present a new, simple schematic definition composed of a  small set of initial quantum functions and a small  set of construction rules, intending  to make this schematic definition appropriately capture polynomial-time computable quantum functions, where a \emph{quantum function} is a function mapping $\HH_{\infty}$ to $\HH_{\infty}$.
As remarked briefly in Section \ref{sec:introduction}, it is important to note that our term of ``quantum function'' is quite different from the one used in, e.g., \cite{Yam03}, in which ``quantum function'' refers to  functions that take classical input strings and produce either classical output  strings or acceptance probabilities of multi-tape polynomial-time well-formed QTMs and thus it maps $\Sigma^*$ to either $\Sigma^*$ or the real unit interval $[0,1]$, where $\Sigma$ is an appropriate alphabet.

\subsection{Definition of $\squareqp$-Functions}\label{sec:definition}

Our schematic definition induces a special function class, called $\squareqp$ (where $\Box$ is pronounced ``square''), capturing polynomial-time computable quantum functions mapping $\HH_{\infty}$ to $\HH_{\infty}$, which is composed of a small set of initial quantum functions and four special  construction rules: composition, swapping, branching, and multi-qubit quantum recursion.
Definition \ref{def:initial} formally presents our schematic definition.

Hereafter, we say that a quantum function $f$ from $\HH_{\infty}$ to $\HH_{\infty}$ is
{\em dimension-preserving} if, for every quantum state $\qubit{\phi}\in\HH_{\infty}$ and any number $n\in\nat^{+}$, $\qubit{\phi}\in\HH_{2^n}$ implies
$f(\qubit{\phi})\in\HH_{2^n}$ (i.e., $\ell(\qubit{\phi}) = \ell(f(\qubit{\phi}))$).

\begin{definition}\label{def:initial}
Let $\squareqp$ denote the collection of all quantum functions that are obtained
from the initial quantum functions in Scheme I by a finite number (including zero) of applications of
construction rules II--IV to quantum functions that have been already constructed, where Schemata I--IV\footnote{The current formalism of Schemata I--IV corrects discrepancies caused by the early formalism given in the extended abstract \cite{Yam17}.}  are given as follows. Let
$\qubit{\phi}$ be any quantum state in $\HH_{\infty}$.
\begin{enumerate}\vs{-2}
  \setlength{\topsep}{-2mm}%
  \setlength{\itemsep}{1mm}%
  \setlength{\parskip}{0cm}%

\item[I.] The initial quantum functions.  Let
$\theta\in[0,2\pi)\cap\appcomplex$ and $a\in\{0,1\}$.

1) $I(\qubit{\phi}) = \qubit{\phi}$. (identity)

2) $PHASE_{\theta}(\qubit{\phi}) = \qubit{0}\!\measure{0}{\phi} + e^{i\theta}\qubit{1}\!\measure{1}{\phi}$. (phase shift)

3) $ROT_{\theta}(\qubit{\phi}) = \cos\theta \qubit{\phi}
+ \sin\theta
(\qubit{1}\!\measure{0}{\phi} - \qubit{0}\!\measure{1}{\phi})$. (rotation around $xy$-axis at angle $\theta$)

4) $NOT(\qubit{\phi}) = \qubit{0}\!\measure{1}{\phi} +
\qubit{1}\!\measure{0}{\phi}$. (negation)

5) $SWAP(\qubit{\phi}) = \left\{ \begin{array}{ll}
 \qubit{\phi} & \mbox{if $\ell(\qubit{\phi})\leq 1$,} \\
 \sum_{a,b\in\{0,1\}} \qubit{ab}\!\measure{ba}{\phi}
 & \mbox{otherwise.}
 \end{array}\right.$ \text{(swapping of 2 qubits)}

6) $MEAS[a](\qubit{\phi}) = \qubit{a}\!\measure{a}{\phi}$. (partial projective measurement)

\item[II.] The composition rule.
From $g$ and $h$, we define $Compo[g,h]$ as
follows:

\n\hs{10}$Compo[g,h](\qubit{\phi}) = g\circ h(\qubit{\phi})$
($=g(h(\qubit{\phi}))$).

\item[III.] The branching rule.
From $g$ and $h$, we define $Branch[g,h]$
as:

\n\hs{10}(i) $Branch[g,h](\qubit{\phi}) =
\qubit{\phi}$ \hs{44}if $\ell(\qubit{\phi})\leq 1$, \\
\n\hs{9}(ii) $Branch[g,h](\qubit{\phi}) =
\qubit{0}\otimes g(\measure{0}{\phi}) + \qubit{1}\otimes
h(\measure{1}{\phi})$ \hs{3}otherwise.

\item[IV.] The multi-qubit quantum recursion rule.
From $g$, $h$, dimension-preserving $p$, and $k,t\in\nat^{+}$, we define
$kQRec_t[g,h,p|\FF_k]$ as:

\n\hs{10}(i) $kQRec_t[g,h,p|\FF_k](\qubit{\phi}) = g(\qubit{\phi})$
\hs{38}if $\ell(\qubit{\phi})\leq t$, \\
\n\hs{9}(ii) $kQRec_t[g,h,p|\FF_k](\qubit{\phi}) =
h( \sum_{s:|s|=k} \qubit{s}\otimes f_s(\measure{s}{\psi_{p,\phi}}) )$ \hs{3}otherwise,

where $\qubit{\psi_{p,\phi}} = p(\qubit{\phi})$ and $\FF_k=\{f_s\}_{s\in\{0,1\}^k} \subseteq  \{kQRec_t[g,h,p|\FF_k],I\}$.
To emphasize ``$k$,'' we call this rule by the \emph{$k$-qubit quantum recursion}. In the case of $k=1$, we write $QRec_t[g,h,p|f_0,f_1]$ in place of $1QRec_t[g,h,p|\{f_0,f_1\}]$ for brevity.
\end{enumerate}
\end{definition}

In Scheme I, $PHASE_{\theta}$ and $ROT_{\theta}$ correspond respectively to the matrices $Z_{2,\theta}$ and $R_{\theta}$ given in Section \ref{sec:quantum-circuit}.
Latter, we will argue that an angle $\theta$ in $PHASE_{\theta}$ and $ROT_{\theta}$ could be fixed to $\pi/4$.
The quantum function $MEAS$ is associated with a \emph{partial projective measurement}, in the computational basis $\{0,1\}$, applied to the first qubit of $\qubit{\phi}$, when $\ell(\qubit{\phi})\geq1$, and it obviously follows that $\ell(MEAS[i](\qubit{\phi})) \leq \ell(\qubit{\phi})$.

Before proceeding further, to help the reader understand the behaviors of the initial quantum functions listed in Scheme I, we briefly illustrate how these functions transform basic qustrings of length $3$. For bits $a,b,c,d\in\{0,1\}$ with $d\neq a$, it follows  that $I(\qubit{abc})=\qubit{abc}$, $PHASE_{\theta}(\qubit{abc} = e^{i\theta a}\qubit{abc}$,  $ROT_{\theta}(\qubit{abc})=\cos\theta \qubit{abc} + (-1)^{a}\sin\theta \qubit{\overline{a}bc}$, $NOT(\qubit{abc}) = \qubit{\overline{a}bc}$, $SWAP(\qubit{abc}) = \qubit{bac}$,  $MEAS[a](\qubit{abc})=\qubit{abc}$, and $MEAS[d](\qubit{abc})=\bm{0}$, where $\overline{a}=1-a$.


Scheme IV is a core of the definition of $\squareqp$.
The standard recursion rule used to define a primitive recursive function $f$ from two functions $g$ and $h$ has the form: $f(0,x)=g(x)$ and $f(n+1,x) = h(n,x,f(n,x))$ for any $n\in\nat$. This rule requires an internal counter (in the first argument place of $f$) that controls the number of iterated applications of $h$.
In Scheme IV, however, we do not use such a counter. Instead, we use a divide-and-conquer strategy to slice a given quantum state qubit by qubit. At each inductive step, we deal with $k$-qubit shorter quantum states until the quantum state has become length $t$ or less.
We wish to provide two concrete examples of how Scheme IV works in the cases of $k=1,2$. Notice that, in the case of $k=1$, Scheme IV becomes the \emph{1-qubit} (or \emph{single-qubit}) \emph{quantum recursion rule} described as:

\n\hs{10}(i$'$) $QRec_t[g,h,p|f_0,f_1](\qubit{\phi}) = g(\qubit{\phi})$
\hs{56}if $\ell(\qubit{\phi})\leq t$, \\
\n\hs{9}(ii$'$) $QRec_t[g,h,p|f_0,f_1](\qubit{\phi}) =
h( \qubit{0}\otimes f_0(\measure{0}{\psi_{p,\phi}}) + \qubit{1}\otimes f_1(\measure{1}{\psi_{p,\phi}}) )$ \hs{3}otherwise,

\n where each of $f_0$ and $f_1$ must be either $I$ or $QRec_t[g,h,p|f_0,f_1]$.

\s
\n{\bf Example 1.}
We first consider the case of $k=1$. In this example, we set $t=1$, $g=NOT$, $h=SWAP$, and $p=NOT$, and we briefly write $F$ for $QRec_{1}[NOT,SWAP,NOT|f_0,f_1]$ with $f_0=I$ and $f_1= F$. It is worth mentioning that $\bm{0}$ is a special object and we obtain $g(\bm{0})=\bm{0}$ for any quantum function $g$ by Lemma \ref{hatsquare-property}(1). Let $\qubit{\phi}$ denote any quantum state in $\HH_{\infty}$ given as an input to $F$.

\s
(1) Assume that $\qubit{\phi}$ has length $1$ and is  of the form $\alpha\qubit{0}+\beta\qubit{1}$ in general. By Lemma \ref{hatsquare-property}(2), it suffices for us to consider the computational basis $B_1= \{\qubit{0},\qubit{1}\}$. Since $\ell(\qubit{\phi})\leq t$, the outcomes of $f$ for those basis qubits $\qubit{0}$ and $\qubit{1}$ are calculated as follows.

\s
\n\hs{6}(i) $F(\qubit{0}) = g(\qubit{0}) = NOT(\qubit{0})=\qubit{1}$.

\s
\n\hs{5}(ii) $F(\qubit{1}) = g(\qubit{1}) = NOT(\qubit{1})=\qubit{0}$.

\s
Obviously, $F(\bm{0})=\bm{0}$ holds. Since $\qubit{\phi}$ is a superposition of the form $\alpha\qubit{0}+\beta\qubit{1}$, the above calculations instantly imply

\s
\n\hs{10} $F(\qubit{\phi}) = F(\alpha\qubit{0}+\beta\qubit{1}) = NOT(\alpha\qubit{0}+\beta\qubit{1}) = \alpha NOT(\qubit{0}) + \beta NOT(\qubit{1}) = \alpha\qubit{1}+\beta\qubit{0}$.

\s
(2) Next, let us assume that $\qubit{\phi}$ is of length $2$. In the case where $\qubit{\phi}$ is the basis quantum state $\qubit{00}$ in $B_2=\{\qubit{00},\qubit{01},\qubit{10},\qubit{11}\}$, we obtain  $\qubit{\psi_{p,\phi}} = \qubit{\psi_{NOT,00}}= NOT(\qubit{00}) = \qubit{10}$. Similarly, if $\qubit{\phi}$ is $\qubit{01}$, $\qubit{10}$, and $\qubit{11}$, then the quantum state $\qubit{\psi_{p,\phi}}$ is $\qubit{11}$, $\qubit{00}$, and $\qubit{01}$, respectively.
Concerning the partial projective measurement, it follows that,
for any bit $a\in\{0,1\}$,  $\measure{1}{1a}=\qubit{a}$, $\measure{0}{0a} =\qubit{a}$,  $\measure{1}{0a}= \bm{0}$, and $\measure{0}{1a} =\bm{0}$. Note also that $I(\bm{0})=\bm{0}$. Using these equalities together with $F(\qubit{0})=\qubit{1}$ and $F(\qubit{1})=\qubit{0}$ obtained in (1), we can calculate the outcome $F(\qubit{\phi})$ as follows.

\s
\n\hs{5}(i) $F(\qubit{00}) = h( \qubit{0}\otimes f_0(\measure{0}{\psi_{NOT,00}}) + \qubit{1}\otimes f_1(\measure{1}{\psi_{NOT,00}}) = h( \qubit{0}\otimes I(\measure{0}{10}) + \qubit{1}\otimes F(\measure{1}{10}) ) =
h( \qubit{0}\otimes I(\bm{0}) + \qubit{1}\otimes F(\qubit{0}) ) =
h(\qubit{1}\otimes F(\qubit{0})) = SWAP( \qubit{1}\otimes \qubit{1}) = SWAP(\qubit{11}) =\qubit{11}$.

\s
\n\hs{5}(ii) $F(\qubit{10}) = h( \qubit{0}\otimes f_0(\measure{0}{\psi_{NOT,10}}) + \qubit{1}\otimes f_1(\measure{1}{\psi_{NOT,10}}) = h( \qubit{0}\otimes I(\measure{0}{00}) + \qubit{1}\otimes F(\measure{1}{00}) ) =
h( \qubit{0}\otimes I(\qubit{0}) + \qubit{1}\otimes F(\bm{0}) ) =
h(\qubit{0}\otimes I(\qubit{0})) =  SWAP(\qubit{0}\otimes \qubit{0})
= SWAP( \qubit{00}) = \qubit{00}$.

\s
\n\hs{5}(iii) $F(\qubit{01}) = h( \qubit{0}\otimes f_0(\measure{0}{\psi_{NOT,01}}) + \qubit{1}\otimes f_1(\measure{1}{\psi_{NOT,01}}) = h( \qubit{0}\otimes I(\measure{0}{11}) + \qubit{1}\otimes F(\measure{1}{11}) ) =
h( \qubit{0}\otimes I(\bm{0}) + \qubit{1}\otimes F(\qubit{1}) ) =
h(\qubit{1}\otimes F(\qubit{1})) = SWAP( \qubit{1}\otimes\qubit{0}) = SWAP( \qubit{10}) = \qubit{01}$.

\s
\n\hs{5}(iv) $F(\qubit{11}) = h( \qubit{0}\otimes f_0(\measure{0}{\psi_{NOT,11}}) + \qubit{1}\otimes f_1(\measure{1}{\psi_{NOT,11}}) = h( \qubit{0}\otimes I(\measure{0}{01}) + \qubit{1}\otimes F(\measure{1}{01}) ) =
h( \qubit{0}\otimes I(\qubit{1}) + \qubit{1}\otimes F(\bm{0}) ) =
h(\qubit{0}\otimes I(\qubit{1})) = SWAP( \qubit{0}\otimes \qubit{1}) = SAWP( \qubit{01}) = \qubit{10}$.

\s
By Lemma \ref{hatsquare-property}(2), when  $\qubit{\phi}$ is of the form   $\alpha\qubit{00}+\beta\qubit{10}$ for example, we obtain

\s
\n\hs{10} $F(\qubit{\phi}) = F(\alpha\qubit{00}+\beta\qubit{10}) = \alpha F(\qubit{00}) + \beta F(\qubit{10}) = \alpha\qubit{11}+\beta\qubit{00}$.

\s
(3) Let us consider the case where the length of $\qubit{\phi}$ is $3$. For the computational basis  $B_3= \{\qubit{000},\qubit{001},\qubit{010},\ldots,\qubit{111}\}$,
we calculate the outcomes $F(\qubit{\phi})$ only for inputs   $\qubit{\phi}$ in $\{\qubit{001},\qubit{101}\}$.
Notice that $\qubit{\psi_{NOT,001}}=\qubit{101}$ and $\qubit{\psi_{NOT,101}} = \qubit{001}$. Recall from (1)--(2) that $F(\qubit{01})=\qubit{01}$ and $F(\bm{0})=\bm{0}$.

\s
\n\hs{5}(i) $F(\qubit{001}) = h( \qubit{0}\otimes f_0(\measure{0}{\psi_{NOT,001}}) + \qubit{1}\otimes f_1(\measure{1}{\psi_{NOT,001}}) = h( \qubit{0}\otimes I(\measure{0}{101}) + \qubit{1}\otimes F(\measure{1}{101}) ) =
h( \qubit{0}\otimes I(\bm{0}) + \qubit{1}\otimes F(\qubit{01}) ) =
h(\qubit{1}\otimes F(\qubit{01})) = SWAP( \qubit{1}\otimes \qubit{01}) = SAWP( \qubit{101}) = \qubit{011}$.

\s
\n\hs{5}(ii) $F(\qubit{101}) = h( \qubit{0}\otimes f_0(\measure{0}{\psi_{NOT,101}}) + \qubit{1}\otimes f_1(\measure{1}{\psi_{NOT,101}}) = h( \qubit{0}\otimes I(\measure{0}{001}) + \qubit{1}\otimes F(\measure{1}{001}) ) =
h( \qubit{0}\otimes I(\qubit{01}) + \qubit{1}\otimes F(\bm{0}) ) =
h(\qubit{0}\otimes I(\qubit{01})) = SWAP( \qubit{0}\otimes \qubit{01}) = SAWP( \qubit{001}) = \qubit{001}$.

\s
If $\qubit{\phi}$ is of the form $\alpha\qubit{001}+\beta\qubit{101}$ for example, Lemma \ref{hatsquare-property}(2) implies that $F(\qubit{\phi}) = F(\alpha\qubit{001}+\beta\qubit{101}) =  \alpha F(\qubit{001}) + \beta F(\qubit{101}) = \alpha\qubit{011} + \beta\qubit{001}$.

\s
\n{\bf Example 2.}
We see another example for the case of $k=2$. By setting $t=3$, $g=NOT$, $h=SWAP$, and $p=NOT$ for instance, we write $F$ for $2QRec_{3}[NOT,SWAP,NOT|f_{00},f_{01},f_{10},f_{11}]$ with $f_{00}=I$, $f_{01}=NOT$, $f_{10}=F$, and $f_{11}= F$. Let $\qubit{\phi}$ denote any quantum state in $\HH_{\infty}$.

\s
(1) When the length of $\qubit{\phi}$ is at most $3$, we instantly obtain $F(\qubit{\phi}) = g(\qubit{\phi}) = NOT(\qubit{\phi})$. For example, $F(\qubit{0}) = \qubit{1}$, $F(\qubit{10}) = \qubit{00}$, and $F(\qubit{101}) = \qubit{001}$.

\s
(2) Assuming that the length of $\qubit{\phi}$ is $4$, let us consider the basis quantum states in $B_4 = \{\qubit{0000},\qubit{0001},\qubit{0010},\ldots,\qubit{1111}\}$.
Here, we are focused only on two cases:  $\qubit{\phi}\in\{\qubit{0101},\qubit{1011}\}$.
We remark that $\qubit{\psi_{NOT,0010}} = NOT(\qubit{0010}) = \qubit{1010}$ and $\qubit{\psi_{NOT,1011}} = NOT(\qubit{1011}) = \qubit{0011}$.

\s
\n\hs{10}(i) $F(\qubit{0010}) = h(
\qubit{00}\otimes f_{00}(\measure{00}{\psi_{NOT,0010}})
+ \qubit{01}\otimes f_{01}(\measure{01}{\psi_{NOT,0010}})
+ \qubit{10}\otimes f_{10}(\measure{10}{\psi_{NOT,0010}})
+ \qubit{11}\otimes f_{11}(\measure{11}{\psi_{NOT,0010}}) )
= h(
\qubit{00}\otimes I(\measure{00}{1010})
+ \qubit{01}\otimes NOT(\measure{01}{1010})
+ \qubit{10}\otimes F(\measure{10}{1010})
+ \qubit{11}\otimes F(\measure{11}{1010}) )
= h(
\qubit{00}\otimes I(\bm{0})
+ \qubit{01}\otimes NOT(\bm{0})
+ \qubit{10}\otimes F(\qubit{10})
+ \qubit{11}\otimes F(\bm{0}) )
= h(
\qubit{10}\otimes F(\qubit{10}) )
= SWAP(
\qubit{10}\otimes \qubit{00} ) = SWAP(\qubit{1000}) = \qubit{0100}$.

\s
\n\hs{10}(ii) $F(\qubit{1101}) = h(
\qubit{00}\otimes f_{00}(\measure{00}{\psi_{NOT,1101}})
+ \qubit{01}\otimes f_{01}(\measure{01}{\psi_{NOT,1101}})
+ \qubit{10}\otimes f_{10}(\measure{10}{\psi_{NOT,1101}})
+ \qubit{11}\otimes f_{11}(\measure{11}{\psi_{NOT,1101}}) )
= h(
\qubit{00}\otimes I(\measure{00}{0101})
+ \qubit{01}\otimes NOT(\measure{01}{0101})
+ \qubit{10}\otimes F(\measure{10}{0101})
+ \qubit{11}\otimes F(\measure{11}{0101}) )
= h(
\qubit{00}\otimes I(\bm{0})
+ \qubit{01}\otimes NOT(\qubit{01})
+ \qubit{10}\otimes F(\bm{0})
+ \qubit{11}\otimes F(\bm{0}) )
= h(
\qubit{01}\otimes NOT(\qubit{01}) )
= SWAP(
\qubit{01}\otimes \qubit{11} ) = SWAP(\qubit{0111}) = \qubit{1011}$.

\subsection{A Subclass of $\squareqp$ and Basic Properties}

Among the six initial quantum functions in Scheme I, $MEAS$ makes a quite different behavior. It can change the norm as well as the dimensionality of quantum states, and that fact makes it irreversible in nature.
For this reason, it is often beneficial in practice to limit our attention to a subclass of $\squareqp$, called $\hatsquareqp$, which entirely prohibits the use of $MEAS$.

\begin{definition}
The notation $\hatsquareqp$ denotes the subclass of $\squareqp$ defined by Schemata I-IV except for $MEAS$ listed in Scheme I.
\end{definition}

With no use of $MEAS$ in Scheme I, every quantum function $f$ in $\hatsquareqp$ preserves the dimensionality of inputs; in other words, $f$ satisfies $\ell(f(\qubit{\phi})) = \ell(\qubit{\phi})$ for any input $\qubit{\phi}\in\HH_{\infty}$.

Next, let us demonstrate how to construct typical
unitary gates using our schematic definition.

\begin{lemma}\label{simple-function}
The following functions are in $\hatsquareqp$. Let $\qubit{\phi}$ be any element in $\HH_{\infty}$.
\begin{enumerate}\vs{-2}
  \setlength{\topsep}{-2mm}%
  \setlength{\itemsep}{1mm}%
  \setlength{\parskip}{0cm}%

\item $CNOT(\qubit{\phi}) = \left\{
\begin{array}{ll}
\qubit{\phi} & \mbox{if $\ell(\qubit{\phi})\leq1$,} \\
\qubit{0}\!\measure{0}{\phi} +
\qubit{1}\otimes NOT(\measure{1}{\phi}) & \mbox{otherwise.}
\end{array} \right.$
(controlled-NOT)

\item $Z_{1,\theta}(\qubit{\theta}) =
e^{i\theta}\qubit{0}\!\measure{0}{\phi} + \qubit{1}\!\measure{1}{\phi}$.

\item $zROT_{\phi}(\qubit{\phi}) =
e^{i\theta}\qubit{0}\!\measure{0}{\phi} +  e^{-i\theta}\qubit{1}\!\measure{1}{\phi}$.
(rotation around the $z$-axis)

\item  $GPS_{\theta}(\qubit{\phi}) = e^{i\theta} \qubit{\phi}$. (global phase shift)

\item $WH(\qubit{\phi}) = \frac{1}{\sqrt{2}}\qubit{0}\otimes
(\measure{0}{\phi} + \measure{1}{\phi}) +
\frac{1}{\sqrt{2}}\qubit{1}\otimes (\measure{0}{\phi} -
\measure{1}{\phi}) $. (Walsh-Hadamard transform)

\item $CPHASE_{\theta}(\qubit{\phi}) = \left\{
\begin{array}{ll}
\qubit{\phi} & \mbox{if $\ell(\qubit{\phi})\leq1$,} \\
\frac{1}{\sqrt{2}} \sum_{b\in\{0,1\}} ( \qubit{0}\!\measure{b}{\phi} + e^{i\theta b} \qubit{1}\!\measure{b}{\phi} ) & \mbox{otherwise.}
\end{array} \right.$
(controlled-PHASE)
\end{enumerate}
\end{lemma}

\begin{proof}
It suffices to build each of the quantum functions given in the lemma from the initial quantum functions and by applying the construction rules.
These target functions are constructed as follows.
(1) $CNOT=Branch[I,NOT]$. Note that, when $\ell(\qubit{\phi})\leq1$, Scheme III(i) implies $CNOT(\qubit{\phi})=\qubit{\phi}$, matching Item 1 for $CNOT$.
(2) $Z_{1,\theta} = NOT \circ PHASE_{\theta}\circ NOT$.
(3) $zROT_{\theta}= Z_{1,\theta}\circ PHASE_{-\theta}$.
(4) $GPS_{\theta} = Z_{1,\theta} \circ PHASE_{\theta}$.
(5) $WH= ROT_{\frac{\pi}{4}}\circ NOT$.
(6) $CPHASE_{\theta} = Branch[WH,f]$, where $f=Branch[I,GPS_{\theta}]\circ WH\circ NOT$.
\end{proof}

Since any $\squareqp$-function $f$ is constructed by applying Schemata I--IV, an application of one of the schemata is viewed as a basic step of the construction process of generating $f$. This fact helps us define the \emph{descriptional complexity} of $f$ to be the minimal number of times we use those schemata in order to construct $f$. For instance, all the initial functions have descriptional complexity $1$ because they use Scheme I only once. As demonstrated in the proof of Lemma \ref{simple-function}, the quantum functions $CNOT$, $Z_{1,\theta}$, and $WH$  have descriptional complexity at most $3$ and $zROT_{\theta}$ and $GPS_{\theta}$ have descriptional complexity at most $4$, whereas $CPHASE_{\theta}$ is of descriptional complexity at most $15$. This complexity measure is essential in proving, e.g.,
Lemma \ref{hatsquare-property} since the lemma will be proven by induction on the descriptional complexity of a target quantum function.
In Section \ref{sec:descriptional-complexity}, we will give a short discussion on this complexity measure for a future study.

Fundamental properties of
$\hatsquareqp$-functions are given in the following lemma.
A quantum function from $\HH_{\infty}$ to
$\HH_{\infty}$ is called {\em norm-preserving} if
$\|f(\qubit{\phi})\|=\|\qubit{\phi}\|$ holds for all quantum states
$\qubit{\phi}$ in $\HH_{\infty}$.

\begin{lemma}\label{hatsquare-property}
Let $f$ be any quantum function in $\hatsquareqp$ and let $\qubit{\phi},\qubit{\psi}\in \HH_{\infty}$ and $\alpha\in\complex$.
\begin{enumerate}\vs{-2}
  \setlength{\topsep}{-2mm}%
  \setlength{\itemsep}{1mm}%
  \setlength{\parskip}{0cm}%

\item $f(\bm{0}) = \bm{0}$, where $\bm{0}$ is the null vector.
\item $f(\qubit{\phi}+\qubit{\psi}) = f(\qubit{\phi}) + f(\qubit{\psi})$.
\item $f(\alpha\qubit{\phi}) = \alpha \cdot f(\qubit{\phi})$.
\item $f$ is dimension-preserving and norm-preserving.
\end{enumerate}
\end{lemma}

\begin{proof}
Let $f$ be any $\hatsquareqp$-function, let $\qubit{\phi},\qubit{\psi}\in\HH_{\infty}$, and let $\alpha\in\complex$ and $\theta\in[0,2\pi)$ be constants.  As mentioned earlier, we will prove the lemma by induction of the \emph{descriptional complexity} of $f$. If $f$ is one of the initial quantum functions in Scheme I, then it is easy to check that it satisfies Conditions 1--4 of the lemma. In particular, when $\qubit{\phi}$ is the null vector $\bm{0}$, all of $e^{i\theta}\qubit{\phi}$, $\cos\theta\qubit{\phi}$, $\sin\theta\qubit{\phi}$,  $\measure{0}{\phi}$, $\measure{1}{\phi}$, and $\measure{ba}{\phi}$ used in Scheme I are $\bm{0}$; thus, $\qubit{b}\otimes \measure{1}{\phi}$ and $\qubit{b}\otimes \measure{0}{\phi}$ are also $\bm{0}$ for each bit $b\in\{0,1\}$. Therefore, Condition 1 follows.

Among Schemata II--IV, let us consider Scheme IV since the other schemata are easily shown to meet Conditions 1--4.
Let $g$, $h$, and $p$ be quantum functions in $\hatsquareqp$ and assume that  $p$ is dimension-preserving. By induction hypothesis, we assume that $g$, $h$, and $p$ satisfy Conditions 1--4. For readability, we write $f$ for $kQRec_t[g,h,p|\FF_k]$. In what follows, we are focused only on Conditions 2 and 4 since the other conditions are easy to check. Our argument will  employ induction on the length of input $\qubit{\phi}$ given to $f$.

(i) Our goal is to show that $f$ satisfies Condition 2. Firstly, consider the case of $\ell(\qubit{\phi})\leq t$. It then follows that $f(\qubit{\phi}+\qubit{\xi}) = g(\qubit{\phi}+\qubit{\xi}) = g(\qubit{\phi}) + g(\qubit{\xi})$ since $g$ is assumed to meet Condition 2.
Next, we consider the case where $\ell(\qubit{\phi})>t$. From the definition of $f$, we obtain $f(\qubit{\phi}+\qubit{\xi}) = h(\sum_{s:|s|=k} \qubit{s}\otimes f_s(\measure{s}{\psi_{p,\phi,\xi}}))$, where  $\qubit{\psi_{p,\phi,\xi}} = p(\qubit{\phi}+\qubit{\xi})$. Since $p(\qubit{\phi}+\qubit{\xi}) = p(\qubit{\phi}) + p(\qubit{\xi})$ by induction hypothesis, we conclude that  $\measure{s}{\psi_{p,\phi,\xi}} = \measure{s}{\psi_{p,\phi}} + \measure{s}{\psi_{p,\xi}}$ for each string $s\in\{0,1\}^k$.
It then follows  by induction hypothesis that  $f_s(\measure{s}{\psi_{p,\phi,\xi}}) = f_s(\measure{s}{\psi_{p,\phi}})+f_s(\measure{s}{\psi_{p,\xi}})$, leading to the equation $\qubit{s}\otimes f_s(\measure{s}{\psi_{p,\phi,\xi}}) = \qubit{s}\otimes f_s(\measure{s}{\psi_{p,\phi}}) + \qubit{s}\otimes f_s(\measure{s}{\psi_{p,\xi}})$.
Using Condition 2 for $h$, we obtain $h(\sum_{s:|s|=k}\qubit{s}\otimes f_s(\measure{s}{\psi_{p,\phi,\xi}})) = \sum_{s:|s|=k} h(\qubit{s}\otimes f_s(\measure{s}{\psi_{p,\phi}})) + \sum_{s:|s|=k} h(\qubit{s}\otimes f_s(\measure{s}{\psi_{p,\xi}}))$. We then conclude that $f(\qubit{\phi}+\qubit{\xi}) = f(\qubit{\phi})+f(\qubit{\xi})$.

(ii) We want to show that $f$ satisfies Condition 4. By induction hypothesis, it follows that, for any $s\in\{0,1\}^k$,  $\|f_s(\measure{s}{\psi_{p,\phi}})\| = \|\measure{s}{\psi_{p,\phi}}\|$ and $\|h(\qubit{\phi})\|=\|\qubit{\phi}\|$. These equalities imply that $\|f(\qubit{\phi})\|^2 = \| h( \sum_{s:|s|=k} \qubit{s}\otimes f_s(\measure{s}{\psi_{p,\phi}})) \|^2
= \|\sum_{s:|s|=k} \qubit{s}\otimes f_s(\measure{s}{\psi_{p,\phi}}) \|^2
= \sum_{s:|s|=k} \| f_s(\measure{s}{\psi_{p,\phi}}) \|^2$, which equals $\sum_{s:|s|=k} \|\measure{s}{\psi_{p,\phi}}\|^2$.
The last term coincides with $\|\qubit{\psi_{p,\phi}}\|^2$, which equals $\|p(\qubit{\phi})\|^2$.  This implies Condition 4 because $\|p(\qubit{\phi})\| = \|\qubit{\phi}\|$.
\end{proof}

Lemma \ref{hatsquare-property}(4) indicates that all functions in $\hatsquareqp$ also serve as functions mapping $\Phi_{\infty}$ to $\Phi_{\infty}$ in place of $\HH_{\infty}$ to $\HH_{\infty}$.

Given a quantum function $g$ that is dimension-preserving and norm-preserving, the \emph{inverse} of $g$ is a unique quantum function $f$ satisfying the condition that, for every $\qubit{\phi}\in\HH_{\infty}$, $f\circ g(\qubit{\phi}) = g\circ f(\qubit{\phi}) =\qubit{\phi}$. This special quantum function $f$ is expressed as $g^{-1}$.

The next proposition guarantees that $\hatsquareqp$ is closed
under ``inverse'' since $\hatsquareqp$ lacks $MEAS$.

\begin{proposition}\label{inverse}
For any quantum function $g\in\hatsquareqp$, $g^{-1}$ exists and belongs to $\hatsquareqp$.
\end{proposition}

\begin{proof}
We prove this proposition by induction on the aforementioned \emph{descriptional complexity} of $g$. If $g$ is one of the initial quantum functions, then we define its inverse $g^{-1}$ as: $I^{-1}=I$, $PHASE^{-1}_{\theta} = PHASE_{-\theta}$, $ROT^{-1}_{\theta} = ROT_{-\theta}$, $NOT^{-1} = NOT$, and $SWAP^{-1} = SWAP$. If $g$ is obtained from another quantum function or functions by one of the construction rules, then its inverse is defined as: $Compo[g,h]^{-1} = Compo[h^{-1},g^{-1}]$, $Branch[g,h]^{-1} = Branch[g^{-1},h^{-1}]$, and $kQRec_t[g,h,p|\{f_s\}_{s\in\{0,1\}^k}]^{-1} = kQRec_t[g^{-1},p^{-1},h^{-1}|\{f_s^{-1}\}_{s\in\{0,1\}^k}]$.
\end{proof}


Our choice of Schemata I--IV is motivated by a particular universal set of quantum gates. Notice that a different choice of initial quantum functions and construction rules may lead to a different set of $\squareqp$-functions. Scheme I uses an ``arbitrary'' angle of $\theta$ in $[0,2\pi)\cap\ptcomplex$ to introduce $PHASE_{\theta}$ and $ROT_{\theta}$; however, we can restrict $\theta$ to a unique value of $\frac{\pi}{4}$ since $PHASE_{\theta}$ and $ROT_{\theta}$ for an arbitrary value $\theta\in[0,2\pi)$ can be approximated to any desired accuracy using $WH$ and $Z_{2,\frac{\pi}{4}}$.
This last part comes from the fact that any single-qubit unitary operator can be approximated to any accuracy from the quantum gates $WH$ and $Z_{2,\frac{\pi}{4}}$ \cite{BMP+99} (see also \cite{NC00}). For a further discussion on the choice of the schemata, see Section \ref{sec:open-question}.

\subsection{Construction of More Complicated Quantum Functions}\label{sec:lemma}

Before presenting the main theorem (Theorem \ref{theorem:character}), we wish to prepare useful quantum functions and new construction rules derived directly from Schemata I--IV. These quantum functions and construction rules will be used for the proof of our key lemma (Lemma \ref{lemma:second-part}), which supports the main theorem.

For each $k\in\nat^{+}$, let us assume the standard lexicographic ordering $<$ on $\{0,1\}^k$ and all elements in $\{0,1\}^k$ are enumerated lexicographically as  $s_1<s_2<\cdots <s_{2^k}$.  For example, when $k=2$, we obtain $00<01<10<11$.
Given each string $s\in\{0,1\}^n$,  $s^R$ denotes the {\em reversal}
of $s$; that is,
$s^R=s_ns_{n-1}\cdots s_2s_1$ if $s=s_1s_2\cdots s_{n-1}s_n$.
We expand this notion to quantum states in the following manner.
Given a quantum state $\qubit{\phi}\in\HH_{2^n}$,  the {\em reversal} of
$\qubit{\phi}$,  denoted by $\qubit{\phi^R}$, is of the form  $\sum_{s:|s|=n}\measure{s}{\phi} \otimes \qubit{s^R}$, where $\measure{s}{\phi}$ is merely a scalar. For instance, if $\qubit{\phi}=\alpha\qubit{01}+\beta\qubit{10}$, then $\qubit{\phi^R} = \alpha\qubit{10}+\beta\qubit{01}$.

\begin{lemma}\label{lemma:special}
Let $k\in\nat$ with $k\geq2$, let $g,h\in\hatsquareqp$, and let $\GG_k=\{g_s\mid s\in\{0,1\}^k\}$  be a set of $\hatsquareqp$-functions. The following quantum functions all belong to $\hatsquareqp$. The lemma also holds even if  $\hatsquareqp$ is replaced by  $\squareqp$. Let $\qubit{\phi}$ be any quantum state in $\HH_{\infty}$.
\begin{enumerate}\vs{-2}
  \setlength{\topsep}{-2mm}%
  \setlength{\itemsep}{1mm}%
  \setlength{\parskip}{0cm}%

\item $Compo[\GG_k](\qubit{\phi}) = g_{s_1}\circ g_{s_2}\circ \cdots \circ  g_{s_{2^k}}(\qubit{\phi})$. (multiple composition)

\item $Switch_k[g,h](\qubit{\phi}) =
g(\qubit{\phi})$ if $\ell(\qubit{\phi}) < k$ and  $Switch_k[g,h](\qubit{\phi}) =
h(\qubit{\phi})$ otherwise. (switching)

\item $LENGTH_k[g](\qubit{\phi}) = \qubit{\phi}$ if $\ell(\qubit{\phi})< k$  and $LENGTH_k[g](\qubit{\phi}) = g(\qubit{\phi})$ otherwise.

\item $REMOVE_k(\qubit{\phi}) = \qubit{\phi}$ if $\ell(\qubit{\phi})< k$ and  $REMOVE_k(\qubit{\phi})= \sum_{s:|s|=k}\measure{s}{\phi}\otimes \qubit{s}$   otherwise.

\item $REP_k(\qubit{\phi}) = \qubit{\phi}$ if $\ell(\qubit{\phi})< k$ and $REP_k(\qubit{\phi}) =  \sum_{s:|s|=n-k} \measure{s}{\phi}\otimes \qubit{s}$ otherwise.

\item $SWAP_{k}(\qubit{\phi}) = \qubit{\phi}$ if $\ell(\qubit{\phi})<2k$ and $SWAP_{k}(\qubit{\phi}) = \sum_{s:|s|=k}\sum_{t:|t|=k} \qubit{st}\!\measure{ts}{\phi}$ otherwise.

\item $REVERSE(\qubit{\phi})= \qubit{\phi^R}$.

\item $Branch_k[\GG_k](\qubit{\phi})= \qubit{\phi}$  if $\ell(\qubit{\phi})<  k$ and $Branch_k[\GG_k](\qubit{\phi})= \sum_{s:|s|=k}\qubit{s}\otimes g_s(\measure{s}{\phi})$ otherwise.

\item $RevBranch_k[\GG_k](\qubit{\phi})= \qubit{\phi}$ if $\ell(\qubit{\phi})< k$ and $RevBranch_k[\GG_k](\qubit{\phi})= \sum_{s:|s|=k}g_s\left(\sum_{u:|u|=n-k}\measure{us}{\phi}\otimes \qubit{u}\right) \otimes\qubit{s}$ otherwise, where $n=\ell(\qubit{\phi})$.
\end{enumerate}
\end{lemma}

The difference between $REMOVE_k$ and $REP_k$ is subtle but $REMOVE_k$ moves the first $k$ qubits of an input to the end, whereas $REP_k$ moves
the last $k$ qubits to the front.
For basic qustrings of length $4$, for instance, it holds that $REP_{1}(\qubit{a_1a_2a_3a_4}) = \qubit{a_4a_1a_2a_3}$ and $REMOVE_1(\qubit{a_1a_2a_3a_4})
= \qubit{a_2a_3a_4a_1}$.
Similarly, $Branch_k[\GG_k]$ applies each $g_s$ in $\GG_k$ to the quantum state obtained from $\qubit{\phi}$ by eliminating the first $k$ qubits
whereas
$RevBranch_k[\GG_k]$ applies $g_s$ to
the quantum state obtained by eliminating the last $k$ qubits, whenever  $k\leq \ell(\qubit{\phi})$. For example, if $\GG_k=\{h\}_{s\in\{0,1\}^k}$ for a single quantum function $h$, then, for every string $s\in\{0,1\}^k$, we obtain $Branch_k[\GG_k](\qubit{s}\qubit{\phi}) = \qubit{s}\otimes h(\qubit{\phi})$ and $RevBranch_k[\GG_k](\qubit{\phi}\qubit{s}) = h(\qubit{\phi}) \otimes \qubit{s}$.

\vs{-2}
\begin{proofof}{Lemma \ref{lemma:special}}\hs{2}
Let $k\in\nat^{+}$, $g\in\hatsquareqp$, and $\GG_k=\{g_s\mid s\in\{0,1\}^k\}\subseteq \hatsquareqp$. For each index $i\in[k]$, let $s_i$ denote lexicographically the $i$th element of $\{0,1\}^k$.

1) We first  set $f_{2^k}=g_{s_{2^k}}$ and inductively define $f_{i-1}= Compo[g_{s_{i-1}},f_i]$ for every index $i\in[2,2^k]_{\integer}$ to obtain $f_1 = Compo[\GG_k]$. The resulted quantum function $f_1$ clearly belongs to $\hatsquareqp$ because $k$ is a fixed constant independent of inputs to $f_1$.

2) This is a special case of the single-qubit quantum recursion rule (or the $1$-qubit quantum recursion rule) in which $t=k-1$, $p=I$, and $f_0=f_1=I$. Therefore, $Switch_k[g,h]$ belongs to $\hatsquareqp$.

3) Since $LENGTH_k[g] = Switch_{k}[I,g]$, $LENGTH_k$ is in $\hatsquareqp$.

4) We begin with the case of $k=1$. The desired quantum function $REMOVE_1$ is defined as
\[
REMOVE_1(\qubit{\phi}) = \left\{
\begin{array}{ll}
\qubit{\phi} & \mbox{if $\ell(\qubit{\phi})\leq 1$,} \\
\sum_{a\in\{0,1\}} \qubit{a}\otimes REMOVE_1(\measure{a}{\psi_{SWAP,\phi}}) & \mbox{otherwise,}
\end{array} \right.
\]
More formally, we set $REMOVE_1 = QRec_1[I,I,SWAP|REMOVE_1,REMOVE_1]$.
In the case of $k=2$, we define $REMOVE_2=REMOVE_1\circ REMOVE_1$. For each index $k\geq3$, $REMOVE_k$ is obtained as follows. Let $h'_k$ be the $k$ compositions of $REMOVE_1$ and define $REMOVE_k$ to be $LENGTH_{k}[h'_k]$. We remark that $LENGTH_{k}$ is necessary in our definition because $h'_k$ is not guaranteed to equal $REMOVE_k$ when $\ell(\qubit{\phi})\leq k-1$.

5) First, we define $REP_1$ as $REP_1 =QRec_1[I,SWAP,I|REP_1,REP_1]$, that is,
\[
REP_1(\qubit{\phi}) = \left\{
\begin{array}{ll}
 \qubit{\phi} & \mbox{if $\ell(\qubit{\phi})\leq 1$,} \\
\sum_{a\in\{0,1\}} SWAP(\qubit{a}\otimes
REP_1(\measure{a}{\phi})) & \mbox{otherwise.}
\end{array} \right.
\]
Clearly, $REP_1$ belongs to $\hatsquareqp$.  For a general index
$k>1$, we define $REP_k$ in the following way. We first set $h'_k$ to be the $k$ compositions of $REP_1$. Finally, we set $REP_k = LENGTH_{k}[h'_k]$.

6) We first realize a quantum function $SWAP_{i,i+j}$, which swaps between the $i$th and the $(i+j)$th qubits of any input.
This goal is achieved by inductively constructing $SWAP_{i,i+j}$ in the following way.  Initially, we set $SWAP_{i,i+1}= REP_{i-1}\circ SWAP \circ REMOVE_{i-1}$. For any index $j\in[k-i]$, we define $SWAP_{i,i+j} = SWAP_{i+j-1,i+j}\circ
SWAP_{i,i+j-1}\circ  SWAP_{i+j-1,i+j}$. We then set $g = SWAP_{k,2k}\circ SWAP_{k-1,2k-1}\circ \cdots \circ SWAP_{2,k+2}\circ SWAP_{1,k+1}$. At last, it suffices to define $SWAP_{k}$ to be $LENGTH_{2k}[g]$.

7) The desired quantum function $REVERSE$ can be defined as $REVERSE = QRec_1[I,REMOVE_1,I|REVERSE,REVERSE]$, namely,
\[
REVERSE(\qubit{\phi}) = \left\{
\begin{array}{ll}
\qubit{\phi} &
\mbox{if $\ell(\qubit{\phi})\leq 1$,} \\
REMOVE_1( \sum_{a\in\{0,1\}} (\qubit{a}\otimes
REVERSE(\measure{a}{\phi}))) & \mbox{otherwise.}
\end{array} \right.
\]

8) Note that, when $k=2$,  $Branch_k[\GG_k]$ coincides with $Branch$ and it thus belongs to $\hatsquareqp$. Hereafter, we assume that $k\geq3$.
For each string $s\in\{0,1\}^k$, let $g_s^{(0)}=g_s$. For each index $i\in\nat$ with $i<k$ and each string $s\in\{0,1\}^*$ with $|s|=k-i-1$, we inductively define $g_s^{(i+1)}$ to be $Branch[g_{s0}^{(i)},g_{s1}^{(i)}]$, that is,
\[
 g^{(i+1)}_s(\qubit{\phi}) = \left\{
 \begin{array}{ll}
\qubit{\phi} & \mbox{if $\ell(\qubit{\phi}\leq 1$,} \\
\qubit{0}\otimes
g^{(i)}_{s0}(\measure{0}{\phi}) + \qubit{1}\otimes
g^{(i)}_{s1}(\measure{1}{\phi}) & \mbox{otherwise.}
\end{array}  \right.
\]
Finally, we set $Branch_k[\GG_k]= g^{(k)}_{\lambda}$, where $\lambda$ is the empty string.

9) Assuming $k\geq2$, let $RevBranch_k[\{g_s\}_{s\in\{0,1\}^k}]=
REMOVE_k\circ Branch_k[\{g_s\}_{s\in\{0,1\}^k}] \circ REP_k$.
\end{proofof}


The next lemma shows that we can extend any classical bijection on $\{0,1\}^k$ to its associated $\hatsquareqp$-function, which behaves  in exactly the same way as the
bijection does on the first $k$ bits of its input.

\begin{lemma}\label{bijection-function}
Let $k$ be a constant in $\nat^{+}$. For any bijection $f$ from $\{0,1\}^k$ to
$\{0,1\}^k$, there exists a $\squareqp$-function $g_f$ such that, for any quantum state $\qubit{\phi}\in\HH_{\infty}$,
\[
g_f(\qubit{\phi})= \left\{
\begin{array}{ll}
\qubit{\phi} & \mbox{if $\ell(\qubit{\phi})\leq k-1$,} \\
\sum_{s\in
\{0,1\}^k}  \qubit{f(s)}\!\measure{s}{\phi} & \mbox{otherwise.}
\end{array} \right.
\]
\end{lemma}

\begin{proof}
Given a bijection $f$ on $\{0,1\}^k$, it suffices to show the existence of $\hatsquareqp$-function $h$ satisfying $h(\qubit{s}\qubit{\phi}) = \qubit{f(s)}\qubit{\phi}$ for any string $s\in\{0,1\}^k$ and any quantum state $\qubit{\phi}\in\HH_{\infty}$ since $g_f$ is obtained from $h$ simply by setting $g_f=LENGTH_{k}[h]$. Notice that, if $h\in\hatsquareqp$, $h(\bm{0})=\bm{0}$ makes $g_f(\bm{0})$ equal $\bm{0}$.

A bijection on $\{0,1\}^k$ is, in essence, a \emph{permutation} on $\{s_1,s_2,\ldots,s_{2^k}\}$, where each $s_i$ is lexicographically the $i$th string in $\{0,1\}^k$, and thus it can be expressed as the multiplication of a finite number of \emph{transpositions}, each of which
swaps between two distinct numbers. This can be done by an application of the multiple composition of the corresponding $SWAP_{i,i+j}$, which is defined in the proof of Lemma \ref{lemma:special}(6).
Therefore, $h$ belongs to $\hatsquareqp$.
\end{proof}


Given a quantum state $\qubit{\phi}$, it is possible to apply simultaneously a quantum function $f$ to the first $k$ qubits of $\qubit{\phi}$ and another quantum function $g$ to the rest.

\begin{lemma}\label{two-diff-function}
Given $f,g\in\hatsquareqp$ and $k\in\nat^{+}$, the quantum function $f^{\leq k}\otimes g$, which is defined by
\[
(f^{\leq k}\otimes g)(\qubit{\phi})= \left\{
\begin{array}{ll}
f(\qubit{\phi}) & \mbox{if $\ell(\qubit{\phi})\leq k$,} \\
\sum_{s\in
\{0,1\}^k} f( \qubit{s}) \otimes g(\measure{s}{\phi}) & \mbox{otherwise,}
\end{array} \right.
\]
belongs to $\hatsquareqp$. We write this $g$ as $Skip[f]$.
\end{lemma}

\begin{proof}
Given a quantum function $f$, we restrict the application of $f$ to the last $k$ qubits of any input $\qubit{\phi}$ by setting $f' = QRec_k[f,I,I|f_0,f_1]$, where $f_0$ and $f_1$ are $f'$. It follows that $f'(\qubit{\phi}\qubit{s}) = \qubit{\phi}\otimes f(\qubit{s})$ for any $\qubit{\phi}\in\HH_{\infty}$ and $s\in\{0,1\}^k$.

Let $h= REP_k\circ f' \circ REMOVE_k$. This quantum function $h$ satisfies that $h(\qubit{s}\qubit{\phi}) = f(\qubit{s})\otimes \qubit{\phi}$ for any $s\in\{0,1\}^k$. We then define $\GG_k =\{g_s\}_{s\in\{0,1\}^k}$ with $g_s=g$ for every string $s\in\{0,1\}^k$ and set $g'= h\circ Branch_k[\GG_k]$. It then follows that the desired quantum function $f^{\leq k}\otimes g$ equals $Switch_{k+1}[f,g']$.  \end{proof}


Next, we present Lemma \ref{lemma:add-zero}, which is useful for the proof of our key lemma (Lemma \ref{lemma:second-part}) in Section \ref{sec:second-part}. The lemma allows us to skip, before applying a given quantum function, an arbitrary number of $0$s until we read a fixed number of $1$s.

\begin{lemma}\label{lemma:add-zero}
Let $f$ be a quantum function in $\hatsquareqp$ and let $k$ be a constant in $\nat^{+}$. There exists a quantum
function $g$ in $\hatsquareqp$ such that $g(\qubit{0^m1^k}\otimes
\qubit{\phi}) =
\qubit{0^m1^k}\otimes f(\qubit{\phi})$ and $g(\qubit{0^{m+1}}) = \qubit{0^{m+1}}$
for any number $m\in\nat$ and any quantum state $\qubit{\phi}\in\HH_{\infty}$. The lemma also holds when $\hatsquareqp$ is replaced by $\squareqp$.
\end{lemma}

\begin{proof}
Let $k\geq2$. Given a quantum function $f\in\hatsquareqp$, we first expand $f$ to $f'$ so that $f'(\qubit{1^{k-1}}\qubit{\phi}) =
\qubit{1^{k-1}}\otimes f(\qubit{\phi})$ for any $\qubit{\phi}\in\HH_{\infty}$ and $f'(\qubit{0^{m+1}}) = \qubit{0^{m+1}}$ for any $m\in\nat$. This quantum function $f'$ can be obtained inductively as follows.
We set $f_{k-1}= Branch[I,f]$, $f_i=Branch[I,f_{i+1}]$ for each $i\in[k-2]$, and finally define $f'$ to be $f_1$.
When $k=1$, we simply set $f'=f$.
The desired quantum function $g$ in the lemma must satisfy
\[
g(\qubit{\phi}) = \left\{
\begin{array}{ll}
\qubit{\phi} &
\mbox{if $\ell(\qubit{\phi})\leq 1$,} \\
 \qubit{0}\otimes f'( g(\measure{0}{\phi}) ) +
\qubit{1}\!\measure{1}{\phi}  & \mbox{otherwise.}
\end{array} \right.
\]
This $g$ is formally defined as $g=QRec_1[I,Branch[f',I],I|g,I]$.
This completes the proof.
\end{proof}


Within our framework, it is possible to construct a ``restricted'' form of the \emph{quantum Fourier transform} (QFT). Given a binary string $s=s_1s_2\cdots s_k$ of length $k$ with $s_i\in\{0,1\}$, we denote by $num(s)$ the integer of the form $\sum_{i=1}^{k}s_i2^{k-i}$. For instance,  $num(011) = 1\cdot 2^1 + 1\cdot 2^0 = 5$ and $num(1010) = 1\cdot 2^3 + 1\cdot 2^1 = 10$.
Moreover, let $\omega_{k} = e^{2\pi i/2^k}$, where $i=\sqrt{-1}$.

\begin{lemma}\label{Fourier-transform}
Let $k$ be any fixed constant in $\nat^{+}$. The following $k$-qubit quantum Fourier transform belongs to $\hatsquareqp$. For any element $\qubit{\phi}$ in $\HH_{\infty}$, let
\[
F_k(\qubit{\phi}) = \left\{
    \begin{array}{ll}
    \qubit{\phi} & \mbox{if $\ell(\qubit{\phi})<k$,} \\
    \frac{1}{2^{k/2}}\sum_{t:|t|=k}\sum_{s:|s|=k} \omega_{k}^{num(s) num(t)}  \qubit{s}\! \measure{t}{\phi} & \mbox{otherwise.}
    \end{array} \right.
\]
\end{lemma}

\begin{proof}
When $k=1$, $F_1$ coincides with $WH$ and therefore $F_1$ belongs to $\hatsquareqp$ by Lemma \ref{simple-function}. Next, assume that $k\geq2$.
It is known that, for any $x_1,x_2,\ldots,x_k\in\{0,1\}$,
\begin{equation}\label{eqn:QFT}
F_k(\qubit{x_1x_2\cdots x_k}) = \frac{1}{2^{k/2}}(\qubit{0}+\omega_1^{x_k}\qubit{1}) (\qubit{0}+\omega_1^{x_{k-1}}\omega_2^{x_k}\qubit{1})\cdots (\qubit{0}+\prod_{i=1}^{k}\omega_i^{x_{i}}\qubit{1}).
\end{equation}
For this fact and its proof, refer to, e.g., \cite{NC00}.

Let us recall the special quantum function $SWAP_{i,i+j}$ from the proof of Lemma \ref{lemma:special}(6), which swaps between the $i$th and the $j$th qubits.
Using $CPHASE_{\theta}$, for any index pair $i,j$ with $i<j$, we define $CPHASE_{\theta}^{(i,j)}$ to be $SWAP_{1,i}\circ SWAP_{2,j}\circ CPHASE_{\theta} \circ SWAP_{2,j}\circ SWAP_{1,i}$, in which we apply  $CPHASE_{\theta}$ to the $i$th and the $j$th qubits.
We first want to construct $G_k^{(1)}  = F_k\circ REVERSE$, which works similarly to $F_k$ but takes $\qubit{\phi^R}$ as an input instead. To achieve this goal, we define $\{G_{j}^{(i)}\}_{i,j\in\nat^{+}}$ inductively as follows.
Initially, we set $G_0^{(i)} = I$ for any $i\in\nat^{+}$.
Next, we define $G_1^{(i)}$ as $G_1^{(i)} = H$ if $i=1$, and $G_1^{(i)}= REP_{i-1}\circ H\circ REMOVE_{i-1}$ otherwise.
For any index $k\geq2$, $G_k^{(i)}$ is defined to be $G_{k-1}^{(i)}\circ CPHASE_{\frac{\pi}{2^{k-1}}}^{(i,i+k-1)} \circ (G_{k-2}^{(i)})^{-1} \circ G_{k-1}^{(i+1)}$.
It is not difficult to show that $G_k^{(1)}$ coincides with $F_k\circ REVERSE$ by Eqn.~(\ref{eqn:QFT}).

Since $G^{(1)}_{k} = F_k\circ REVERSE$, it suffices to define $F_k$ as  $G_k^{(1)}\circ REVERSE$.
\end{proof}

A general form of QFT, in which $k$ is not limited to a particular constant, will be discussed in Section \ref{sec:open-question} in connection to our choice of Schemata I--IV that form the function class $\squareqp$.

\section{Main Contributions}\label{sec:main-contributions}

In Section \ref{sec:definition}, we have introduced the $\squareqp$-functions and the $\hatsquareqp$-functions mapping $\HH_{\infty}$ to $\HH_{\infty}$ by applying Schemata I--IV for finitely many times.
Our main theorem (Theorem \ref{theorem:character}) asserts that $\squareqp$ can precisely characterize all functions in $\fbqp$ mapping $\{0,1\}^*$ to $\{0,1\}^*$, and therefore characterize all languages in $\bqp$ over $\{0,1\}$ by identifying languages with their corresponding characteristic functions. This theorem will be proven by using two key lemmas, Lemmas \ref{lemma:first-part}--\ref{lemma:second-part}.

\subsection{A New Characterization of FBQP}\label{sec:main-theorem}

Our goal is to demonstrate the power of $\squareqp$-functions (and thus  $\hatsquareqp$-functions) by showing in Theorem \ref{theorem:character} that $\squareqp$-functions (as well as $\hatsquareqp$-functions) precisely characterize
$\fbqp$-functions on $\{0,1\}^*$.
For this purpose, there are unfortunately two major difficulties to overcome.

The first difficulty arises in dealing with tape symbols of QTMs by qubits alone. Notice that QTMs working over non-binary input alphabets are known to be simulated by QTMs taking the binary input alphabet $\{0,1\}$.
Even if we successfully reduce the size of input alphabets down to $2$, in fact, the simulation of such QTMs must require the proper handling of  \emph{non-binary tape symbols}, in particular,
the distinguished blank symbol. For our later convenience, we use ``$b$'' to denote the blank symbol, instead of $\#$.
In a similar vein, suppose that the outcome of a $\squareqp$-function $f$ is composed of an important,  meaningful portion and the other ``garbage'' portion, which is a remnant of the computation process. Since we use only qubits ($\qubit{0}$ and $\qubit{1}$) to express inputs and outputs, how can we distinguish between the meaningful portion and the garbage portion?
In order to simulate QTMs by $\squareqp$-functions, we therefore need to ``encode'' all tape symbols into qubits.

This paper introduces the following simple coding scheme. We set $\hat{0}=00$, $\hat{1}=01$, and $\hat{b}=10$.
We also set $\hat{2}=11$ and $\hat{3}=10$ for a later use.
The input alphabet $\Sigma=\{0,1\}$ is thus translated into
$\{\hat{0},\hat{1}\}$ and the tape alphabet $\Gamma=\Sigma\cup\{b\}$
is encoded into $\{\hat{0},\hat{1},\hat{b}\}$.
Given each binary string
$s=s_1s_2\cdots s_n$ with $s_i\in\{0,1\}$ for every index $i\in[n]$, a \emph{code}  $\tilde{s}$ of $s$ indicates the string $\hat{s}_1\hat{s}_2\cdots \hat{s}_n\hat{2}$, where the last item $\hat{2}$ serves as an endmarker, which marks the end of the code.
It then follows that  $|\tilde{s}| =
2|s| +2$. As quick examples, we obtain $\widetilde{0110} = \hat{0}\hat{1}\hat{1}\hat{0}\hat{2} = 0001010011$ and $\widetilde{01}\cdot \widetilde{11} = \hat{0}\hat{1}\hat{2}\hat{1}\hat{1}\hat{2} = 000111010111$. Later, we will show in lemma \ref{encode-decode} that the encoding of strings and the decoding of encoded strings can be carried out by suitable $\hatsquareqp$-functions.

Another difficulty comes from the inability of $\hatsquareqp$-functions to expand their qubits. Notice that a QTM is designed to freely use additional storage space by moving its tape head simply to new, blank tape cells beyond its \emph{input area}, in which an input is initially written. To simulate such a QTM, we need to simulate its entire activities made in the freely expanding tape space as well.
On the contrary, every $\hatsquareqp$-function is dimension-preserving by Lemma \ref{hatsquare-property}(4), and thus the number of input qubits must match that of output qubits. If we want to simulate extra storage space of the QTM, then we need to feed the same amount of extra qubits to the target $\hatsquareqp$-function for use at the very beginning.

We resolve this second issue by extending each input of quantum functions by adding extra $0$s whose length is associated with the running time of the QTM. For any polynomial $p$ and any quantum function $g$ on $\HH_{\infty}$, we define $\qubit{\phi^p(x)} =
\qubit{0^{|x|}1} \qubit{0^{p(|x|)}10^{11p(|x|)+6}1} \qubit{x}$ and $\qubit{\phi_g^p(x)} =
g(\qubit{\phi^p(x)})$ for every string $x\in\{0,1\}^*$. Similarly, for any function $f$ on $\{0,1\}^*$, we set $\qubit{\phi^{p,f}(x)} = \qubit{\widetilde{0^{|f(x)|}}} \qubit{0^{|f(x)|+1}1} \qubit{\phi^p(x)}$ and $\qubit{\phi^{p,f}_g(x)} = g(\qubit{\phi^{p,f}(x)})$.


Any $\fbqp$-function takes classical input strings and produces classical strings that are outcomes of a polynomial-time quantum Turing machine with high probability. In contrast, our $\squareqp$-functions $f$ are to transform each quantum state in $\HH_{\infty}$ to another one in $\HH_{\infty}$. To obtain classical output strings, we need to observe the outcomes of $f$.

The main theorem (Theorem 4.1) roughly asserts the following: for any $\fbqp$-function $f$, there always exists a $\hatsquareqp$-function $g$ such that, when we observe the first $|f(x)|$ qubits of the outcome  $g(\qubit{\phi^p(x)})$ ($=\qubit{\phi^p_g(x)}$) of $g$ on input $\qubit{\phi^p(x)}$, we correctly obtain $f(x)$ with high probability.
Notice that $g(\qubit{\phi^p(x)})$ also contains extra qubits, called ``garbage'' qubits, which are left unobserved in the process of calculating  $f(x)$. Those garbage qubits are actually the remnant of the computation process of $f$.
It is, however, possible to remove those qubits by partly reversing the whole computation, if we know the output size $|f(x)|$ ahead of the computation (by expanding $\qubit{\phi^p(x)}$ to $\qubit{\phi^{p,f}(x)}$ and $\qubit{\phi^p_g(x)}$ to $\qubit{\phi^{p,f}_g(x)}$).

\begin{theorem}[Main Theorem]\label{theorem:character}
Let $f$ be a polynomially-bounded function on $\{0,1\}^*$.  The following three statements are logically equivalent.
\begin{enumerate}\vs{-1}
  \setlength{\topsep}{-2mm}%
  \setlength{\itemsep}{1mm}%
  \setlength{\parskip}{0cm}%

\item The function $f$ is in $\fbqp$.

\item For any constant $\varepsilon\in[0,1/2)$, there exist a quantum function $g$ in $\hatsquareqp$ and a polynomial $p$ such that $|f(x)|\leq p(|x|)$ and
$\|\measure{f(x)}{\phi_g^p(x)}\|^2 \geq 1-\varepsilon$ for all $x\in\{0,1\}^*$.

\item For any constant $\varepsilon\in[0,1/2)$, there exist a quantum function $g$ in $\hatsquareqp$ and a
polynomial $p$ such that $|f(x)|\leq p(|x|)$ and
$|\measure{\Psi_{f(x)}}{\phi_g^{p,f}(x)}|^2 \geq 1-\varepsilon$ for all $x\in\{0,1\}^*$, where $\qubit{\Psi_{f(x)}} = \qubit{f(x)}\qubit{(0^{|f(x)|+1}1)^2} \qubit{\phi^p(x)}$.
\end{enumerate}
\end{theorem}

In Statements 2--3 of Theorem \ref{theorem:character},  $\measure{f(x)}{\phi_g^p(x)}$ is a non-null vector whereas  $\measure{\Psi_{f(x)}}{\phi_g^{p,f}(x)}$ is just a scalar because $\ell(\qubit{f(x)}\qubit{\phi^p(x)})= \ell(\qubit{\phi_g^{p,f}(x)})$.

Hereafter, we wish to prove the main theorem, Theorem \ref{theorem:character}. For a strategic reason, we split the theorem into two technical lemmas, Lemmas \ref{lemma:first-part} and \ref{lemma:second-part}.
To present the lemmas, we need to introduce additional terminology for QTMs.
It is easier in practice to design multi-tape well-formed QTMs rather than  single-tape ones. However, since multi-tape QTMs can be simulated by single-tape QTMs by translating multiple tapes to multiple tracks of a single tape, toward the proof of Theorem \ref{theorem:character}, it suffices to focus our attention only on single-tape QTMs.

A single-tape QTM is said to be in \emph{normal form} if  $\delta(q_f,\sigma)=\qubit{q_0}\qubit{\sigma}\qubit{R}$  holds for any tape symbol $\sigma\in\Gamma$.
If a QTM halts in a superposition of final configurations in which a tape head returns to the start cell, then we call such a machine \emph{stationary}. Refer to \cite{BV97} for their basic properties.
For convenience, we further call a QTM {\em conservative} if it is
well-formed, stationary, and in normal form.
Moreover, a well-formed QTM is said to be {\em plain} if its transition
function satisfies the following specific requirement: for every pair  $(p,\sigma)\in
Q\times\Gamma$, $\delta(p,\sigma)$ has the form either
$\delta(p,\sigma) =e^{i\theta}\qubit{q,\tau,d}$ or $\delta(p,\sigma) =
\cos\theta\qubit{q,\tau,d} + \sin\theta\qubit{q',\tau,',d'}$ for certain $\theta\in[0,2\pi)$ and two distinct tuples $(q,\tau,d)$ and $(q',\tau',d')$.
Bernstein and Vazirani \cite{BV97} claimed
that any single-tape, polynomial-time,
conservative QTM $M$ can be simulated by an appropriate
single-tape, polynomial-time, conservative, plain QTM $M'$. Additionally, when  $M$ is of $\tilde{\complex}$-amplitudes, so is $M'$.

Let us recall that any output string of a QTM begins at the start cell and stretches to the right until the first blank symbol although there may be tape symbols left unerased in other parts of the tape. For our later convenience, a QTM $M$ is said to have \emph{clean outputs} if,  when $M$ halts, no non-blank symbol appears in the left-side region of the output string, i.e., the region consisting of all cells indexed by negative integers.
Take a polynomial $p$ that bounds the running time of $M$ on every input.
Since $M$ halts in at most $p(|x|)$ steps on every instance $x\in\{0,1\}^*$,
it suffices for us to pay attention only to its \emph{essential tape region} that covers all tape cells indexed by numbers between $-p(|x|)$ and $+p(|x|)$.
In practice, we redefine a configuration $\gamma$ of $M$ on input $x$ of length $n$ as a triplet $(q,h,\sigma_1\cdots\sigma_{2p(n)+1})$, where $p\in Q$, $h\in\integer$ with $-p(n)\leq h\leq p(n)$, and $\sigma_1,\cdots, \sigma_{2p(n)+1}\in\{0,1,b\}$ such that, for every index $i\in[2p(n)+1]$,  $\sigma_i$ is a tape symbol written at the cell indexed by $i-p(n)-1$.
Note that the start cell comes in the middle of $\sigma_1\cdots\sigma_{2p(n)+1}$. For notational convenience, we further modify the notion of configuration by splitting the essential tape region into two parts $(z_1,z_2)$, in which $z_1=\sigma_1\sigma_2\cdots \sigma_{p(n)}$ refers to the left-side region of the start cell, not including the start cell, and $z_2=\sigma_{p(n)+1}\sigma_{p(n)+2}\cdots \sigma_{2p(n)+1}$ refers to the rest of the essential tape region.
This provides us with a modified configuration of the form  $(q,h,z_1z_2)$. For a practical reason, we further alter it into $(z_2,z_1,h,q)$, which we call by a \emph{skew configuration}.
Associated with this alteration of configurations, we also modify the original time-evolution operator, $U_{\delta}$, of $M$ so that it works on skew configurations. Be aware that this new operator cannot be realized by the standard QTM model. To distinguish it from the original time-evaluation operator $U_{\delta}$ of $M$, we call it the \emph{skew time-evolution operator} and write it as $\widehat{U}_{\delta}$.

\begin{lemma}\label{lemma:first-part}
Let $f$ be any quantum function in $\squareqp$.
There exist a polynomial $p$
and a single-tape, conservative, plain QTM $M$ producing clean outputs with $\appcomplex$-amplitudes  such
that, for any quantum state $\qubit{\phi}$ in $\HH_{2^n}$, $M$ starts with a non-null quantum state $\qubit{\phi}$ given on its input/work tape and, when it halts after $p(n)$ steps, the superposition $\qubit{\eta_{M,\phi}}$ of skew final configurations of $M$ on $\qubit{\phi}$ is of the form  $f(\qubit{\phi})\otimes \qubit{0,q_f}$, where $f(\qubit{\phi})$ appears as the content of the tape from the start cell to the right until the first blank symbol and $q_f$ is a unique final inner state of $M$.
\end{lemma}

\begin{proof}
We first show that all the initial functions in Scheme I can be exactly computed in polynomial time on appropriate single-tape, $\appcomplex$-amplitude, conservative, plain QTMs having clean outputs over input/output alphabet $\{0,1\}$.
For clarity, our goal here is to demonstrate that, for every quantum function $f$ in Scheme I, there exists a QTM with the lemma's properties such that a superposition $\qubit{\eta_{M,\phi}}$ of $M$'s skew final  configurations is of the form $f(\qubit{\phi})\otimes \qubit{0,q_f}$, where $f(\qubit{\phi})$ is the content of $M$'s tape from the start cell to the right until the first blank symbol.

Since $I$ (identity) is easy to simulate, let us  consider $PHASE_{\theta}$. In this case, we take a QTM that applies a  transition of the form $\delta(q_0,\sigma) = e^{i\theta \sigma}\qubit{q_f,\sigma,N}$
for any bit $\sigma\in\{0,1\}$. Clearly, if we start with a skew initial  configuration $\qubit{\phi}\qubit{0}\qubit{q_0}$, then we halt with $PHASE_{\theta}(\qubit{\phi})\otimes \qubit{0}\qubit{q_f}$; in short,  this QTM ``exactly computes''  $PHASE_{\theta}$.
For $ROT_{\theta}$, we use the transition defined by $\delta(q_0,\sigma) =\cos\theta\qubit{q_f,\sigma,N} + (-1)^{\sigma}\sin\theta \qubit{q_f,\bar{\sigma},N}$ to exactly compute $ROT_{\theta}$.
To simulate $NOT$, we then define a QTM to have a transition of $\delta(q_0,\sigma) = \qubit{q_f,\overline{\sigma},N}$.
In the case of $SWAP$, it suffices to prepare inner states $p_{\sigma_1}$ and $r_{\sigma_2}$ as well as the following transitions: $\delta(q_0,\sigma_1) = \qubit{p_{\sigma_1},\#,R}$, $\delta(p_{\sigma_1},\sigma_2) = \qubit{r_{\sigma_2},\sigma_1,L}$, $\delta(r_{\sigma_2},\#) = \qubit{q_f,\sigma_2,N}$ for bits $\sigma_1,\sigma_2\in\{0,1\}$.
Concerning $MEAS[a](\qubit{\phi})$, we start with checking the first qubit of $\qubit{\phi}$. If it is not $a$, then we make a QTM \emph{reject} the input; otherwise, we do nothing. More formally, we define $\delta(q_0,\overline{a})=\qubit{q_{rej},\overline{a},N}$ and $\delta(q_0,a) = \qubit{q_f,a,N}$. It is not difficult to see that $\qubit{\eta_{M,\phi}}$ has the form $f(\qubit{\phi})\otimes \qubit{0,q_0}$.

Next, we intend to simulate each of the construction rules on an appropriate QTM. By induction hypothesis, there exist three polynomial-time, single-tape, $\appcomplex$-amplitude, conservative, plain QTMs $M_g$, $M_h$, and $M_p$ that satisfy the lemma for $g$, $h$,
and $p$, respectively.  By installing an \emph{internal clock} in an appropriate way with a certain polynomial $r$, we can make $M_g$, $M_h$, and $M_p$ halt in  exactly $r(n)$ time on any input of length $n\in\nat$. In what follows, let $\qubit{\phi}$ denote any input in $\HH_{\infty}$.

[composition]
Consider the case of $f=Compo[g,h]$. We compute $f$ as follows. We first run $M_h$ on $\qubit{\phi}$ and obtain a superposition of skew final  configurations, say, $h(\qubit{\phi})\otimes \qubit{0,q'_f}$ for a unique final inner state $q'_f$ of $M_h$. Since $M_h$ is stationary and in normal form, we can further run $M_g$ on the resulted quantum state $h(\qubit{\phi})$, treating $q'_f$ as its new initial inner state and  generating $\qubit{\eta_{M,\phi}} = g(h(\qubit{\phi}))\otimes \qubit{0,q_f}$ for a unique final inner state $q_f$ of $M_g$.

[branching]
Assume that $f=Branch[g,h]$. We check the first qubit of $\qubit{\phi}$. If it is $0$, then we run $M_g$ on $\measure{0}{\phi}$; otherwise, we run $M_h$ on $\measure{1}{\phi}$.
This produces $\qubit{0}\otimes g(\measure{0}{\phi})\otimes \qubit{0,q_{h,f}} + \qubit{1}\otimes h(\measure{1}{\phi})\otimes \qubit{0,q_{g,f}}$, where $q_{g,f}$ and $q_{h,f}$ are respectively unique final inner states of $M_g$ and $M_h$.
Notice that we do not need to check whether $\ell(\qubit{\phi})\geq2$.
Formally, we add the following transitions to those of $M_g$ and $M_h$: $\delta(q_0,0)=\qubit{q_{0},\#_g,R}$ and $\delta(q_0,1)=\qubit{q_{0},\#_h,R}$, where $\#_g$ and $\#_h$ are designated symbols marking the left of the new ``start cells'' for $M_g$ and $M_h$, respectively.  At terminating, the tape head moves back to the start cell. We then replace $\#_g$ and $\#_h$ respectively with $0$ and $1$ by entering a new final inner state $q_f$; namely, $\delta(q_{g,f},\#_g) = \qubit{q_f,0,N}$ and $\delta(q_{h,f},\#_h) = \qubit{q_f,1,N}$. This transition is  unitarily possible.

[multi-qubit quantum recursion]
Finally, let us demonstrate a simulation of the multi-qubit quantum recursion introduced by $f= kQRec_t[g,h,p|\FF_k]$ with $\FF_k=\{f_s\}_{s\in\{0,1\}^k}$. For readability, let
us consider only the simplest case where $k=1$, $f_0=f$, and $f_1=I$. The other cases can be similarly treated.
Consider the ``multi-tape'' QTM $M_f$ that roughly behaves as described below.  Let $T$ be a sufficiently large positive constant.
\begin{quote}
(1) In this initial phase, starting with input $\qubit{\phi}$, we prepare a counter in a new work tape and an internal clock in another work tape. We use the clock to adjust the terminating timing of all  computation paths.
Count the number $\ell(\qubit{\phi})$ of qubits simply by incrementing the counter as moving the input tape head from the
start cell to the right. Initially, we set the current quantum state, say,  $\qubit{\xi}$ expressed on the tape to be $\qubit{\phi}$ and set the current counter $k$ to be $\ell(\qubit{\xi})$ ($=n$). Go to the splitting phase.

(2) In this splitting phase, we inductively perform the following procedure using the clock. (*)  Assume that the input tape currently contains a quantum state $\qubit{\xi}$ and the counter has $k=\ell(\qubit{\xi})$.
If $k\leq t$, then idle until the clock hits $T$ and then go to the processing phase.
Otherwise, run $M_p$ on $\qubit{\xi}$ to generate $\qubit{\psi_{p,\xi}}$ and observe the first qubit of $\qubit{\psi_{p,\xi}}$ in the computational basis, obtaining $\measure{b}{\psi_{p,\xi}}$ for each $b\in\{0,1\}$.
If $b$ is $1$, then run $M_h$ on $\qubit{1}\! \measure{1}{\psi_{p,\xi}}$, obtain $h(\qubit{1}\! \measure{1}{\psi_{p,\xi}})$, which is viewed as  $f(\qubit{1}\! \measure{1}{\xi})$, idle until the clock hits $T$, and then start the processing phase. On the contrary, when $b$ is $0$, move this bit $0$ to a separate tape to remember and then update both $\qubit{\xi}$ and $k$ to be $\measure{0}{\psi_{p,\xi}}$ and $k-1$, respectively.  Continue (*).

(3) In this processing phase, we start with a quantum state  $\qubit{\xi}$, which is produced in the splitting phase.
Let $k=\ell(\qubit{\xi})$. We inductively perform the following procedure. (**) If
$k\leq t$, then we run $M_g$ on the input $\qubit{\xi}$ and
produce $g(\qubit{\xi})$, which is viewed as $f(\qubit{\xi})$. Update $\qubit{\xi}$ to be the resulted quantum state.
Otherwise, we move back the last stored bit $0$ from the separate tape, run $M_h$ on $\qubit{0}\qubit{\xi}$, obtain $h( \qubit{0}\qubit{\xi})$, and then update $\qubit{\xi}$ to be the obtained quantum state and $k$ to be $k+1$  since $\ell(\qubit{0}\qubit{\xi})=k+1$. If all $b$'s are consumed (equivalently, $k=n$), then idle until the clock hits $2T$, output $\qubit{\xi}$, and halt. Otherwise, continue (**).
\end{quote}

The running time of the above QTM is bounded from above by a certain polynomial in the length $\ell(\qubit{\phi})$ because each of the procedures (*) and (**) is repeated for at most $\ell(\qubit{\phi})$ times.
Although $M_f$ stores bits on the separate tape, those bits are all moved back and used up by the end of the computation. This fact shows that a superposition of $M_f$'s skew final configurations is of the form $f(\qubit{\phi})\otimes \qubit{0,q_f}$ for a unique final inner state $q_f$. Finally, we convert this multi-tape QTM into a computationally  equivalent single-tape QTM of the desired properties in the lemma.

This completes the proof of Lemma \ref{lemma:first-part}.
\end{proof}


For notational sake, we write $M[r]$ for an output string written in a  skew final configuration $r$, which covers only an essential tape region of $M$. We write  $FSC_{M,n}$ to denote the set of all possible skew final  configurations of $M$ produced for any input of length $n$.

The key to the proof of Theorem \ref{theorem:character} is the following lemma, which ensures the existence of a $\hatsquareqp$-function $g$ whose outcome $g(\qubit{\phi^p(\qubit{\psi})})$ ($=\qubit{\phi^p_g(\qubit{\psi})}$) ``almost'' characterize the encoded skew final configuration $\sum_{x:|x|=n}\sum_{r\in FSC_{M,n}} \measure{x}{\psi} \qubit{\widetilde{M[r]}} \qubit{\xi_{x,r}}$ of $M$, where $\widetilde{M[r]}$ is the encoding of $M[r]$.
Here, the encoding $\widetilde{M[r]}$ is needed for the construction of $g$ in the proof of the lemma. However, as shown also in the proof, executing an extra decoding procedure for $\widetilde{M[r]}$ allows us to replace $\widetilde{M[r]}$ by $M[r]$.

\begin{lemma}[Key Lemma]\label{lemma:second-part}
Let $M$ be a single-tape, polynomial-time, $\appcomplex$-amplitude, conservative, plain QTM having clean outputs over input/output alphabet $\Sigma=\{0,1\}$ and tape alphabet $\Gamma=\{0,1,b\}$. Assume that, when  $M$ halts on input $\qubit{\psi}$ of length $n$, a superposition of coded skew final configurations is of the form $\sum_{x:|x|=n}\sum_{r\in FSC_{M,n}}\measure{x}{\psi} \qubit{\widetilde{M[r]}} \qubit{\xi_{x,r}}$,  where $\qubit{\xi_{x,r}}$ denotes an appropriate quantum state describing the rest of the coded skew final configurations other than $\widetilde{M[r]}$.
There exist a quantum function $g$ in $\hatsquareqp$
and a polynomial $p$ such that, for any number $n\in\nat^{+}$ and every quantum state $\qubit{\psi}\in\HH_{2^n}$,
$\qubit{\phi^p_g(\qubit{\psi})}$ has the form $\sum_{x:|x|=n} \sum_{r\in FSC_{M,n}} \measure{x}{\psi}  \qubit{\widetilde{M[r]}} \qubit{\widehat{\xi}_{x,r}}$ for certain quantum states $\{\qubit{\widehat{\xi}_{x,r}}\}_{x,r}$ satisfying  that, for any $x,x'\in\{0,1\}^n$ and $r,r'\in FSC_{M,n}$, (i) $\|\measure{\xi_{x,r}}{\xi_{x',r'}}\|= \|\measure{\widehat{\xi}_{x,r}}{\widehat{\xi}_{x',r'}}\|$ and (ii) $\measure{\widehat{\xi}_{x,r}}{\widehat{\xi}_{x,r'}}=0$ if $r\neq r'$.
Furthermore, it is possible to modify $g$ to $g'$, which satisfies  $\qubit{\phi^p_{g'}(\qubit{\phi})} = \sum_{x:|x|=n} \sum_{r\in FSC_{M,n}} \measure{x}{\psi} \qubit{M[r]} \qubit{\widehat{\xi}'_{x,r}}$ for appropriate quantum states $\{\qubit{\widehat{\xi}'_{x',r'}}\}_{x',r'}$ satisfying Conditions (i)--(ii).
\end{lemma}

The proof of Lemma  \ref{lemma:second-part} is lengthy and it is postponed until Section  \ref{sec:second-part}.
Meanwhile, we return to Theorem \ref{theorem:character} and present its proof using Lemmas \ref{lemma:first-part} and \ref{lemma:second-part}.

\begin{proofof}{Theorem \ref{theorem:character}}\hs{2}
Let $\varepsilon\in[0,1/2)$ be any constant and let $f$ be any polynomially-bounded function mapping $\{0,1\}^*$ to $\{0,1\}^*$.

(1 $\Rightarrow$ 2)
Assume that $f$ is in $\fbqp$. Take a
multi-tape, polynomial-time, $\appcomplex$-amplitude, well-formed QTM $N$ that computes $f$ with bounded-error probability.
Let us choose a
polynomial $p$ that bounds the running time of $N$ on every input.
It is possible to ``simulate'' $N$ with high success probability by a ceratin
single-tape, $\appcomplex$-amplitude, conservative, plain QTM, say, $M$ having clean outputs in such a way that the machine takes input $x$ and terminates with generating ${f(x)}bw$ in the right-side region of the start cell of the input/work tape with bounded-error probability, where $b$ is a unique blank tape symbol.
Since any bounded-error QTM freely amplifies its success probability,  we assume without loss of generality that the error probability of $M$ is at most   $\varepsilon$.
Notice that the coded skew final configuration begins with an output string. Thus, the superposition of coded skew final configurations of $M$ on input $x$ of length $n$ must be of the form $\sum_{r\in FSC_{M,n}}\qubit{\widetilde{M[r]}} \qubit{\xi_{x,r}}$ for appropriately chosen quantum states $\{\qubit{\xi_{x,r}}\}_{r\in FSC_{M,n}}$. Since $M$ computes $f$ with error probability at most $\varepsilon$,
we conclude that
$\sum_{r\in FSC_{M,n}}\|\measure{\widetilde{f(x)}}{\widetilde{M[r]}} \qubit{\xi_{x,r}}\|^2
\geq 1-\varepsilon$ for all $x$.
Lemma \ref{lemma:second-part} further provides us with a special
quantum function
$g\in\hatsquareqp$ such that, for any number $n\in\nat$ and any string $x\in\{0,1\}^n$,
$\qubit{\phi^p_g(x)}$ has the form $\sum_{r\in FSC_{M,n}} \qubit{M[r]} \qubit{\widehat{\xi}_{x,r}}$  for certain quantum states $\{\qubit{\widehat{\xi}_{x,r}}\}_{r\in FSC_{M,n}}$ satisfying that, for all $r,r'\in FSC_{M,n}$, $\|\measure{\xi_{x,r}}{\xi_{x,r'}}\|= \|\measure{\widehat{\xi}_{x,r}}{\widehat{\xi}_{x,r'}}\|$  and $\measure{\widehat{\xi}_{x,r}}{\widehat{\xi}_{x,r'}}=0$ if $r\neq r'$.
We thus conclude that
$\|\measure{f(x)}{\phi^p_g(x)}\|^2  = \sum_{r\in FSC_{M,n}}
\| \measure{f(x)}{M[r]} \qubit{\widehat{\xi}_{x,r}}\|^2
= \sum_{r\in FSC_{M,n}} | \measure{\widetilde{f(x)}}{\widetilde{M[r]}} |^2 \|\qubit{\widehat{\xi}_{x,r}}\|^2$ since
and $\measure{f(x)}{M[r]} = \measure{\widetilde{f(x)}}{\widetilde{M[r]}}$ and $\measure{\widehat{\xi}_{x,r}}{\widehat{\xi}_{x,r'}}=0$ if $r\neq r'$.
Moreover, from $\|\qubit{\xi_{x,r}}\|=\|\qubit{\widehat{\xi}_{x,r}}\|$, it follows that the term $| \measure{\widetilde{f(x)}}{\widetilde{M[r]}} |^2 \|\qubit{\widehat{\xi}_{x,r}}\|^2$ equals $\|\measure{\widetilde{f(x)}}{\widetilde{M[r]}} \qubit{\xi_{x,r}}\|^2$. Therefore, $\|\measure{f(x)}{\phi^p_g(x)}\|^2$ is at least $1-\varepsilon$.

(2 $\Rightarrow$ 3)
Let $\varepsilon$ be any constant in $[0,1/2)$ and set $\varepsilon'=1-\sqrt{1-\varepsilon}$.
Let us choose a polynomial $p$ and a function $g\in\hatsquareqp$ for which $|f(x)|\leq p(|x|)$ and $\|\measure{f(x)}{\phi^p_g(x)}\|^2 \geq 1-\varepsilon'$ for all strings $x\in\{0,1\}^*$.
Starting with an input $\qubit{\phi^{p,f}(x)}$ ($= \qubit{\widetilde{0^{|f(x)|}}} \qubit{0^{|f(x)|+1}1} \qubit{\phi^p(x)}$),
we first apply $g$ to the last part $\qubit{\phi^p(x)}$ of $\qubit{\phi^{p,f}(x)}$ and obtain $\qubit{\widetilde{0^{|f(x)|}}} \qubit{0^{|f(x)|+1}1} \qubit{\phi^p_g(x)}$.
Next, we apply Lemma \ref{encode-decode} to encode the first $|f(x)|$ qubits of $\qubit{\phi^p_g(x)}$ with the help of $\qubit{0^{|f(x)|+1}1}$ and then  obtain $\qubit{\eta_{f,x}} = \sum_{s:|s|=|f(x)|} \qubit{\widetilde{0^{|f(x)|}}} \qubit{\tilde{s}}\otimes \measure{s}{\phi^p_g(x)}$.
Using the quantum function $COPY_2$ given in Lemma \ref{copy-algorithm}, we then copy each qustring $\qubit{\tilde{s}}$ of $\qubit{\eta_{f,x}}$ onto $\qubit{\widetilde{0^{|f(x)|}}}$ and generate a quantum state of the form $\sum_{s:|s|=|f(x)|} ( \qubit{\tilde{s}}\qubit{\tilde{s}} \otimes  \measure{s}{\phi^p_g(x)} )$.
We then decode the first and the second occurrences of $\qubit{\tilde{s}}$  to $\qubit{0^{|f(x)|+1}1}\qubit{s}$. For later convenience, we transform the first occurrence of $\qubit{0^{|f(x)|+1}1}\qubit{s}$ to $\qubit{s}\qubit{0^{|f(x)|+1}1}$.
In what follows, we abbreviate
$\qubit{s}\!\measure{s}{\phi^p_g(x)}$ as $\qubit{\zeta^p_g[s]}$.
Finally, we locally apply $g^{-1}$ to the last part $\qubit{\zeta^p_g[s]}$, producing $\qubit{\xi_x} = \sum_{s:|s|=|f(x)|} (\qubit{s} \qubit{(0^{|f(x)|+1}1)^2} \otimes g^{-1}(\qubit{\zeta^p_g[s]}))$.

Let $\qubit{\Psi_{f(x)}} = \qubit{f(x)} \qubit{(0^{|f(x)|+1}1)^2} \qubit{\phi^p(x)}$.
Since $\qubit{\phi^{p}_g(x)} = \sum_{s:|s|=|f(x)|} \qubit{\zeta^p_g[s]}$ and $g^{-1}(\qubit{\phi^{p}_g(x)}) = \qubit{\phi^{p}(x)}$,
we derive $\qubit{\phi^{p}(x)} = \sum_{s:|s|=|f(x)|} g^{-1}(\qubit{\zeta^p_g[s]})$.
Therefore, $\qubit{\Psi_{f(x)}}$ coincides with $\sum_{s:|s|=|f(x)|} ( \qubit{f(x)}\qubit{(0^{|f(x)|+1}1)^2} \otimes g^{-1}( \qubit{\zeta^p_g[s]}))$.
Let us consider the inner product $\measure{\Psi_{f(x)}}{\xi_x}$.
By a simple calculation, $\measure{\Psi_{f(x)}}{\xi_x}$ equals $\sum_{s:|s|=|f(x)|}\measure{f(x)}{s} \otimes \tau_x(s)$, where $\tau_x(s)$ is the inner product between $\qubit{\phi^p(x)}$ and $g^{-1}(\qubit{\zeta^p_g[s]})$. Since $g^{-1}$ belongs to $\hatsquareqp$ and is dimension-preserving and norm-preserving by Proposition \ref{inverse} and Lemma \ref{hatsquare-property}, $\tau_x(s)$ equals $\measure{\zeta^p_g[s]}{\zeta^p_g[s]}$.
Since $s$ is forced to take the value $f(x)$ in $\measure{\Psi_{f(x)}}{\xi_x}$, it follows that $\measure{\Psi_{f(x)}}{\xi_x} = \measure{\zeta^p_g[f(x)]}{\zeta^p_g[f(x)]}$, which equals $\|\measure{f(x)}{\phi^p_g(x)}\|^2$.
Therefore, we conclude that $|\measure{\Psi_{f(x)}}{\xi_x}|^2 = \|\measure{f(x)}{\phi^p_g(x)}\|^4\geq (1-\varepsilon')^2 = 1-\varepsilon$, as requested.

(3 $\Rightarrow$ 1)
Since $f$ is polynomially bounded, take a polynomial $p$ such that $|f(x)|\leq p(|x|)$ holds for all strings $x$. Since we do not know the length of  $f(x)$, we want to expand $f$ by setting   $f_1(x)=f(x)01^{p(|x|)+2-|f(x)|}$ so that $|f_1(x)|=p(|x|)+2$ for all strings $x$. Fix $\varepsilon\in[0,1/2)$.
Let us assume that
there exist a function $g_1\in\hatsquareqp$ and a polynomial
$p_1$ that satisfy $|f_1(x)|\leq p_1(|x|)$ and
$|\measure{\Psi_{f_1(x)}}{\phi^{p_1,f_1}_g(x)}|^2 \geq 1-\varepsilon$
for all instances $x\in\{0,1\}^*$.

Using Lemma \ref{lemma:first-part} for the quantum function $g_1$, we can take a single-tape,
polynomial-time,  conservative, plain QTM $M$ with $\appcomplex$-amplitudes for which $M$ on input $\qubit{\phi}$ produces a clean output of  $g_1(\qubit{\phi})$ on its tape. We consider the following machine. On input $x\in\Sigma^*$, we first compute the value $p(|x|)$ deterministically and  generate  $\qubit{\phi^{p,f_1}} = \qubit{0^{p(|x|)+2}1}\qubit{0^{p(|x|)}10^{9p(|x|)}1}\otimes \qubit{x}$. We then run $M$ on $\qubit{\phi^{p,f_1}(x)}$ to produce $\qubit{\phi^{p,f_1}_{g_1}(x)}$.
Since $|\measure{\Psi_{f_1(x)}}{\phi^{p,f_1}_{g_1}(x)}|^2 \geq 1-\varepsilon$, $\qubit{\phi^{p,f_1}_{g_1}(x)}$ contains $f_1(x)$ with probability at least  $1-\varepsilon$. From $f_1(x)$, we extract $f(x)$ and output it. This concludes that $f$ is in $FBQP$.
\end{proofof}

\subsection{Quantum Normal Form Theorem}\label{sec:normal-form}

Our key lemmas, Lemmas \ref{lemma:first-part} and \ref{lemma:second-part}, further lead to a quantum version of Kleene's \emph{normal form theorem} \cite{Kle36,Kle43}, which asserts the existence of a primitive recursive predicate $T(e,x,y)$ and a primitive recursive function $U(y)$ such that, for any recursive function $f(x)$, an appropriate index (called a G\"{o}del number) $e\in\nat$ satisfies $f(x)=U(\mu y.T(e,x,y))$ for all inputs $x\in\nat$, where $\mu$ is the \emph{minimization operator}. This statement is, in essence, equivalent to the existence of universal Turing machine \cite{Tur36}.
Here, we wish to prove a slightly weaker form of the \emph{quantum
normal form theorem} using Lemmas \ref{lemma:first-part}--\ref{lemma:second-part}.

\begin{theorem}[Quantum Normal Form Theorem]\label{normal-form}
There exists a quantum function $f$ in $\squareqp$ such that, for any quantum function $g$ in
$\squareqp$ and any constant $\varepsilon\in (0,1/2)$, there exist a binary string $e$ and a polynomial $p$ satisfying
$\tracenorm{ \density{\psi_{g,x}}{\psi_{g,x}} -
tr_n(\density{\eta_{f,x}}{\eta_{f,x}}) }  \leq \varepsilon$ for any input  $\qubit{x}$ with $x\in\{0,1\}^n$, where $\qubit{\psi_{g,x}} = g(\qubit{x})$ and $\qubit{\eta_{f,x}} =  f(\qubit{\tilde{e}} \qubit{0^{p(|x|)}1}\qubit{x})$. Such a function $f$ is called \emph{universal}.
\end{theorem}

The extra term $\qubit{0^{p(|x|)}1}$ in $\qubit{\eta_{f,x}}$
is needed for providing $g$ with enough work space as in the case of Theorem \ref{theorem:character}. To prove the theorem, we utilize the fact that there is a \emph{universal QTM}, which can simulate all single-tape well-formed QTMs with polynomial slowdown with
any desired accuracy. Such a universal machine was constructed by Bernstein and Vazirani \cite[Theorem 7.1]{BV97} and by Nishimura and Ozawa \cite[Theorem 4.1]{NO02}. We say that $M_1$ on input $x_1$ \emph{simulates} $M_2$ on input $x_2$ \emph{with accuracy}
at most $\varepsilon$ if the total variation distance between two probability distributions  $\{\|\bra{y}{U_{M_1}^{p_1(|x_1|)}} \ket{c^{(x_1)}_{0,1}}\|^2\}_{y\in\{0,1\}^{\leq n}}$ and $\{\|\bra{y}{U_{M_2}^{p_2(|x_2|)}} \ket{c^{(x_2)}_{0,2}}\|^2\}_{y\in\{0,1\}^{\leq n}}$ is at most $\varepsilon$, where $n=\max\{p_1(|x_1|),p_2(|x_2|)\}$ and, for each index $i\in\{1,2\}$, $U_{M_i}$ is the time-evolution operator of $M_i$, $p_i(\cdot)$ expresses the running time of $M_i$, $c^{(x_i)}_{0,i}$ is the skew initial configuration of $M_i$ on input $x_i$, and $y$ ranges over all possible output strings of $M_i$.


\begin{proposition}\label{lemma:universal}{\rm \cite{BV97,Kit97,NC00,NO02}}\hs{2}
There exists a single-tape, well-formed, stationary QTM $U$ such that, for every constant
$\varepsilon\in(0,1)$, a number $t\in\nat$, a single-tape well-formed
QTM $M$ with $\appcomplex$-amplitudes, $U$ on input
$\pair{M,x,t,\varepsilon}$ simulates $M$
on input $x$ for $t$ steps with accuracy at most $\varepsilon$ with
slowdown of a polynomial in $t$ and $\log(1/\varepsilon)$, where $\pair{M,x,t,\varepsilon}$ refers to a fixed, efficient encoding of a quadruplet $(M,x,t,\varepsilon)$. Such a QTM is called \emph{universal}.
\end{proposition}

The improved factor $\log(1/\varepsilon)$ in Proposition \ref{lemma:universal} is attributed to Kitaev \cite{Kit97} and Solovay (cited in \cite[Appendix 3]{NC00}).

For the proof of Theorem \ref{normal-form}, we need to simulate a universal QTM $U$ provided by Proposition \ref{lemma:universal} on a certain conservative QTM even with lower accuracy. Concerning the form of inputs given to $f$, we need to split a quadruplet $(M,x,t,\varepsilon)$ in the proposition into three parts $(M,\varepsilon)$, $t$, and $x$ and then modify $U$ so that $U$ can take any  input of the form $\qubit{\tilde{e}}\qubit{0^{t}1}\qubit{x}$, where $e=\pair{M,\varepsilon}$, and mimic $M$ on $x$ within time $t$ with accuracy
at most $\varepsilon$.
Furthermore, we need to force $U$ to produce \emph{clean outputs} by relocating  all non-blank symbols appearing in the left-side region of any output string to elsewhere in time polynomial in $t$.

Theorem \ref{normal-form} now follows directly by combining Lemmas \ref{lemma:first-part}--\ref{lemma:second-part} and Proposition
\ref{lemma:universal}.


\begin{proofof}{Theorem \ref{normal-form}}
As explained above, let $U$ denote a modified universal QTM that takes inputs of the form $\qubit{\tilde{e}}\qubit{0^t1}\qubit{x}$ for any numbers  $e,t\in\nat^{+}$ and any string $x\in\{0,1\}^*$ and produces clean outputs.
Given any quantum function $g\in\squareqp$,
Lemma \ref{lemma:first-part} guarantees the existence of a polynomial $r$ and a single-tape, $\appcomplex$-amplitude, conservative, plain QTM $M$ having clean outputs for which, on any input string $x\in\{0,1\}^n$,  $M$ produces within $r(n)$ steps a superposition  $g(\qubit{0^{p(n)}1}\qubit{x})\otimes \qubit{0,q_f}$ of skew final configurations composed of $M$'s essential tape region and $M$'s internal status.
To be more precise, we denote by $\widehat{U}_{M}$ the skew time-evolution operator of $M$ and by $c_{0,M}^{(x)}$ the skew initial  configuration of $M$ on any input $\qubit{x}$. We further set $\qubit{\zeta_{M,x}}$ to be $\widehat{U}_{M}^{r(n)}\qubit{c_{0,M}^{(x)}}$ and $g(\qubit{0^{p(n)}1}\qubit{x})$ to be $\qubit{\psi_{g,x}}$. Note that   $\qubit{\zeta_{M,x}}$ equals $\qubit{\psi_{g,x}}\otimes \qubit{0,q_f}$. Since  $\density{\psi_{g,x}}{\psi_{g,x}} = tr_{n}(\density{\psi_{g,x}}{\psi_{g,x}} \otimes \density{0,q_f}{0,q_f})$,  $\density{\psi_{g,x}}{\psi_{g,x}}$ can be expressed as $tr_{n}(\density{\zeta_{M,x}}{\zeta_{M,x}})$.

In contrast, we denote by $\widehat{U}_{\delta}$ the skew time-evolution operator of $U$ and by $c_{0,U}^{(x_e)}$ the skew initial configuration of $U$ on the input $\qubit{x_e}$ for $e=\pair{M,\varepsilon}$ and $x_e = \tilde{e} 0^{r(|x|)}1 x$. Let $m=|x_e|$ and set $\qubit{\zeta_{U,x_e}}$ to be $\widehat{U}_{\delta}^{p(m)}\qubit{c_{0,U}^{(x_e)}}$, which is written as  $\sum_{r\in FSC_{U,m}} \qubit{U[r]} \qubit{\xi_{x_e,r}}$ for an appropriate  set $\{\qubit{\xi_{x_e,r}}\}_{r}$ of orthogonal quantum states.
Proposition \ref{lemma:universal} then ensures that, by  an appropriate choice of a polynomial $s$,
the total variation distance between  $\{\|\bra{y}\widehat{U}_{M}^{t}\ket{c_{0,M}^{(x)}}\|^2\}_{y\in\{0,1\}^n}$ and $\{\|\bra{y}\widehat{U}_{\delta}^{s(t)} \ket{c_{0,U}^{(x_e)}}\|^2\}_{y\in\{0,1\}^n}$ is at most $\varepsilon$ for any number $t\geq0$.

We apply Lemma \ref{lemma:second-part} and then  obtain a $\squareqp$-quantum function $f$ such that, for any quantum state $\qubit{\phi}$, $f(\qubit{\phi})$ represents the result of $\widehat{U}^{s(r(n))}_{\delta}$ applied to $\qubit{\phi}$.
For convenience, we express  $f(\qubit{x_e})$ as $\qubit{\eta_{f,x_e}}$.
Lemma \ref{lemma:second-part} again implies that $\qubit{\eta_{f,x_e}} = \sum_{r\in FSC_{U,m}} \qubit{U[r]} \qubit{\widehat{\xi}_{x_e,r}}$
with $\measure{\xi_{x_e,r}}{\xi_{x_e,r'}} = \measure{\widehat{\xi}_{x_e,r}}{\widehat{\xi}_{x_e,r'}}$ and $\measure{\xi_{x,r}}{\xi_{x,r'}}=0$ if $r\neq r'$ for all $r,r'\in FSC_{U,m}$.
We thus obtain $tr_n(\density{\eta_{f,x_e}}{\eta_{f,x_e}}) = \sum_{r,r'\in FSC_{U,m}} \density{U[r]}{U[r']} \cdot tr(\density{\widehat{\xi}_{x_e,r}}{\widehat{\xi}_{x_e,r'}}) = \sum_{r,r'\in FSC_{U,m}} \measure{\widehat{\xi}_{x_e,r'}}{\widehat{\xi}_{x_e,r}}  \density{U[r]}{U[r']}$, which equals $\sum_{r\in FSC_{U,m}} \|\qubit{\widehat{\xi}_{x_e,r}}\|^2  \density{U[r]}{U[r]}$.
Similarly, we obtain $tr_n(\density{\zeta_{U,x_e}}{\zeta_{U,x_e}}) = \sum_{r\in FSC_{U,m}} \|\qubit{\xi_{x_e,r}}\|^2 \density{U[r]}{U[r]}$.
From those calculations together with $\|\qubit{\xi_{x_e,r}}\|=\|\qubit{\widehat{\xi}_{x_e,r}}\|$,
the equality $tr_{n}(\density{\eta_{f,x_e}}{\eta_{f,x_e}}) = tr_{n}(\density{\zeta_{U,x_e}}{\zeta_{U,x_e}})$ follows immediately.

We thus conclude that $\tracenorm{ \density{\psi_{g,x}}{\psi_{g,x}} -
tr_{n}(\density{\eta_{f,x}}{\eta_{f,x}}) } =
\tracenorm{ tr_{n}(\density{\zeta_{M,x}}{\zeta_{M,x}}) -  tr_{n}(\density{\zeta_{U,x_e}}{\zeta_{U,x_e}}) }$. The last term is upper-bounded by $\sum_{y:|y|=n} | \|\bra{y} \widehat{U}_{\delta}^{s(r(n))} \ket{c_{0,U}^{(x)}}\|^2 - \|\bra{y} \widehat{U}_{M}^{r(n)} \ket{c_{0,M}^{(x)}}\|^2   |$, which is clearly at most $\varepsilon$. Therefore, $f$ is universal.
\end{proofof}

\section{Proof of the Key Lemma}\label{sec:second-part}

To complete the proof of Theorem \ref{theorem:character}, we need to prove the key lemma, Lemma \ref{lemma:second-part}.
This section intends to provide  the lemma's desired proof.
Our proof is inspired by a result of Yao \cite{Yao93}, who demonstrated a quantum-circuit simulation of a QTM.

\subsection{Functional Simulation of QTMs}\label{sec:simulation-QTM}

An essence of the proof of Lemma \ref{lemma:second-part} is a direct step-by-step simulation of the behavior of a single-tape, $\appcomplex$-amplitude, conservative,
plain QTM $M=(Q,\Sigma,\Gamma,\delta,q_{0},Q_f)$, which has clean outputs.
For simplicity, we assume that
$\Sigma=\{0,1\}$ and $\Gamma=\{0,1,b\}$, where $b$
here stands for a unique blank tape symbol, instead of $\#$ used in early sections.
We further assume that $Q_f=\{q_f\}$ and $Q = \{0,1\}^{\ell}$ for a certain fixed even number $\ell>0$ with $q_0 =0^{\ell}$ and $q_f =1^{\ell}$.
Let us assume that, starting
with \emph{binary} input string $x$ written on the single input/work/output tape, $M$ halts in at most $p(|x|)$ steps, where $p$ is an appropriate polynomial associated only with $M$.
We further assume that all computation paths of $M$ on each input halt simultaneously.
For convenience, we also demand that $p(n)>\ell$ for any $n\in\nat$.
It is important to remember that $M$ eventually halts by entering a unique final inner state $q_f$ and making the tape head stationed at the start cell and that no non-blank symbol appears in the left-side region of any output string.

Let $x=x_1x_2\cdots x_n$
be any input given to $M$, where $x_i$ is a bit in $\{0,1\}$ for each index  $i\in[n]$. Associated with $p$, an essential tape region of $M$ on $x$ consists of all the tape cells indexed between $-p(n)$ and $+p(n)$.
We express the tape content of such an essential tape region as a string of the form $\sigma_1\sigma_2\cdots \sigma_{2p(n)+1}$ having length exactly $2p(|x|)+1$ over the tape alphabet $\Gamma = \{0,1,b\}$.
We trace the changes of these symbols as $M$ makes its moves, where  $\sigma_i$ is a tape symbol written at the cell indexed
$i-p(n)-1$ for every index $i\in[2p(n)+1]$.

Since all $\squareqp$-functions are defined to handle quantum states in $\HH_{\infty}$, we need to encode each tape symbol and
thus a tape content. We thus define a new qustring that properly encodes a
configuration $\gamma = (q,h,\sigma_1\sigma_2\cdots \sigma_{2p(n)+1})$
of $M$ to be
\[
 \qubit{q}\otimes \qubit{s_1,\hat{\sigma}_1}\otimes \qubit{s_2,\hat{\sigma}_2}\otimes
\qubit{s_3,\hat{\sigma}_3}\otimes
\cdots \otimes \qubit{s_{2p(n)+1},\hat{\sigma}_{2p(n)+1}},
\]
where each $\hat{\sigma}_i$ is in $\{\hat{0},\hat{1},\hat{b}\}$, each $s_i\in\{\hat{2},\hat{3}\}$ indicates the presence of the tape
head (where  $\hat{2}$ means ``the head rests here'' and $\hat{3}$ means ``no head is here'') at cell $i-p(n)-1$, and $q$ is an inner state in $Q$.
In the subsequent subsections, we call such a qustring a {\em code} of the
configuration $\gamma$ of $M$ and denote it by $\qubit{\hat{\gamma}}$. This code $\qubit{\hat{\gamma}}$ has length $\ell(\qubit{\hat{\gamma}}) = 8p(n)+\ell+4$, which is even and greater than $n$.

Given any binary input $x=x_1x_2\cdots x_n$ of length $n$, let us recall from Section \ref{sec:main-theorem} that    $\qubit{\phi^{p}(x)} = \qubit{0^{|x|}1} \qubit{0^{p(|x|)}1}\qubit{0^{11p(|x|)+6}1}\qubit{x}$ and $\qubit{\phi^{p,f}(x)} = \qubit{\widetilde{0^{|f(x)|}}} \qubit{0^{|f(x)|+1}1} \qubit{\phi^p(x)}$.
Except for Step 1) in Section \ref{step-initialization} as well as all steps in Section \ref{step-entire}, we always ignore the prefix strings  $\widetilde{0^{|f(x)|}} 0^{|f(x)|+1}1 0^{|x|}1 0^{p(|x|)}1$ in $\qubit{\phi^{p,f} (x)}$ and $0^{|x|}1 0^{p(|x|)}1$ in $\qubit{\phi^p(x)}$, and we pay our attention to the remaining qubits.
The desired quantum function $g$ will be constructed step by step through Sections \ref{step-initialization}--\ref{step-output}.

\subsection{Constructing a Coded Initial Configuration}\label{step-initialization}

Given a binary input $x=x_1x_2\cdots x_n$, the initial configuration $\gamma_0$ of $M$ on $x$ is of the form $(q_0,0,b\cdots b x b \cdots b)$, and thus the code  $\qubit{\hat{\gamma}_0}$ of $\gamma_0$ must have the form
\[
\qubit{q_0}\otimes
\qubit{\hat{3},\hat{b}}\otimes
\cdots \otimes \qubit{\hat{3},\hat{b}}\otimes
\qubit{\hat{2},\hat{x}_1}\otimes
\qubit{\hat{3},\hat{x}_2}\otimes
\qubit{\hat{3},\hat{x}_3}\otimes
\cdots \otimes \qubit{\hat{3},\hat{x}_n}\otimes
\qubit{\hat{3},\hat{b}}\otimes
\cdots \otimes \qubit{\hat{3},\hat{b}},
\]
where $\hat{x}_i$ is the code of $x_i$, $q_0$ ($=0^{\ell}$) is the initial inner state, and $x_1$ rests in cell $0$. Notice that $\ell(\qubit{\hat{\gamma}_0})=8p(n)+\ell+4$. In what follows, we show how to generate this particular code $\qubit{\hat{\gamma}_0}$ from the quantum state  $\qubit{\phi^p(x)}$.
For simplicity, we will ignore the term
$\qubit{\hat{q}_0}$ in the following steps except for Step 8).

1) Starting with the input $\qubit{\phi^p(x)}$, we first transform it to $\qubit{0^{n}1} \qubit{0^{p(n)}1} \qubit{10^{11p(n)+5}1}\qubit{x}$ by a quantum function  $h_1 = Skip[NOT]$, which satisfies both  $h_1(\qubit{0^m1}\qubit{\psi}) = \qubit{0^m1}\otimes NOT(\qubit{\psi})$ and $h_1(\qubit{0^{m+1}})=\qubit{0^{m+1}}$ for any number $m\in\nat$ and any quantum state  $\qubit{\psi}\in\HH_{\infty}$. Lemma \ref{lemma:add-zero} guarantees that $h_1$ actually exists in $\hatsquareqp$.

From $\qubit{0^{p(n)}1}\qubit{10^{11p(n)+5}1}$, we wish to generate $\qubit{0^{p(n)}1}\qubit{0^{p(n)}1}\qubit{0^{10p(n)+5}1}$ by an appropriately constructed $\hatsquareqp$-function $f_1$.
To define the desired quantum function $f_1$, we first construct another quantum function $g_1$ that maps   $\qubit{0^{p(n)}1}\qubit{0^m10^k1}$ to $\qubit{0^{p(n)}1}\qubit{0^{m+1}10^{k-1}1}$ for any two integers  $m\geq0$ and $k\geq1$ by setting
\[
g_1(\qubit{\phi}) = \left\{
\begin{array}{ll}
\qubit{\phi} & \mbox{if $\ell(\qubit{\phi})\leq 1$,} \\
\qubit{0}\otimes g_1(\measure{0}{\phi}) + \qubit{1}\otimes g_2(\measure{1}{\phi}) & \mbox{otherwise,}
\end{array} \right.
\]
where $g_2$ is introduced as
\[
g_2(\qubit{\phi}) = \left\{
\begin{array}{ll}
\qubit{\phi} & \mbox{if $\ell(\qubit{\phi})\leq 1$,} \\
SWAP(\qubit{0}\otimes g_2(\measure{0}{\phi}) + \qubit{1}\! \measure{1}{\phi}) & \mbox{otherwise.}
\end{array}  \right.
\]
More formally, we set $g_2 = QRec_1[I,SWAP,I|g_2,I]$ and set  $\hat{g}_2 = Branch[I,g_2]$. We then set $g_1= QRec_1[I,I,\hat{g}_2|g_1,I]$.
The quantum function $f_1$ is finally defined as $f_1=QRec_1[I,g_1,I|f_1,I]$; namely,
\[
f_1(\qubit{\phi}) = \left\{
\begin{array}{ll}
\qubit{\phi} & \mbox{if $\ell(\qubit{\phi})\leq 1$,} \\
g_1(\qubit{0}\otimes f_1(\measure{0}{\phi}) + \qubit{1}\! \measure{1}{\phi}) & \mbox{otherwise.}
\end{array} \right.
\]

In a similar manner, we further transform $\qubit{0^{p(n)}1}\qubit{0^{10p(n)+5}1}$ to $\qubit{0^{p(n)}1} \qubit{0^{2p(n)+2}1} \qubit{0^{8p(n)+2}1}$. We then change $\qubit{0^{n}1}\qubit{0^{2p(n)+2}1}$ to $\qubit{0^{n}1}\qubit{0^{n}1}\qubit{0^{2p(n)-n+1}1}$ and $\qubit{0^{n}1}\qubit{0^{8p(n)+2}1}$ to $\qubit{0^{n}1}\qubit{0^{n}1}\qubit{0^{8p(n)-n+1}1}$.
Overall, the input $\qubit{\phi^p(x)}$ is turned into $\qubit{0^{n}1}\qubit{(0^{p(n)}1)^3}  \qubit{(0^{n}1)^2} \qubit{0^{2p(n)-n+1}1} \qubit{0^{8p(n)-n+1}1} \qubit{x}$.

(*) In what follows, by ignoring $\qubit{0^{n}1}\qubit{(0^{p(n)}1)^3}  \qubit{(0^{n}1)^2}$, we assume that our input is temporarily
$\qubit{0^{8p(n)-n+1}1}\qubit{x}$.

2) For readability, we explain this step using an illustrative example of $\qubit{0^61}\qubit{x_1x_2x_3}$, which we intend to transform to $\qubit{00}\qubit{\hat{3}\hat{x}_1\hat{x}_2\hat{x}_3}$.
For this purpose, we begin with changing $\qubit{0^61x_1x_2x_3}$ to $\qubit{x_3x_2x_110^6}$ by applying $REVERSE$.

To obtain $\qubit{x_30x_20x_1001}\qubit{00}$ from $\qubit{x_3x_2x_1}\qubit{10^6}$, we further apply $g_3=SWAP\circ REP_1$, which changes  $\qubit{x_3x_2x_110^6}$ to $\qubit{x_30x_2x_110^5}$.
We repeatedly apply $g_3$ and transform $\qubit{x_3x_2x_1}\qubit{10^6}$ to $\qubit{x_30x_20x_10}\qubit{10^3}$. This process can be done by the quantum function
$h_3=2QRec_2[I,h',g_3|\{h''_s\}_{s\in\{0,1\}^2}]$ defined by the $2$-qubit quantum recursion, where $h' = Branch_2[\{h'_s\}_{s\in\{0,1\}^2}]$ with $h'_{a1}=SWAP$ and $h'_{a0} = I$ as well as $h''_{a1}=I$ and $h''_{a0}=h_3$ for each bit $a\in\{0,1\}$; namely,
\[
h_3(\qubit{\phi}) = \left\{
\begin{array}{ll}
\qubit{\phi} & \mbox{if $\ell(\qubit{\phi})\leq 2$,} \\
\sum_{a\in\{0,1\}} ( SWAP(\qubit{a1}\!\measure{a1}{\psi_{g_3,\phi}}) + \qubit{a0}\otimes h_3(\measure{a0}{\psi_{g_3,\phi}}) ) & \mbox{otherwise,}
\end{array} \right.
\]
where $\qubit{\psi_{g_3,\phi}}$ stands for $g_3(\qubit{\phi})$.
We further apply
$g_4 = REVERSE\circ h_3$ to change  $\qubit{x_3x_2x_1}\qubit{10^6}$ to $\qubit{0^31}\qubit{\hat{x}_1\hat{x}_2\hat{x}_3}$.
In a general case, the above process transforms  $\qubit{0^n1}\qubit{x_1x_2\cdots x_m}$ to $\qubit{1}\qubit{\hat{x}_1\hat{x}_2\cdots \hat{x}_n}$.

Finally, we apply $g_5$, which maps  $\qubit{0^31}\qubit{\hat{x}_1\hat{x}_2\hat{x}_3}$ to $\qubit{00}\qubit{\hat{3}\hat{x}_1\hat{x}_2\hat{x}_3}$, defined by
\[
g_5(\qubit{\phi}) = \left\{
\begin{array}{ll}
\qubit{\phi} & \mbox{if $\ell(\qubit{\phi})\leq 1$,} \\
SWAP( \qubit{00}\otimes g_5(\measure{00}{\phi}) + \sum_{y\in\{0,1\}^2-\{00\}}  (\qubit{y}\!\measure{y}{\phi})) & \mbox{otherwise.}
\end{array} \right.
\]
In general, $g_5$ changes $\qubit{0^{n+1}1}\qubit{x_1x_2\cdots x_n}$ to $\qubit{\hat{3}\hat{x}_1\hat{x_2}\cdots \hat{x}_n}$.

(*) To explain our procedure further, for readability, we include more $0$s to $\qubit{00}\qubit{\hat{3}\hat{x}_1\hat{x}_2\hat{x}_3}$, and hereafter we are focused on $\qubit{(00)^6}\qubit{\hat{3}\hat{x}_1\hat{x}_2\hat{x}_3}$.

3) We transform the first two bits $00$ in $\qubit{(00)^6}\qubit{\hat{3}\hat{x}_1\hat{x}_2\hat{x}_3}$ to $11$ ($=\hat{2}$)
and then move them to the end of the qustring, resulting in
$\qubit{(00)^5}\qubit{\hat{3}\hat{x}_1\hat{x}_2\hat{x}_3\hat{2}}$.
This transformation is carried out as follows. Similarly to $CNOT$, we define $k_1= NOT\circ SWAP\circ NOT$.
Since $k_1(\qubit{00}) = \qubit{11}$, it suffices to define $f_3= REMOVE_2\circ k_1$.

4) In the beginning, our qustring is of the form  $\qubit{\hat{3}\hat{x}_1\hat{x}_2\hat{x}_3\hat{2}}$ by ignoring the leading bits $(00)^5$.
We pay our attention to the qustring located between $\hat{3}$ and $\hat{2}$. We
place the last two bits $\hat{2}$ into the location immediately right to $\hat{3}\hat{x}_1$ together with changing $\hat{2}$ to $\hat{3}$.
We then obtain
$\qubit{\hat{3}\hat{x}_1 \hat{3}\hat{x}_2\hat{x}_3}$. This process is realized by an appropriate $\hatsquareqp$-function
in the following fashion.

Let $k_2$ denote a bijection
from $\{0,1\}^6$ to $\{0,1\}^6$ satisfying that
$k_2(vw\hat{2})=v\hat{2}w$ if $v,w\in\{\hat{0},\hat{1}\}$,
$k_2(\hat{3}w\hat{2})=\hat{3}w\hat{3}$ if $w\in\{\hat{0},\hat{1}\}$,
$k_2(w\hat{3}\hat{2}) = w\hat{3}\hat{3}$ if $w\in\{\hat{0},\hat{1}\}$,
and $k_2(vwz) = vwz$ if
$v,w,z\neq \hat{2}$. With this $k_2$, we define $g_5$ as
\[
g_5(\qubit{\phi}) = \left\{
\begin{array}{ll}
\qubit{\phi} & \mbox{if $\ell(\qubit{\phi})\leq 6$,} \\
\sum_{y:|y|=6}
\qubit{k_2(y)}\! \measure{y}{\phi} & \mbox{otherwise.}
\end{array} \right.
\]
By Lemma \ref{bijection-function}, $g_5$ belongs to $\hatsquareqp$.
We then define a quantum function $h_5$ by setting $h_5= 2QRec_2[I,g_5,I|h_5,h_5,h_5,h_5]$, namely,
\[
h_5(\qubit{\phi}) = \left\{
\begin{array}{ll}
\qubit{\phi} & \mbox{if
$\ell(\qubit{\phi})\leq 2$,} \\
\sum_{y\in\{\hat{0},\hat{1},\hat{2},\hat{3}\}}
g_5(\qubit{y}\otimes h_5(\measure{y}{\phi})) &
\mbox{otherwise.}
\end{array} \right.
\]
After placing $\hat{3}$ into the right of $\hat{3}\hat{x}_1$, we finally obtain the
qustring $\qubit{(00)^5}\qubit{\hat{3}\hat{x}_1\hat{3}\hat{x}_2\hat{x}_3}$.

5) Let us define $h'_5 = h_5\circ f_3$, which is the compositions of Steps 3)--4). By applying $h'_5$ repeatedly, we can generally change $\qubit{(00)^{n-1}}\qubit{\hat{3}\hat{x}_1\hat{x}_2\cdots \hat{x}_n}$ to $\qubit{\hat{3}\hat{x}_1 \hat{3}\hat{x}_2\cdots \hat{3}\hat{x}_n}$.
Furthermore, if we apply $h'_5$ twice to $\qubit{(00)^3} \qubit{\hat{3}\hat{x}_1\hat{3}\hat{x}_2\hat{3}\hat{x}_3}$ in our example,  then it is possible to append $\qubit{\hat{3}\hat{b}}$ to the end of the qustring by consuming $(00)^3$, and then we obtain $\qubit{\hat{3}\hat{x}_1\hat{3}\hat{x}_2\hat{3}\hat{x}_3} \qubit{\hat{3}\hat{b}}$.
Using $\qubit{0^{2p(n)-n+1}1}$ in $\qubit{\phi^p}$, we repeat Steps 3)--4)  $2p(n)-n+1$ times to encode the content of the first $p(n)+1$ tape cells indexed by nonnegative numbers.
Returning to our example, if we take $\qubit{0^41}$,
then this process transforms   $\qubit{0^{4}1}\qubit{(00)^6}\qubit{\hat{3}\hat{x}_1\hat{x}_2\hat{x}_3}$ into $\qubit{0^{4}1}\qubit{(00)^2} \qubit{\hat{3}\hat{x}_1\hat{3}\hat{x}_2\hat{3}\hat{x}_3} \qubit{\hat{3}\hat{b}}$.
To realize this transform by an appropriate $\hatsquareqp$-function,
since Steps 3)--4) exclude $\qubit{0^41}$, we first need to extend $h'_5$ to $H_1=Skip[h'_5]$ by Lemma \ref{lemma:add-zero}.
The repetition of Steps 3)--4) is done by the following
quantum function $f_5$:
\[
 f_5(\qubit{\phi}) = \left\{
\begin{array}{ll}
\qubit{\phi} & \mbox{if $\ell(\qubit{\phi})\leq 1$,} \\
H_1(\qubit{0}\otimes
 f_5(\measure{0}{\phi}) +
\qubit{1}\!\measure{1}{\phi}) &
\mbox{otherwise.}
\end{array} \right.
\]
More generally, $f_5$ changes $\qubit{0^{2p(n)-n+1}1} \qubit{(00)^{2p(n)-n+1}} \qubit{\hat{3}\hat{x}_1 \hat{x}_2\cdots \hat{x}_n}$ to $\qubit{0^{2p(n)-n+1}1} \qubit{\hat{3}\hat{x}_1\hat{3}\hat{x}_2\cdots \hat{3}\hat{x}_n} \qubit{\hat{3}\hat{b}\hat{3}\hat{b}\cdots \hat{3}\hat{b}}$ with $p(n)-n+1$ copies of $\hat{3}\hat{b}$.

6) Suppose that our qustring has the form $\qubit{(00)^2}\qubit{\hat{3}\hat{x}_1 \hat{3}\hat{x}_2\hat{3} \hat{x}_3}\qubit{\hat{3}\hat{b}}$ by ignoring $\qubit{0^41}$.
We want to change the leftmost
$\hat{3}$ to $\hat{2}$, resulting in $\qubit{(00)^2}\qubit{\hat{2}\hat{x}_1\hat{3}\hat{x}_2\hat{3}
\hat{x}_3}\qubit{\hat{3}\hat{b}}$.
For our purpose, we first choose a unique bijection $p$ satisfying that $p(\hat{2}) = \hat{3}$, $p(\hat{3}) = \hat{2}$, and $p(y)=y$ for all other $y\in\{0,1\}^2$. Using Lemma \ref{bijection-function}, we expand this bijection $p$ to its associated quantum function $g_p$. The quantum function $f_6$ is then defined as
\[
 f_6(\qubit{\phi}) = \left\{
\begin{array}{ll}
\qubit{\phi} & \mbox{if $\ell(\qubit{\phi})\leq 2$,}  \\
g_p ( \qubit{00}\otimes f_6(\measure{00}{\phi}) +  \sum_{y\in\{0,1\}^2-\{00\}} \qubit{y}\! \measure{y}{\phi} )
& \mbox{otherwise.}
\end{array} \right.
\]

7) We then change the series of $00$'s in $\qubit{(00)^2}\qubit{\hat{2}\hat{x}_1\hat{3}\hat{x}_2\hat{3}
\hat{x}_3}\qubit{\hat{3}\hat{b}}$ into
$10$'s. To make this change, we prepare $h_6 = SWAP \circ NOT \circ CNOT \circ NOT$. Note that $h_6(\qubit{00}) = \qubit{10}$ and $h_6(\qubit{11})=\qubit{11}$. We then set
\[
 f_7(\qubit{\phi}) = \left\{
\begin{array}{ll}
\qubit{\phi} & \mbox{if $\ell(\qubit{\phi})\leq 2$,} \\
h_6(\qubit{00}\otimes f_7(\measure{00}{\phi}) +
   \sum_{y:|y|=2\wedge y\neq 00} \qubit{y}\!\measure{y}{\phi} )
  & \mbox{otherwise.}
\end{array} \right.
\]
This quantum function $f_7$ changes $\qubit{(00)^2}$ to $\qubit{(10)^2}$, which equals $\qubit{\hat{3}\hat{b}}$, and thus
we obtain
$\qubit{\hat{3}\hat{b}}\qubit{\hat{2} \hat{x}_1\hat{3}\hat{x}_2\hat{3}\hat{x}_3}
\qubit{\hat{3}\hat{b}}$.
To obtain $p(n)$ copies of $\qubit{\hat{3}\hat{b}}$ in the left-side region of $\qubit{\hat{2}}$ in general,
the string $(00)^{2p(n)}$ is needed to consume.

8) In this final step, we include the term $\qubit{q_0}$ ($=\qubit{0^{\ell}}$) into our procedure.
We combine Steps 1)--7) to transform $\qubit{0^{\ell}} \qubit{0^{2p(n)-n+1}1} \qubit{0^{8p(n)-n+2}1} \qubit{x}$ to $\qubit{q_0} \qubit{\hat{3}\hat{b} \cdots \hat{3}\hat{b}} \qubit{\hat{2}\hat{x}_1\hat{3}\hat{x}_2 \hat{3}\hat{x}_3 \cdots \hat{3}\hat{x}_n} \qubit{\hat{3}\hat{b} \cdots \hat{3}\hat{b}}$ by the quantum function $F_1$ defined as $F_1 = RevBranch_{\ell}[\{g_s\}_{s\in\{0,1\}^{\ell}}]$, where
$g_{0^{\ell}}=f_7$ and $g_{s}=I$ for any string $s$ different from $0^{\ell}$.

\subsection{Simulating a Single Step}\label{step-single}

To simulate an entire computation of $M$ on any given input $x$, we need to simulate all steps of $M$ one by one until $M$ eventually enters the final inner state $q_f$. In what follows, we demonstrate how to simulate a single step of $M$ by changing a
head position, a tape symbol, and an inner state in a given skew configuration.

Note that $M$'s step involves only three consecutive cells, one of which is being scanned by the tape head. To describe three consecutive cells together with $M$'s inner state, in general, we use an \emph{expression} $r$ of the form $ps_1\sigma_1s_2\sigma_2s_3\sigma_3$ using  $p\in\{0,1\}^{\ell}$,
$\sigma_i\in\{0,1,b\}$, and
$s_i\in\{2,3\}$ for each index $i\in[3]$.
Each expression with $s_i=2$ indicates that $M$ is in state $q$, scanning the $i$th cell of the three cells. Note
that the length of $r$ is $\ell+6$.  Let $T$ be the set of all
possible such $r$'s.
Notice that $T$ is a finite set. The \emph{code} $\qubit{\hat{r}}$ of $r = ps_1\sigma_1s_2\sigma_2s_3\sigma_3$ is $\qubit{q}\qubit{\hat{s}_1,\hat{\sigma}_1} \qubit{\hat{s}_2,\hat{\sigma}_2}\qubit{\hat{s}_3,\hat{\sigma}_3}$, which is  of length $\ell+12$.


For simplicity, we call $r$ a {\em target} if $s_1=3$, $s_2=2$,
$s_3=3$, and $\delta(q,\sigma_2)$ is defined. For later use, $s_1\sigma_1s_2\sigma_2s_3\sigma_3$ without $q$ is called a \emph{pre-target} if $qs_1\sigma_1s_2\sigma_2s_3\sigma_3$ is a target.

1) Let $r = q s_1\sigma_1s_2\sigma_2s_3\sigma_3$ be any fixed element in $T$. We prepare a flag qubit $\qubit{0}$ in the end of $\qubit{\phi}$ to mark that a simulation of $M$'s single step is in progress or has been already done. Let us define a quantum function $f_8$, which transforms $\qubit{\tilde{r}}\qubit{0}$ to either $\qubit{\tilde{r}}\qubit{0}$ or $\qubit{\tilde{s}}\qubit{1}$, where $s$ is an expression obtained from $r$ by applying $\delta$ once if $r$ is a target. Let $A$ denote the set of all targets in $T$ and set $A^{\diamond} = \{\tilde{r}\mid r\in A\}$.
Similarly, we define $T^{\diamond}$ from $T$.

For this purpose, we first define a supporting quantum function $h_{r}(\qubit{b}\qubit{\psi})  = NOT(\qubit{b})\otimes \qubit{\tilde{r}}\!\measure{\tilde{r}}{\psi} + \sum_{s\in\{0,1\}^{\ell+12}-A^{\diamond}} (\qubit{b} \otimes \qubit{s}\!\measure{s}{\psi})$ for any target $r\in T$ and any $b\in\{0,1\}$.
Next, we define $\{g_r\}_{r\in T}$ as follows.
If $r$ is a non-target in $T$, then we set $g_r = I$. In what follows, we assume that $r$ is a target. Let $g_r(\qubit{0}\qubit{s}\qubit{\phi}) = \qubit{0}\qubit{s}\qubit{\phi}$ for any string $s\in\{0,1\}^{\ell+12}$.
Since $M$ is plain, $M$ has only two kinds of transitions, shown in (i) and (ii) below.

(i) Consider the case where $\delta(q,\sigma_2) =
e^{i\theta}\qubit{q',\tau,d}$. When $d=L$, we set
\[
g_r(\qubit{1} \qubit{q} \qubit{\hat{3},\hat{\sigma}_1}
\qubit{\hat{2},\hat{\sigma}_2}\qubit{\hat{3},\hat{\sigma}_3} \qubit{\phi})
=
e^{i\theta} \qubit{1} \qubit{q'} \qubit{\hat{2},\hat{\sigma}_1} \qubit{\hat{3},\hat{\tau}}
\qubit{\hat{3},\hat{\sigma}_3} \qubit{\phi}.
\]
When $d=R$, in contrast, we define
\[
g_r(\qubit{1} \qubit{q} \qubit{\hat{3},\hat{\sigma}_1}
\qubit{\hat{2},\hat{\sigma}_2}\qubit{\hat{3},\hat{\sigma}_3} \qubit{\phi})
=
e^{i\theta} \qubit{1} \qubit{q'} \qubit{\hat{3}, \hat{\sigma}_1} \qubit{\hat{3},\hat{\tau}}
\qubit{\hat{2},\hat{\sigma}_3} \qubit{\phi}.
\]

(ii) Consider the case where $\delta(q,\sigma_2) =
\cos\theta\qubit{q_1,\tau_1,d_1} + \sin\theta\qubit{q_2,\tau_2,d_2}$.
If $(d_1,d_2)=(R,L)$, then we define $g_r$ as
\[
g_r(\qubit{1} \qubit{q} \qubit{\hat{3},\hat{\sigma}_1}\qubit{\hat{2},\hat{\sigma}_2}
\qubit{\hat{3},\hat{\sigma}_3} \qubit{\phi})
= \cos\theta
\qubit{1} \qubit{q_1} \qubit{\hat{3},\hat{\sigma}_1} \qubit{\hat{3},\hat{\tau}_1}
\qubit{\hat{2},\hat{\sigma}_3} \qubit{\phi}
+ \sin\theta
\qubit{1} \qubit{q_2} \qubit{\hat{2},\hat{\sigma}_1} \qubit{\hat{3},\hat{\tau}_2}
\qubit{\hat{3},\hat{\sigma}_3} \qubit{\phi}.
\]
The other values of $(d_1,d_2)$ are similarly handled.

Notice that $\{g_r(\qubit{1}\qubit{\tilde{r}})\}_{r\in A}$ forms an orthonormal set because $M$ is well-formed, and thus $\delta$ satisfies all the conditions stated in Section \ref{sec:def-QTMs}.
Once the flag qubit becomes $\qubit{1}$, we do not need to apply $g_r\circ h_r$. Thus, we further set $g'_r = Branch[g_r\circ h_r, I]$. By combining all $g'_r$'s, we define $g$ as $g = Compo[\{g'_r\}_{r\in T}]$, which implies
\[
g(\qubit{0}\qubit{\phi}) = \sum_{r\in T} ( g_r(\qubit{1}\qubit{\tilde{r}}\! \measure{\tilde{r}}{\phi} )) + \sum_{s\in \{0,1\}^{\ell+12}-T^{\diamond}} ( \qubit{0}\qubit{s}\! \measure{s}{\phi} )
\]
for any quantum state $\qubit{\phi}\in\HH_{\infty}$.
We can claim that $g$ is norm-preserving and it can be constructed from the initial quantum functions in Definition \ref{def:initial}
by applying the construction rules. From this claim, $g$ falls in  $\hatsquareqp$.
Finally, we define $f_8 = REMOVE_1\circ (g^{\leq \ell+13}\otimes I)$, where $g^{\leq \ell+13}\otimes I$ is defined as in Lemma \ref{two-diff-function}.
Intuitively, $f_8$ changes the content of three consecutive tape cells whose middle cell is being scanned by a tape head.

2) We want to apply $f_8$ to all three consecutive tape cells.
Firstly, we find a code of a pre-target $\hat{s}_1\hat{\sigma}_1\hat{s}_2 \hat{\sigma}_2\hat{s}_3\hat{\sigma}_3$ (using the quantum
recursion), move an inner state $q$ as well as a marker $\qubit{0}$ forward (by $REP_{\ell+1}$) to generate a block of the form
$\qubit{0}\qubit{q}\qubit{\hat{s}_1\hat{\sigma}_1 \hat{s}_2\hat{\sigma}_2\hat{s}_3\hat{\sigma}_3}$ and apply $f_8$ to this block. This changes $\qubit{0}$ to $\qubit{1}$ to record the execution of the current procedure.
We then move the obtained inner state and the marker back to the end (by $REMOVE_{\ell+1}$).
This entire procedure can be executed by an appropriate quantum function, say, $F_2$.

To be more formal, we first set another quantum function $p$ to be $REMOVE_{\ell+1}\circ LENGTH_{\ell+13}[f_8] \circ REP_{\ell+1}$.
The quantum function $F'_2$ is then defined as
\[
F'_2 (\qubit{\phi}\qubit{0}\qubit{q}) = \left\{
\begin{array}{ll}
\qubit{\phi}\qubit{0}\qubit{q} & \mbox{if $\ell(\qubit{\phi}) < \ell+13$,} \\
 p(\sum_{s:|s|=4} (\qubit{s}\otimes F'_2(\measure{s}{\phi}\qubit{0}\qubit{q}))) & \mbox{otherwise.}
\end{array}  \right.
\]
To complete the transformation, we further define $F_2 = REP_{\ell+1}\circ F'_2 \circ REMOVE_{\ell+1}$. After the application of $F_2$, the first qubit of $\qubit{0}\qubit{q}\qubit{\phi}$ turns to $\qubit{1}$, marking an execution of $M$'s single step. If we want to repeat an application of $F_2$, we need to reset $\qubit{1}$ back to $\qubit{0}$.

\subsection{Completing the Entire Simulation}\label{step-entire}

We have shown in Section \ref{step-single} how to simulate a single step of $M$ on $x$ by $F_2$. Here, we want to simulate all steps of $M$ by applying $F_2$ inductively to any input of the form  $\qubit{0^{p(|x|)}1}\otimes\qubit{\psi}$,  where
$\qubit{\psi}$ represents a superposition of coded skew configurations of $M$ on $x$. This process is implemented by a new quantum function $F_3$, which repeatedly applies $NOT\circ F_2$ to $\qubit{0}\qubit{\psi}$, $p(|x|)$ times, where $NOT$ resets $\qubit{1}$ to $\qubit{0}$.

We first take a quantum function  $\hat{g} = Skip[NOT\circ F_2]$ by Using Lemma \ref{lemma:add-zero}.
The desired quantum function $F_3$ must satisfy
\[
F_3(\qubit{\phi}) =  \left\{
\begin{array}{ll}
\qubit{\phi} & \mbox{if
$\ell(\qubit{\phi})\leq 1$,} \\
\qubit{0}\otimes
\hat{g}(F_3(\measure{0}{\phi})) + \qubit{1}\otimes I(\measure{1}{\phi}) &
\mbox{otherwise.}
\end{array} \right.
\]
Note that the number of the applications of $\hat{g}$ is exactly $p(|x|)$. This function $F_3$ can be realized with a use of the single-qubit quantum recursion of the form $F_3= QRec_1[I,Branch[\hat{g},I],I|F_3,I]$.

\subsection{Preparing an Output}\label{step-output}

Assume that a superposition of coded skew final configurations of $M$ on the given input $x$ of length $n$ has the form $\sum_{r\in FSC_{M}(x)} \qubit{\widetilde{M[r]}} \qubit{\xi_{x,r}}$ for a certain series  $\{\qubit{\xi_{x,r}}\}_{r\in FSC_{M,n}}$ of quantum states, where $M[r]$ indicates $M$'s output string appearing in a skew final configuration $r$ including only an essential tape region of $M$ on $x$.

To show the first part of Lemma \ref{lemma:second-part}, in the end of the simulation,
we need to generate $\widetilde{M[r]}$ in the leftmost portion of the qustring obtained by the simulation.
To achieve this goal, we first move the content of the left-side region of the start cell to the right end of the essential tape region.

As an illustrative example, suppose that we have already obtained the quantum state $\qubit{\hat{3}\hat{b} \hat{3}\hat{b} \hat{3}\hat{b}} \qubit{\hat{2}\hat{0}
\hat{3}\hat{0} \hat{3}\hat{1}} \qubit{\hat{3}\hat{b}}$ after executing steps in Section \ref{step-entire}. In this case, the outcome of $M$ in this skew final configuration $r$ is $001$, and thus  $\widetilde{M[r]} =\hat{0}\hat{0}\hat{1}\hat{2}$ holds.
In what follows, we want to transform  $\qubit{\hat{3}\hat{b} \hat{3}\hat{b} \hat{3}\hat{b}} \qubit{\hat{2}\hat{0}
\hat{3}\hat{0} \hat{3}\hat{1}} \qubit{\hat{3}\hat{b}}$  into
$\qubit{\hat{0}\hat{0}\hat{1}\hat{2}} \qubit{\hat{3}\hat{b} \hat{3}\hat{b} \hat{3}\hat{b}} \qubit{\hat{1}\hat{3}\hat{3}\hat{3}}$, which equals $\qubit{\widetilde{M[r]}} \qubit{\hat{3}\hat{b} \hat{3}\hat{b} \hat{3}\hat{b}} \qubit{\hat{1}\hat{3}\hat{3}\hat{3}}$, by an appropriate $\hatsquareqp$-function, say, $F_4$.

(i) Starting with $\qubit{\hat{3}\hat{b} \hat{3}\hat{b} \hat{3}\hat{b}} \qubit{\hat{2}\hat{0} \hat{3}\hat{0} \hat{3}\hat{1}} \qubit{\hat{3}\hat{b}}$,  to mark the end portion of the tape cell, we  change the last marker $\hat{3}$ to $\hat{1}$ by applying $NOT$ to $\hat{3}$ and then obtain $\qubit{\hat{3}\hat{b}\hat{3}\hat{b} \hat{3}\hat{b}} \qubit{\hat{2}\hat{0} \hat{3}\hat{0}\hat{3}\hat{1}}\qubit{\hat{1}\hat{b}}$. This process is referred to as $f_9$.

(ii) By moving repeatedly the leftmost $\qubit{\hat{3}\hat{b}}$
in $\qubit{\hat{3}\hat{b}\hat{3}\hat{b} \hat{3}\hat{b}} \qubit{\hat{2}\hat{0} \hat{3}\hat{0}\hat{3}\hat{1}}\qubit{\hat{1}\hat{b}}$ to the end of of the qustring, we eventually produce
$\qubit{\hat{2}\hat{0} \hat{3}\hat{0}\hat{3}\hat{1}} \qubit{ \hat{1}\hat{b}  \hat{3}\hat{b}\hat{3}\hat{b} \hat{3}\hat{b}}$.
To realize this entire transform, we need to define
\[
f_{10}(\qubit{\phi}) = \left\{
\begin{array}{ll}
\qubit{\phi} & \mbox{if $\ell(\qubit{\phi})<4$,} \\
\sum_{a\in\{0,1\}} \qubit{\hat{2}\hat{a}}\! \measure{\hat{2}\hat{a}}{\phi} + \sum_{z\in B_4} REMOVE_4 (\qubit{z} \otimes f_{10}(\measure{z}{\phi})) & \mbox{otherwise,}
\end{array}  \right.
\]
where $B_4 = \{0,1\}^4 - \{\hat{2}\hat{0},\hat{2}\hat{1}\}$.
More formally, we define $f_{11}$ as $f_{11} = 2QRec_4[I,h'',I|\{f_z\}_{z\in\{0,1\}^4}]$, where $h''= Branch[\{h''_z\}_{z\in\{0,1\}^4}]$ with $h''_{\hat{2}\hat{0}}=h''_{\hat{2}\hat{1}} = I$ and $h''_z = REMOVE_4$ as well as $f_{\hat{2}\hat{0}}= f_{\hat{2}\hat{1}} = I$ and $f_{z}= f_{10}$ for any $z\in B_4$.


(iii) The current qustring has the form  $\qubit{\hat{2}\hat{0}\hat{3}\hat{0}\hat{3}\hat{1}}
\qubit{\hat{1}\hat{b}} \qubit{\hat{3}\hat{b} \hat{3}\hat{b} \hat{3}\hat{b}}$.
We sequentially remove each of the markers $\{\hat{1}, \hat{2}, \hat{3}\}$ in $\qubit{\hat{2}\hat{0}\hat{3}\hat{0}\hat{3}\hat{1}}
\qubit{\hat{1}\hat{b}}$
to the end of the entire qustring and produce $\qubit{\hat{0}\hat{0}\hat{1} \hat{b}}
\qubit{\hat{3}\hat{b}\hat{3}\hat{b}\hat{3}\hat{b}} \qubit{\hat{1}\hat{3}\hat{3}\hat{2}}$.
To implement this process, we apply $h_8$ defined by $h_8 = 2QREC_3[I,REMOVE_2,I| \{h_{y}\}_{y\in\{0,1\}^4}]$, where $h_{\hat{1}\hat{b}}=I$ and $h_y=h_8$ for all $y\in\{0,1\}^4-\{\hat{1}\hat{b}\}$; in other words,
\[
h_8(\qubit{\phi}) = \left\{
\begin{array}{ll}
\qubit{\phi} & \mbox{if $\ell(\qubit{\phi})< 4$,} \\
REMOVE_{2}( \qubit{\hat{1}\hat{b}}\! \measure{\hat{1}\hat{b}}{\phi} + \sum_{y\in\{0,1\}^4-\{\hat{1}\hat{b}\}} \qubit{y}\otimes
h_8(\measure{y}{\phi}) ) & \mbox{otherwise.}
\end{array} \right.
\]

(iv) At last, we change the rightmost $\hat{b}$ in $\qubit{\hat{0}\hat{0}\hat{1} \hat{b}}$ to $\hat{2}$ and then obtain  $\qubit{\hat{0}\hat{0}\hat{1}\hat{2}}$ ($=\qubit{\widetilde{001}}$).
To produce such a qustring, we apply the quantum function $h_6$ defined by
\[
h_9(\qubit{\phi}) = \left\{
\begin{array}{ll}
\qubit{\phi} & \mbox{if $\ell(\qubit{\phi})<2$,} \\
 \sum_{y\in\{\hat{0},\hat{1}\}}
\qubit{y}\otimes
h_9(\measure{y}{\phi}) + \qubit{\hat{b}}\!\measure{\hat{2}}{\phi} + \qubit{\hat{2}}\!\measure{\hat{b}}{\phi} & \mbox{otherwise.}
\end{array} \right.
\]
Formally, we set $h_9=2QRec_1[I,h'_9,I|\{h_y\}_{y\in\{0,1\}^2}]$, where $h'_9=Branch[\{h''_y\}_{y\in\{0,1\}^2}]$ with $h''_{\hat{2}}=h''_{\hat{b}}=SWAP\circ NOT\circ SWAP$ and $h''_y=I$  together with $h_{\hat{2}}=h_{\hat{b}}=I$ and $h_y=h_9$ for any $y\in\{\hat{0},\hat{1}\}$.

(vi) To perform Steps (i)--(v) at once, we combine all quantum functions used in Steps (i)--(v) and define a single quantum function $F_4= h_9\circ
 h_8 \circ f_{10} \circ f_9$. Overall, the resulted qustring is $\qubit{\hat{0}\hat{0}\hat{1} \hat{2}}
\qubit{\hat{3}\hat{b}\hat{3}\hat{b}\hat{3}\hat{b}} \qubit{\hat{1}\hat{3}\hat{3}\hat{2}}$.

\s

The quantum state $\qubit{\phi^p_{F_4}(x)}$ can be expressed as $\sum_{r\in FSC_{M,n}} \qubit{\widetilde{M[r]}} \qubit{\widehat{\xi}_{x,r}}$ for certain quantum states $\{\qubit{\widehat{\xi}_{x,r}}\}_{r}$. Since we deal with an  essential tape region of $M$, it instantly follows  that,  for every $x,x'\in\{0,1\}^n$ and $r,r'\in FSC_{M,n}$,   $\|\measure{\xi_{x,r}}{\xi_{x',r'}}\| = \|\measure{\widehat{\xi}_{x,r}}{\widehat{\xi}_{x',r'}}\|$ and $\measure{\hat{\xi}_{x,r}}{\hat{\xi}_{x,r'}}=0$ if $r\neq r'$.
Therefore, $F_4$ satisfies the condition of the first part of the lemma.

For the second part of the lemma, we further need to retrieve  $\qubit{M[r]}$ from the coded qustring $\qubit{\widetilde{M[r]}}$. From the previous illustrative example $\qubit{\hat{0}\hat{0}\hat{1} \hat{2}}
\qubit{\hat{3}\hat{b}\hat{3}\hat{b}\hat{3}\hat{b}} \qubit{\hat{1}\hat{3}\hat{3}\hat{2}}$, we need to produce  $\qubit{001}\qubit{1} \qubit{\hat{3}\hat{b}\hat{3}\hat{b}\hat{3}\hat{b}} \qubit{\hat{1}\hat{3}\hat{3}\hat{2}} \qubit{1000}$.
For this purpose, we define the quantum function $g_7$ by setting $g_7 = 2QRec_1[I,REMOVE_1,I|\{g'_z\}_{z\in\{0,1\}^2}]$; namely,
\[
g_7(\qubit{\phi}) = \left\{
\begin{array}{ll}
\qubit{\phi} & \mbox{if $\ell(\qubit{\phi})<2$,} \\
REMOVE_1( \sum_{y\in\{0,1\}} ( \qubit{0y}\otimes g_7(\measure{0y}{\phi}) + \qubit{1y}\!\measure{1y}{\phi}) ) & \mbox{otherwise.}
\end{array} \right.
\]
Using $g_7$, we finally set $F_5 = g_7\circ F_4$ to obtain the second part of the lemma.

This completes the proof of Lemma \ref{lemma:second-part}.


\subsection{Simple Applications of the Simulation Procedures}

The proof of Lemma \ref{lemma:second-part} provides useful procedures not only for the construction of the desired quantum function $g$ but also for other special-purpose quantum functions.
Hereafter, as simple applications of the simulation procedures given
in Steps 1)--6) of Section \ref{step-initialization}, we will explain how to encode/decode classical strings and how to duplicate classical information by $\hatsquareqp$-functions.

Steps 1)--8) in Section \ref{step-initialization} describe a transformation between classical strings and their encodings. A slight modification of Step 2) introduces an encoder $Encode$, which properly encodes binary strings $s$ to $\tilde{s}$.

\begin{lemma}\label{encode-decode}
There exists a quantum function $Encode$ in $\hatsquareqp$ that satisfies $Encode(\qubit{0^{k+1}1}\otimes \qubit{\phi}) = \sum_{s\in\{0,1\}^k}(\qubit{\tilde{s}}\otimes \qubit{s}\!\measure{s}{\phi})$ for any number $k\in\nat$ and any quantum state $\qubit{\phi}\in\HH_{\infty}$. Moreover, the quantum function $Decode = Encode^{-1}$ is also in $\hatsquareqp$.
\end{lemma}

Given extra bits $0^k1$, the first $k$ qubits of $\qubit{\phi}$ is properly encoded by $Encode$ by consuming $0^k1$.
For example, $Encode$ changes $\qubit{0^31}\qubit{a_1a_2}$ to $\qubit{\hat{a}_1\hat{a}_2\hat{2}}$ and $Decode$ returns $\qubit{\hat{a}_1\hat{a}_2\hat{2}}$ back to $\qubit{0^31}\qubit{a_1a_2}$. Notice that Proposition \ref{inverse} ensures $Decode\in\hatsquareqp$ from $Encode\in\hatsquareqp$.

Quantum mechanics in general prohibits us from duplicating unknown quantum states; however, it is possible to copy each classical string quantumly. We thus have the following quantum function $COPY_2$, which copies the content of the first $k$ qubits of any input.

\begin{lemma}\label{copy-algorithm}
There exists a quantum function $COPY_2$ in $\hatsquareqp$ that satisfies the following condition: for any $k\in\nat^{+}$
and any $\qubit{\phi}\in\HH_{\infty}$,
\[
COPY_2(\qubit{\widetilde{0^{k}}}\otimes \qubit{\phi}) = \sum_{s\in\{0,1\}^{k}} (\qubit{\tilde{s}}\otimes  \qubit{\tilde{s}}\! \measure{\tilde{s}}{\phi}).
\]
\end{lemma}

The encoding $\widetilde{0^k}$ of $0^k$ is needed to distinguish $0^k$ from any part of $\qubit{\phi}$ because $k$ is not a fixed constant.
When $\qubit{\phi}$ has the form $\qubit{x}\qubit{\psi}$, the quantum function $COPY_2$ works as $COPY_2(\qubit{\widetilde{0^{k}}} \otimes \qubit{\tilde{x}}\qubit{\psi}) = \qubit{\tilde{x}}\otimes  \qubit{\tilde{x}}\qubit{\psi}$.

\begin{proofof}{Lemma \ref{copy-algorithm}}
We wish to construct the desired quantum function $COPY_2$ as follows. To make our construction process readable, we use an illustrative example of $\qubit{\widetilde{0^k}}\qubit{\phi} = \qubit{\widetilde{0^2}}\qubit{\widetilde{a_1a_2}}$ ($=  \qubit{\hat{0}\hat{0}\hat{2}} \qubit{\hat{a}_1\hat{a}_2\hat{2}}$)
to show  how each constructed quantum function works.

1) Steps 1)--6) of Section \ref{step-initialization} transform  $\qubit{\hat{0}\hat{0}\hat{1}} \qubit{\hat{a}_1\hat{a}_2}$ to $\qubit{\hat{2}\hat{a}_1\hat{3}\hat{a}_2}\qubit{\hat{3}}$.
By a slight modification of these steps, it is possible to transform $\qubit{\hat{0}\hat{0}\hat{2}} \qubit{\hat{a}_1\hat{a}_2\hat{2}}$ to $\qubit{\hat{3}\hat{a}_1\hat{3}\hat{a}_2\hat{2}}\qubit{\hat{2}}$.
We denote by $f_1$ a quantum function that realizes this transformation.

2) By copying $\hat{a}_i$ in $\hat{3}\hat{a}_i$ onto $\hat{3}$ for each $i\in\{1,2\}$, we change $\hat{3}\hat{a}_i$ to $\hat{a}_i\hat{a}_i$ and  then obtain $\qubit{\hat{a}_1\hat{a}_1 \hat{a}_2\hat{a}_2\hat{2}} \qubit{\hat{2}}$. This step can be formally made in the following way.
Firstly, we define a quantum function $h_2$ that satisfies  $h_2(\qubit{b_1b_2}\qubit{b_3b_4}\otimes \qubit{\phi}) = \qubit{b_4b_2b_1b_3}\otimes \qubit{\phi}$ for any $\phi\in\HH_{\infty}$.
Such a quantum function actually exists by Lemma \ref{bijection-function}.
Secondly, we set $DUP$ to be $h_2^{-1}\circ (SWAP \circ NOT \circ SWAP \circ CNOT) \circ h_2 \circ NOT$.
It then follows that $DUP(\qubit{\hat{3}}\qubit{\hat{a}}\otimes \qubit{\phi}) = \qubit{\hat{a}}\qubit{\hat{a}}\otimes \qubit{\phi}$ for any bit $a\in\{0,1\}$ and any quantum state $\qubit{\phi}\in\HH_{\infty}$. With this $DUP$, we further define $f_2$ by the 4-qubit quantum recursion as
\[
f_2(\qubit{\phi}) = \left\{
\begin{array}{ll}
\qubit{\phi} & \mbox{if $\ell(\qubit{\phi})< 4$,} \\
DUP ( \sum_{a\in\{0,1\}} \qubit{\hat{3}\hat{a}}\otimes f_2(\measure{\hat{3}\hat{a}}{\phi}) + \sum_{y\in B'_4} \qubit{y}\!\measure{y}{\phi} & \mbox{otherwise,}
\end{array} \right.
\]
where $B'_4 = \{0,1\}^4-\{\hat{3}\hat{0},\hat{3}\hat{1}\}$.

3) Next, we transform $\qubit{\hat{a}_1\hat{a}_1 \hat{a}_2\hat{a}_2\hat{2}}\qubit{\hat{2}}$ to $\qubit{\hat{a}_1\hat{a}_2\hat{2}} \qubit{\hat{2}\hat{a}_2\hat{a}_1}$.
This transformation can be done by $f_3$ defined as
\[
f_3(\qubit{\phi}) = \left\{
\begin{array}{ll}
\qubit{\phi} & \mbox{if $\ell(\qubit{\phi}) < 4$,} \\
REMOVE_2( \sum_{y\in\{0,1\}^4-\{\hat{2}\hat{2}\}} \qubit{y}\otimes f_3( \measure{y}{\phi}) + \qubit{\hat{2}\hat{2}}\!\measure{\hat{2}\hat{2}}{\phi}) & \mbox{otherwise.}
\end{array} \right.
\]

4) We then change $\qubit{\hat{a}_1\hat{a}_2\hat{2}} \qubit{\hat{2}\hat{a}_2\hat{a}_1}$ to $\qubit{\hat{2}\hat{a}_2\hat{a}_1} \qubit{\hat{2}\hat{a}_2\hat{a}_1}$ by removing each two qubits in the first part to the end. This process is precisely realized by $f_4$ defined as
\[
f_4(\qubit{\phi}) = \left\{
\begin{array}{ll}
\qubit{\phi} & \mbox{if $\ell(\qubit{\phi})< 4$,} \\
REMOVE_2 ( \sum_{y\in\{0,1\}^2-\{\hat{2}\}} ( \qubit{y}\otimes f_5(\measure{y}{\phi})) + \qubit{\hat{2}}\!\measure{\hat{2}}{\phi} ) & \mbox{otherwise.}
\end{array} \right.
\]

5) Finally, we reverse the whole qustring to obtain $\qubit{\hat{a}_1\hat{a}_2\hat{2}} \qubit{\hat{a}_1\hat{a}_2\hat{2}}$ by applying $REVERSE$.

This completes the proof of the lemma.
\end{proofof}

\section{Future Challenges}\label{sec:applications}

In Definition \ref{sec:definition}, we have defined $\squareqp$-functions on  $\HH_{\infty}$ and we have given in Theorem \ref{theorem:character} a new characterization of $\fbqp$-functions in terms of these $\squareqp$-functions. To point out the directions of future research, we wish to raise a challenging open question in Section \ref{sec:open-question} and to present in  Sections \ref{sec:descriptional-complexity}--\ref{sec:type-2-functional} three  possible implications of our schematic definition to the subjects of descriptional complexity, firt-order theories, and higher-type functionals. In Section \ref{sec:programming-language}, we will remark a practical application to the designing of quantum programming languages.

\subsection{Seeking a More Reasonable Schematic Definition}\label{sec:open-question}

Our schematic definition (Definition \ref{def:initial}) is composed of the initial quantum functions, which are derived from  natural, simple quantum gates, and the construction rules, including the multi-qubit quantum recursion, which significantly enriches the scope of constructed quantum functions. The choice of initial functions and construction rules that we have used in this paper directly affects the richness of $\squareqp$-functions. Although our $\squareqp$ is sufficient to characterize $\fbqp$, if we seek for enriching the $\squareqp$-functions, one way is to supplement additional initial quantum functions.
As a concrete example, let us consider the quantum Fourier transform (QFT), which plays an important role in, e.g., Shor's factoring quantum algorithm \cite{Sho97}.
We have demonstrated in Lemma \ref{Fourier-transform} how to implement a restricted form of QFT working on a fixed number of qubits, and thus it belongs to $\hatsquareqp$. Nonetheless, a more general form of QFT, acting on an ``arbitrary'' number of qubits, may not be realized \emph{precisely} by $\hatsquareqp$-functions although it can be approximated to any desired accuracy by the $\hatsquareqp$-functions. To remedy the exclusion of QFT  from our function class $\squareqp$, for instance, we can expand the current $\hatsquareqp$ by including as an initial quantum function the quantum function defined as
\[
CROT(\qubit{\phi}\qubit{0^j}) = \qubit{0}\! \measure{0}{\phi}\qubit{0^j} + \omega_j \qubit{1}\! \measure{1}{\phi}\qubit{0^j} \text{ (controlled rotation)},
\]
where $\omega_j = e^{2\pi i/2^j}$, with an extra term $0^j$. It is not difficult to construct QFT from quantum functions in this expanded $\hatsquareqp$ obtained by adding $CROT$.
As this example shows, it remains important to seek for a simpler, more reasonable schematic definition of quantum functions, which are capable of precisely characterizing both $\bqp$ and $\fbqp$ and also simplifying the proof of Theorem \ref{theorem:character}.

From the minimalist's viewpoint, on the contrary, we may be able to eliminate certain schemata or replace them by simpler ones but still ensure the characterization result of $\bqp$ and $\fbqp$ in terms of $\squareqp$-functions. As a concrete example, we may ask whether our multi-qubit quantum recursion can be replaced by 1-qubit quantum recursion at the cost of adding extra initial quantum functions.

\subsection{Introduction of Descriptional Complexity and First-Order Theories}\label{sec:descriptional-complexity}

As noted in Section \ref{sec:definition}, our schematic definition of $\squareqp$-functions provides us with a natural means of assigning the \emph{descriptional complexity}---a new complexity measure---to each of those quantum functions in $\squareqp$. This complexity measure has been used to prove, for instance, Lemma \ref{hatsquare-property}.  As a consequence of our main theorem, this complexity measure concept also transfers to languages in $\bqp$ and functions in $\fbqp$,  and thus it  naturally helps us introduce the notion of the descriptional complexities of such languages and functions.

It is further possible for us to extend this complexity measure to ``arbitrary''  languages and functions on $\{0,1\}^*$, which are not necessarily limited to $\fbqp$ and $\bqp$, and to discuss their ``relative''  complexity to $\squareqp$.
More formally, given a function $f$ on $\{0,1\}^*$, the \emph{$\squareqp$-descriptional complexity of $f$ at length $n$} is  the minimal number of times we use initial quantum functions and construction rules to build a $\squareqp$-function $g$ for which $\ell(\qubit{\phi^p_g(x)})=|f(x)|$ and $|\measure{f(x)}{\phi^{p}_{g}(x)}|^2\geq 2/3$ hold for a certain polynomial $p$ with $|f(x)|\leq p(|x|)$ for all strings $x$ of length exactly $n$. We write $dc(f)[n]$ to denote the $\squareqp$-descriptional complexity of $f$ at length $n$. Obviously, every $\squareqp$-function has \emph{constant}  $\squareqp$-descriptional complexity at every length.
In a similar spirit but based on quantum finite automata, Villagra and Yamakami \cite{VY15} discussed the quantum state complexity restricted to inputs of length exactly $n$ (as well as length at most $n$). It has turned out that such complexity measure is quire useful. Refer to \cite{VY15} for the detailed definitions.
Our new complexity measure $dc(f)[n]$ is also expected to be a useful tool in classifying languages and functions in descriptional power in a way that is quite different from what QTMs and quantum circuits do.


In a much wider perspective, our schematic definition of polynomial-time  quantum computability may lead to the future development of an appropriate form of first-order theories over quantum states in Hilbert spaces or \emph{first-order quantum theories}, for short.
In the literature, first-order theories and their natural subtheories have become a fruitful  research subject in mathematical logic and recursion theory and they have also found numerous applications in other fields as well.
A weak form of their subtheories has been studied in quantum complexity theory. For instance, using bounded quantifiers over quantum states in Hilbert spaces,  quantum analogues of $\np$ and the Meyer-Stockmeyer polynomial(-time) hierarchy have been discussed in \cite{Yam02}. Unfortunately, we are still far away from obtaining well-accepted first-order theories and useful subtheories for quantum computing.

\subsection{Extension to Type-2 Quantum Functionals}\label{sec:type-2-functional}

Conventionally, functions mapping $\Sigma^*$ to $\Sigma^*$ are  categorized as \emph{type-1 functionals}, whereas \emph{type-2 functionals} are functions taking inputs from $\Sigma^*$ together with type-1 functionals. In computational
complexity theory, such
\emph{type-2 functionals} have been extensively discussed in, e.g.,   \cite{Con73,CK89,Meh76,Tow90,Yam95}.

In analogy to the classical case, we call $\squareqp$-functions on $\HH_{\infty}$  \emph{type-1 quantum functionals}.
To introduce \emph{type-2 quantum functionals}, we start with an
arbitrary quantum function $O$ mapping $\HH_{\infty}$ to $\HH_{\infty}$, which is treated as a \emph{function oracle} (an \emph{oracle function} or simply an \emph{oracle}) in the following formulation.
From such an oracle $O$, we define a new linear operator $\tilde{O}$ as $\tilde{O}(\qubit{\tilde{x}})
= \sum_{s:|s|=|x|}\alpha_s\qubit{\tilde{s}}$ if $O(\qubit{x})$ is of the form
$\sum_{s:|s|=|x|} \alpha_s\qubit{s}$ for any string $x\in\{0,1\}^*$, where $\tilde{s}$ is a code of $s$, defined in Section \ref{sec:main-theorem}. Note that $\tilde{O}(\qubit{\tilde{x}}) = \widetilde{O(\qubit{x})}$ and $\ell(\tilde{O}(\qubit{\tilde{x}})) = 2\ell(O(\qubit{x}))+2$.

Firstly, we expand our initial functions
and construction rules given in Definition \ref{def:initial} by replacing any quantum function, say,  $f(\qubit{\phi})$ in each scheme of the definition with $f(\qubit{\phi},O)$.
Secondly, we introduce another initial quantum
function, called the {\em query function} $QUERY$.
For any two qustrings $\qubit{\phi}$ and $\qubit{\psi}$ of
length $n$, let $\qubit{\phi}\oplus\qubit{\psi} =
\sum_{s:|s|=n}\sum_{t:|t|=n}\measure{s}{\phi}\measure{t}{\psi}\qubit{s\oplus
t}$, where $s\oplus t$ means the \emph{bitwise XOR} of $s$ and $t$.
As a special case, it follows that
$\tilde{O}(\qubit{\tilde{x}})\oplus \tilde{O}(\qubit{\tilde{x}}) =
\qubit{0^{2|x|+2}}$ for any $x\in\{0,1\}^*$.  The query function is then  defined as
\begin{eqnarray*}
\lefteqn{QUERY(\qubit{\phi},O)} \hs{5} && \\
&=& \sum_{n\in[t]}
\sum_{x\in\Sigma^n}\sum_{s\in\Sigma^{n}}
\left(
\qubit{\tilde{x}}
\otimes (\tilde{O}(\qubit{\tilde{x}})
\oplus \qubit{\tilde{s}}) \otimes \measure{\tilde{x}\tilde{s}}{\phi} +
\sum_{y\in\Sigma^{4n+4}\wedge y\neq \tilde{x}\tilde{s}} \qubit{y}\!\measure{y}{\phi}
\right)
\end{eqnarray*}
for all $\qubit{\phi}\in\HH_{\infty}$, where $t=\floors{(\ell(\qubit{\phi})-4)/4}$. In particular,  $QUERY(\qubit{\tilde{x}}\qubit{\tilde{s}}\qubit{\phi})$ equals $\qubit{\tilde{x}}\otimes ( \tilde{O}(\qubit{\tilde{x}}\oplus \qubit{\tilde{s}}) \otimes \qubit{\phi}$.
It also follows that $QUERY\circ QUERY(\qubit{\phi},O) =I(\qubit{\phi})$ since
$\tilde{O}(\qubit{\tilde{x}})\oplus \tilde{O}(\qubit{\tilde{x}}) =
\qubit{0^{2|x|+2}}$.

Notice that if $O$ is in $\squareqp$ then the function
$Q_{O}(\qubit{\phi}) =_{def} QUERY(\qubit{\phi},O)$ also belongs to
$\squareqp$. It is also possible to show similar results discussed in the previous sections.
These basic results can open a door to a rich field of higher-type quantum computability and we expect fruitful results to be discovered in this new field.

\subsection{Application to Quantum Programming Languages}\label{sec:programming-language}

A practical application of our schematic definition can be found in the area of \emph{quantum programming languages}.
Since the early days of quantum computing research, a significant effort has been made physicists, computer scientists, and computer engineers to draw a pragmatic road map to a real-life quantum computer.

Toward the realization of such quantum computers, most research has focused on their hardware construction. For the building of ``multi-purpose'' quantum computers, however, it is more desirable to make them  ``programmable'' in such a way that a run of an appropriate  ``quantum program'' freely alters computing processes for different target problems without remodeling their hardware each time.
A quantum program here refers to
a finite series of instructions on how to operate the quantum computer step by step.
To write such a quantum program, nevertheless, we need to develop well-structured programming languages for the quantum computer (dubbed as \emph{quantum programming languages}).
Various quantum programming languages have been discussed over two decades in due course of developing real-life quantum computers. Refer to surveys, e.g., \cite{Gay06} for a necessary background.

Our schematic definition provides a description of how to define a given $\squareqp$-function. This description can be viewed as a set of instructions, each of which instructs how to apply each scheme to construct the desired quantum function and it thus resembles a program that dictates  how to construct the quantum function.
Therefore, our schematic description of a construction process of quantum functions may help us design appropriate quantum programming languages in the future.

\let\oldbibliography\thebibliography
\renewcommand{\thebibliography}[1]{%
  \oldbibliography{#1}%
  \setlength{\itemsep}{-2pt}%
}
\bibliographystyle{plain}

\end{document}